%% file: main.tex
\documentclass[12pt]{article}
\usepackage{amsfonts}
\usepackage{amsmath}
\usepackage{lipsum}
\usepackage{amssymb}
\usepackage{amsmath}
\usepackage{amsthm}
\usepackage{geometry}
\usepackage{multirow}
\usepackage{booktabs}
\usepackage{xr}

\usepackage[authoryear,comma]{natbib}
\renewcommand{\theequation}{\thesection.\arabic{equation}}
\newcommand{\myref}[2]{\hyperref[#1]{#2}}
\numberwithin{equation}{section}

\usepackage{latexsym}
\usepackage{graphicx} 
\usepackage{enumerate}
\usepackage{setspace}
\usepackage[toc,title,titletoc,header]{appendix}
\usepackage[bottom]{footmisc} 
\usepackage[pdfborder={0 0 0}]{hyperref}
\hypersetup{
 colorlinks=true,linkcolor=magenta, citecolor=blue, filecolor=magenta, 
 linkbordercolor={0 1 1}, citebordercolor={1 0 0},
}

\usepackage{booktabs}

\usepackage{lineno}
\setpagewiselinenumbers

\allowdisplaybreaks
\newtheorem{theorem}{Theorem}[section]
\newtheorem{lemma}{Lemma}[section]

\newtheorem{assumption}{Assumption}[section]

\newtheorem{example}{Example}[section]

\newcounter{assumptionM}
\newcounter{assumptionA}
\def\theassumptionM{M.\arabic{assumptionM}}
\def\theassumptionA{A.\arabic{assumptionA}}
\begin{document}
	\relax
	\hypersetup{pageanchor=false}
	\hypersetup{pageanchor=true}

\author{
\begin{tabular}{c c c}
    Federico A. Bugni & & Mengsi Gao \\
    Northwestern University & &UC Berkeley \\
    & & \\
    Filip Obradovi{\'c} & & Amilcar Velez \\
    Northwestern University & &Northwestern University \\
\end{tabular}
}

 
\title{\Large\vspace{-1.5cm} Identification and Inference on Treatment Effects under Covariate-Adaptive Randomization and Imperfect Compliance\thanks{Corresponding author: \href{mailto:federico.bugni@northwestern.edu}{federico.bugni@northwestern.edu}. We thank Eric Auerbach, Ivan Canay, Joel Horowitz, and seminar and conference participants where this paper was presented, for their helpful comments and discussions.}\vspace{0.5cm}
}

\maketitle

\vspace{-0.3in}
\thispagestyle{empty} 

\begin{spacing}{1.2}

\begin{abstract}
Randomized controlled trials (RCTs) frequently utilize covariate-adaptive randomization (CAR) (e.g., stratified block randomization) and commonly suffer from imperfect compliance. This paper studies the identification and inference for the average treatment effect (ATE) and the average treatment effect on the treated (ATT) in such RCTs with a binary treatment.

We first develop characterizations of the identified sets for both estimands. Since data are generally not i.i.d.\ under CAR, these characterizations do not follow from existing results. We then provide consistent estimators of the identified sets and asymptotically valid confidence intervals for the parameters. Our asymptotic analysis leads to concrete practical recommendations regarding how to estimate the treatment assignment probabilities that enter the estimated bounds. For the ATE bounds, using sample analog assignment frequencies is more efficient than relying on the true assignment probabilities. For the ATT bounds, the most efficient approach is to use the true assignment probability for the probabilities in the numerator and the sample analog for those in the denominator.

\end{abstract}
\end{spacing}



\thispagestyle{empty} 
\newpage

\input{maintext}

\renewcommand{\theequation}{\Alph{section}-\arabic{equation}}
\begin{appendix} 

\input{combined_appendix.tex}

\end{appendix}
\bibliography{BIBLIOGRAPHY}

\end{document}

%% file: maintext.tex
\section{Introduction}

This paper considers identification and inference in a randomized controlled trial (RCT) that features covariate-adaptive randomization (CAR) and imperfect compliance. CAR refers to randomization schemes that first stratify according to baseline covariates and then assign treatment status to achieve ``balance'' within each stratum (\cite{bugni/canay/shaikh:2018,bugni/canay/shaikh:2019,bugni/gao:2023}, among many others). These randomization schemes are widely utilized to assign treatment status in RCTs across various scientific disciplines; see \cite{rosenberger/lachin:2016} for a textbook treatment of this topic focused on clinical trials and \cite{duflo/glennerster/kremer:2007} and \cite{bruhn/mckenzie:2008} for reviews focused on development economics. CAR complicates the econometric analysis, as it typically leads to treatment assignments, treatment decisions, and outcomes that are not independent and identically distributed (i.i.d.). In particular, observations may be dependent, and their distributions may differ across units.

We allow for imperfect compliance in the sense that the RCT participants endogenously decide their treatment, which may or may not coincide with the treatment status assigned by the CAR mechanism. 
This is empirically relevant, as imperfect compliance is a common occurrence in many RCTs. 
In this context, \cite{bugni/gao:2023} study inference on the local average treatment effect (LATE). For more results on inference about the LATE under CAR, see \cite{ansel/hong/li:2018,bai/guo/shaikh/tabord-meehan:2024,jiang/linton/tang/zhang:2024}.

This paper studies the identification and inference for the average treatment effect (ATE) and the average treatment effect on the treated (ATT) in such RCTs with a binary treatment. We extend our analysis to the average treatment effect on the untreated (ATU) in the Appendix. These parameters are of substantial interest to researchers conducting RCTs. Under i.i.d.\ assumptions, the literature has shown that these parameters are partially identified under imperfect compliance; e.g., \cite{manski:1989,manski:1990,manski:1994,balke/pearl:1997,heckman:2001,huber:2017,kitagawa:2021}. Under CAR, however, the data are typically not i.i.d., which renders these results inapplicable.

Our paper makes several contributions. First, we provide the identified set of our estimands of interest under CAR assumptions. Since the data are generally not i.i.d., our identification analysis is based on the joint data-generating process (DGP) of the individual data and the treatment assignment mechanism. To the extent of our knowledge, this is the first characterization of these identified sets in the CAR literature with imperfect compliance. Second, we provide consistent estimators of the bounds of the identified set. Finally, we provide asymptotically valid confidence intervals for our partially identified parameters. Our asymptotic results are uniformly valid in a large class of probability distributions, which is necessary to obtain approximations that adequately represent the finite sample experiment of interest; see \cite{imbens/manski:2004,andrews/soares:2010}.

Our asymptotic analysis delivers concrete and practical recommendations on how to estimate the identified sets, particularly on how to estimate the treatment assignment probabilities. In the context of an RCT, the treatment assignment probabilities are often known, and so the question here is whether the researcher should estimate the identified sets using the true assignment probabilities or their sample analogs. As it turns out, our recommendations depend on the estimand under consideration. In the case of the ATE, using sample analog treatment assignment probabilities is (weakly) more efficient than using the true (limiting) treatment assignment probabilities, which translates into a more powerful inference about the ATE. In the case of the ATT (and the ATU), we recommend using the true (limiting) treatment assignment probabilities in the numerator and the sample analogs in the denominator. This approach yields more powerful inference for the ATT (and the ATU) than alternative bounds based on other configurations of treatment assignment probability estimators. To our knowledge, these asymptotic findings appear to be new in the context of analysis of RCTs with imperfect compliance (with or without CAR). In the context of perfect compliance (i.e., point identified estimands) and i.i.d.\ samples (i.e., without CAR), analogous results have been obtained in \cite{hahn:1998,hirano/imbens/ridder:2003}. 

The rest of the paper is organized as follows. Section \ref{sec:setup}, we describe the setup. Section \ref{sec:ate} studies the identification and inference for the ATE, and Section \ref{sec:att} does the same for the ATT. In Section \ref{sec:MC}, we explore the finite sample behavior of our methods via Monte Carlo simulations. Section \ref{sec:application} illustrates our results in an empirical application based on the RCT in \cite{dupas/karlan/robinson/ubfal:2018}. Section \ref{sec:conclusions} provides concluding remarks. All proofs and several intermediate results are collected in the appendix. For brevity, several auxiliary results have been placed in the Appendix.

\section{Setup}\label{sec:setup}

We consider an RCT with $n$ participants. For each participant $i=1,\dots,n$, $Y_i \in \mathbb{R}$ is the observed outcome of interest, $Z_i \in \mathcal{Z}$ is a vector of observed baseline covariates, $A_i \in \{0,1\}$ is the treatment assignment, and $D_i \in \{0,1\}$ is the treatment decision. 

We propose potential outcome models for outcomes and treatment decisions. We use $Y_i(D)$ to denote the potential outcome of participant $i$ if he/she makes treatment decision $D$, and we use $D_i(A)$ to denote the potential treatment decision of participant $i$ if he/she has assigned treatment $A$. These are related to their observed counterparts in the usual manner:
\begin{align}
 D_i &~=~ D_i(1) A_i + D_i(0) (1-A_i),\notag\\
 Y_i &~=~ Y_i(1) D_i + Y_i(0) (1-D_i) .\label{eq:outcome}
\end{align}

Following the usual classification in the LATE framework in \cite{angrist/imbens:1994}, each participant in the RCT can only be one of four types: complier, always taker, never taker, or a defier. An individual $i$ is said to be a complier if $\{D_i(0)=0,D_i(1)=1\}$, an always taker if $\{D_i(0)=D_i(1)=1\}$, a never taker if $\{D_i(0)=D_i(1)=0\}$, and a defier if $\{D_i(0)=1,D_i(1)=0\}$. As is common in the literature, we will later impose that there are no defiers in our population.

We use ${\bf Q}$ to denote the distribution of the underlying random variables, given by
\begin{align*}
 W^{(n)} ~=~ \{(Y_i(1),Y_i(0),D_i(1),D_i(0),Z_i)~:i=1,\dots,n\}.
\end{align*}
We use ${\bf G}$ to denote the treatment assignment distribution, i.e., the distribution of $\{ A^{( n) }|W^{(n)}\} $. Finally, we use ${\bf P}$ to denote the observed data distribution, given by
\begin{align*}
 X^{(n)} ~=~\{(Y_i,D_i,A_i,Z_i)~:i=1,\dots,n\}.
\end{align*}
Note that ${\bf P}$ is jointly determined by \eqref{eq:outcome} and the distributions ${\bf Q}$ and ${\bf G}$. We state our assumptions below in terms of restrictions on ${\bf Q}$ and ${\bf G}$. These distributions can vary with the sample size, but we keep this dependence implicit for the simplicity of exposition. Throughout the paper, we use $P$, $E$, and $V$ to denote probability, expectation, and variance, likewise keeping the underlying probability measures ${\bf P}$, ${\bf Q}$, or ${\bf G}$ implicit.

Strata are constructed from the observed baseline covariates $Z_i$ using a prespecified function $S:\mathcal{Z} \to \mathcal{S}$, where $\mathcal{Z}$ denotes the support of $Z_i$ and $\mathcal{S}$ is a finite set. For each participant $i=1,\dots,n$, let $S_{i} \equiv S(Z_{i})$ and let $S^{(n)} = \{S_{i}:i=1,\dots,n\}$. 
Our assumptions on ${\bf Q}$ are as follows.
\begin{assumption}{\rm \bf (Underlying distribution)}\label{ass:1}
For constants $Y_L$, $Y_H$, $\xi$, $W^{(n)}$ is an i.i.d.\ sample from ${\bf Q}$ that satisfies
\begin{enumerate}[(a)]
\item $Y_{i}(d) \in [Y_L,Y_H]$ for all $d\in \{0,1\}$, where $Y_L$ and $Y_H$ are known constants. Also, $p(s) \equiv P(S_i=s)\geq \xi >0$ for all $s \in \mathcal{S}$.
\item $ P(D_{i}( 0) =1,D_{i}(1) =0) =0$ or, equivalently, $ P( D_{i}(1)\geq D_{i}(0)) =1$.
\item $V[ Y_{i}( d) |D_{i}( a) =d,S_{i}=s] P(D_{i}( a) =d,S_{i}=s) \geq  \xi >0$ for some $((d,a),s)\in  \{(0,0),(1,1)\} \times \mathcal{S}$.
\item $P(D_i(1)=1,D_i(0)=0)\geq  \xi >0$ or, equivalently, $ P( D_{i}(1)> D_{i}(0)) \geq  \xi >0$.
 	\end{enumerate}
\end{assumption}

Assumption \ref{ass:1} says that the underlying data distribution is i.i.d., i.e., ${\bf Q}$ is the product of $n$ identical marginal distributions of $(Y_i(1),Y_i(0),D_i(1),D_i(0),Z_i)$. In addition, Assumption \ref{ass:1}(a) imposes that $Y_L$ and $Y_H$ are known logical bounds for the potential outcomes
and that all strata are relevant. In principle, one could entertain the case with $Y_L = -\infty$ or $Y_H=\infty$, but our later results reveal that either of these generates a non-informative, trivial identified set for our parameters of interest. Assumption \ref{ass:1}(b) corresponds to the ``no defiers'' or ``monotonicity'' condition that is standard in treatment effect analysis under imperfect compliance.
Assumption \ref{ass:1}(c) places a nontrivial bound on the conditional variance for at least one subgroup. This implies that our inference is non-degenerate, and it is arguably mild, as it is only required to hold for one $((d,a),s)\in  \{(0,0),(1,1)\} \times \mathcal{S}$. Finally, Assumption \ref{ass:1}(d) implies that the proportion of compliers in the population is non-trivial, ensuring that $P(D_i=1)\in(0,1)$, making ATT well-defined. It is not required for the ATE.

Under Assumption \ref{ass:1}, we introduce the following notation for population objects: for each $(d,a,s) \in \{0,1\}\times  \{0,1\}\times \mathcal{S}$,
\begin{align}
    \pi_{D(a)}(s) &~\equiv~ P(D_{i}( a) = 1|S_i=s)\notag\\
    \mu (d,a,s) &~\equiv~ \left\{\begin{array}{cc}
         E[ Y_{i}( d) |D_{i}( a) =d,S_{i}=s]&\text{ if }P(D_{i}( a) =d,S_{i}=s)>0 \\
         0&\text{ if }P(D_{i}( a) =d,S_{i}=s)=0 \\
    \end{array}\right.\notag\\
   \sigma^2 (d,a,s)&~\equiv~  \left\{\begin{array}{cc}
         V[ Y_{i}( d) |D_{i}( a) =d,S_{i}=s]&\text{ if }P(D_{i}( a) =d,S_{i}=s)>0 \\
         0&\text{ if }P(D_{i}( a) =d,S_{i}=s)=0. \\
    \end{array}\right.\label{eq:keyDefns}
\end{align}
We note that the expressions in \eqref{eq:keyDefns} are all features of the distribution ${\bf Q}$ but we omit this dependence for ease of notation. It is also convenient to introduce notation for sample objects: for each $s \in \mathcal{S}$,
\begin{align*}
    n(s) \equiv \sum\nolimits_{i=1}^n I[S_{i} =s ]~~\text{ and }~~
    n_{A}(s) \equiv \sum\nolimits_{i=1}^n I[A_{i} = 1, S_{i} =s ]. 
\end{align*}
With this notation in place, we can state our assumptions regarding the treatment assignment mechanism ${\bf G}$.
\begin{assumption}{\rm \bf (Assignment Mechanism)}\label{ass:2}
For a constant $\varepsilon>0$, ${\bf G}$ satisfies
\begin{enumerate}[{(a)}]
\item $W^{(n)}\perp A^{(n)}~|~S^{(n)}$.
\item $P(A_{i}=1 | S^{(n)}) \in (0,1)$ for all $i=1,\dots,n$.
\item For all $s \in \mathcal{S}$, $n_{A}(s)/n(s)  = \pi _{A}(s) + o_{p}(1)$, where $\pi _{A}(s) \in (\varepsilon,1-\varepsilon)$.
\item $P(A_{i}=1 | S_i ) = P(A_{j}=1 | S_j )$ for all $i,j=1,\dots,n$.
\item $\sqrt{n} E[|\pi_A(S_i) - P( A_i=1|S_i)|] = o(1)$ for all $i=1,\ldots ,n$.
\item $\{\{\sqrt{n}(n_{A}(s)/n(s)-\pi_A(s)):s\in  S\}|S^{(n)}\}\overset{d}{\to }N(\mathbf{0},{{\Sigma}}_{A})$ w.p.a.1, with ${{\Sigma}} _{A}=diag({\tau}(s)(1-\pi _{A}(s))\pi _{A}(s)/p(s):s\in \mathcal{S})$ and ${\tau}(s)\in [0,1]$.
\end{enumerate}
\end{assumption}

Assumption \ref{ass:2} represents the main departure relative to the i.i.d.\ setup, as it allows the treatment assignment vector $A^{(n)}$ to be non-independent and non-identically distributed. Assumption \ref{ass:2}(a) requires that the treatment assignment vector $A^{(n)}$ is a function of the strata vector $S^{(n)}$ and a randomization device that is conditionally independent of the underlying sample $W^{(n)}$. Assumption \ref{ass:2}(b) imposes that the conditional treatment assignment probability for any unit is between zero and one. This is a mild requirement that is satisfied by any of the treatment assignment mechanisms typically used in practice. 
Assumption \ref{ass:2}(c) imposes that the fraction of units assigned to treatment in the stratum $s$, $n_A(s)/n(s)$, converges in probability to {\it some limit} $\pi_A(s) \in (0,1)$. In practice, $\pi_A(s)$ is typically the strata-specific treatment assignment probability desired by the researcher; see Examples \ref{ex:SBR} and \ref{ex:SRS}. For this reason, we refer to $\{\pi_A(s):s\in \mathcal{S}\}$ as the {\it target probabilities}. Assumption \ref{ass:2}(d) requires that the probability of assignment of any unit conditional only on its own strata is independent of the identity of the unit. 
We note that $P(A_{i}=1 | S_i) =E[P(A_i=1|S^{(n)})|S_i]$, which is not necessarily equal to $\pi_A(S_i)$.\footnote{In the typical CAR method, the reason for the difference between $P(A_i=1|S^{(n)})$ and $\pi_A(S_i)$ is that the sample size of strata $s=S_i$ multiplied by $\pi_A(s)$ may not be an integer. In turn, this results in $P(A_{i}=1 | S_i)$ and $\pi_A(S_i)$ being different.} 
While $P(A_i=1|S_i)$ and $\pi_A(S_i)$ need not coincide, Assumption \ref{ass:2}(e) states that these need to converge to each other in expectation at a fast rate. Finally, Assumption \ref{ass:2}(f) strengthens Assumption \ref{ass:2}(c) to require that $\sqrt{n}(n_A(s)/n(s) - \pi_A(s))$ is asymptotically normal conditional on the strata information $S^{(n)}$. For each stratum $s \in \mathcal{S}$, the parameter $\tau(s) \in [0,1]$ determines the amount of dispersion that the CAR mechanism allows on the fraction of units assigned to the treatment. A lower value of $\tau(s)$ implies that the CAR mechanism imposes a higher degree of ``balance'' or ``control'' of the treatment assignment proportion relative to its desired target value. 
Assumption \ref{ass:2} is satisfied by a wide array of CAR schemes, including stratified block randomization (SBR) and simple random sampling (SRS), among others.\footnote{Additional examples of CAR methods include Efron's biased coin design (\cite{efron:1971}) and the so-called minimization methods (\cite{pocock/simon:1975,hu/hu:2012}). Under suitable conditions, these can also be shown to satisfy Assumption \ref{ass:2}. See \cite{bugni/canay/shaikh:2018,bugni/gao:2023} for details.}
SBR deserves special focus, as it is a frequently used CAR method in RCTs in economics.

\begin{example}[Stratified Block Randomization (SBR)]\label{ex:SBR}
This is sometimes also referred to as blocking, block randomization, or permuted blocks within strata. To implement this method, the researcher proposes a vector of desired assignment probabilities for each stratum, which we denote by $\{\pi_A(s):s \in \mathcal{S}\}$. Within every stratum $s \in \mathcal{S}$, SBR assigns exactly $\lfloor n(s)\pi_A(s)\rfloor$ of the $n(s)$ participants in stratum $s$ to treatment and the remaining $n(s) -\lfloor n(s)\pi_A(s)\rfloor $ to control, where all possible $$\binom{n(s)}{\lfloor n(s){\pi}_A(s)\rfloor}$$ assignments are equally likely. Then, $P(A_i = 1| S^{(n)}) ={\lfloor n(S_i)\pi_A(S_i)\rfloor}/{n(S_i)} $.

SBR can be shown to satisfy Assumption \ref{ass:2}. In particular, \cite{bugni/canay/shaikh:2018} show that Assumptions \ref{ass:2}(a)-(c) and (f) hold with $\tau(s)=0$ for all $s \in \mathcal{S}$.
To verify Assumption \ref{ass:2}(d), note that for all $s \in \mathcal{S}$ and $i=1,\dots,n$,
$$P(A_i = 1|S_i = s) ~=~ E[P(A_i=1|S^{(n)})|S_i = s] ~=~ E[{\lfloor n(s)\pi_A(s)\rfloor}/{n(s)}].$$
Finally, Assumption \ref{ass:2}(e) follows from this derivation:
\begin{align*}
\sqrt{n}E\left[ \left\vert \pi _{A}\left( S_{i}\right) -P( A_{i}=1|S_i) \right\vert \right]  &~=~\sqrt{n}E\left[ \left\vert \pi _{A}\left( S_{i}\right) -E[\lfloor n( S_{i}) \pi _{A}( S_{i}) \rfloor /n( S_{i}) | S_i] \right\vert \right]  \\
&~\overset{(1)}{\leq}~\sqrt{n}E\left[ 1\left[ n\left( S_{i}\right) \geq 1\right]/n\left( S_{i}\right) \right]  \\
&~=~\sqrt{n}\sum\nolimits_{s\in \mathcal{S}}E\left[ 1/(\vartheta \left( s\right) +1)\right] p\left( s\right)  \\
&~\overset{(2)}{=}~({1}/{\sqrt{n}})\sum\nolimits_{s\in \mathcal{S}}\left( 1-\left( 1-p\left( s\right) \right) ^{n}\right) ~\overset{(3)}{\to}~0,
\end{align*}
for $\vartheta \left( s\right) \sim Bi\left( n-1,p\left( s\right) \right) $ for all $s\in \mathcal{S}$, where (1) holds by $n\left( S_{i}\right) \pi _{A}\left( S_{i}\right) -1\leq \left\lfloor n\left( S_{i}\right) \pi _{A}\left( S_{i}\right) \right\rfloor$ $ \leq n\left( S_{i}\right) \pi _{A}\left( S_{i}\right) $, (2) by direct computation, and (3) by $\left\vert \mathcal{S}\right\vert <\infty $.
\end{example}

\begin{example}[Simple Random Sampling (SRS)]\label{ex:SRS}
This refers to a treatment assignment mechanism in which $A^{(n)}$ satisfies
\begin{equation}
P(A^{(n)}=(a_i: i=1,\dots,n)|S^{(n)},W^{(n)})~=~\prod_{i=1}^{n} \pi_A(S_i)^{a_i} (1-\pi_A(S_i))^{1-a_i}.
 \label{eq:SRS}
\end{equation}
In other words, SRS assigns each participant in stratum $s$ to treatment with probability $\pi_A(s)$ and to control with probability $(1-\pi_A(s))$, independent of the rest of the sample information. 

It is easy to see that SRS satisfies Assumption \ref{ass:2}. In particular, \cite{bugni/canay/shaikh:2018} show that Assumptions \ref{ass:2}(a)-(c) and (f) hold with $\tau(s)=1$ for all $s \in \mathcal{S}$. Finally, Assumptions \ref{ass:2}(d)-(e) hold by $P(A_i = 1|S^{(n)} ) =P(A_i = 1|S_i) = \pi_A(S_i)$ for all $i=1,\dots,n$.
\end{example}

\section{Average treatment effect (ATE)}\label{sec:ate}

This section studies the identification and inference of the ATE, defined as
\begin{equation}
   \theta
   ~\equiv~E[ Y_i( 1) -Y_i( 0) ] .
\end{equation}
By definition, $\theta$ is the expectation of the treatment effect when the treatment is mandated across the entire population.

We note that the ATE is solely determined by the underlying data distribution ${\bf Q}$, i.e., the treatment assignment mechanism ${\bf G}$ plays no role in determining $ \theta$.

\subsection{Identification}\label{sec:identicationATE}

The next result characterizes the identified set of the ATE.
\begin{theorem}\label{thm:main}
Under Assumptions \ref{ass:1}(a)-(c) and \ref{ass:2}(a)-(c), the identified set for the ATE is $\Theta _{I}({\bf P})= [\theta_{L}({\bf P}),\theta_{H}({\bf P})],$ where
\begin{align}
\theta _{L}({\bf P}) ~& = ~E\bigg[ \frac{(Y_{i}D_{i}+Y_{L}( 1-D_{i}) ) A_{i}}{P(A_{i}=1|S_{i}) }-\frac{( Y_{i}(1-D_{i}) +Y_{H}D_{i}) ( 1-A_{i}) }{1-P(A_{i}=1|S_{i}) }\bigg] , \notag\\
\theta _{H}({\bf P})~& = ~E\bigg[ \frac{(Y_{i}D_{i}+Y_{H}( 1-D_{i}) ) A_{i}}{P(A_{i}=1|S_{i}) }-\frac{( Y_{i}(1-D_{i}) +Y_{L}D_{i}) ( 1-A_{i}) }{1-P(A_{i}=1|S_{i})}\bigg] .\label{eq:bounds}
\end{align}
\end{theorem}

A few remarks are in order. As already mentioned, our setup allows for treatment assignment using CAR methods, which implies that the data may not be i.i.d. Hence, Theorem \ref{thm:main} does not follow from the existing results on the identification of the ATE under i.i.d.\ sampling, e.g., \cite{manski:1990,heckman:2001,huber:2017,kitagawa:2021}. In fact, to our knowledge, Theorem \ref{thm:main} is the first characterization of the identified set of the ATE under the CAR framework with imperfect compliance.

Theorem \ref{thm:main} shows that our bounds coincide with the sharp bounds on the ATE originally derived by \cite{manski:1990} under the monotonicity condition stated in Assumption \ref{ass:1}(b) and when $\mathcal{S}$ is a singleton. This implies that the departure from the i.i.d.\ conditions allowed by the treatment assignment mechanism under CAR does not affect the structure of the sharp bounds.
While this might suggest that our proof is merely a by-product of existing results, this is not the case. Prior work focuses on the i.i.d.\ setting and therefore relies on identifying information contained in the marginal distribution of a single individual to construct the bounds. In contrast, our approach works with the joint distribution of all $n$ individuals in the sample, allowing us to account for the dependence and heterogeneity permitted by the treatment assignment mechanisms under our CAR framework.





It is worth highlighting that the equalities in \eqref{eq:bounds} hold for any $i=1,\dots,n$. That is, the identified set for the ATE is the same for all of the individuals in our sample. This observation would be straightforward in the context of an i.i.d.\ sample, but it requires proof in the CAR framework, where data need not be i.i.d.
It is also notable that Theorem \ref{thm:main} would not change if we restrict attention to Assumptions \ref{ass:1}(a)-(b) and \ref{ass:2}(a)-(b). In other words, Assumptions \ref{ass:1}(c) and \ref{ass:2}(c) do not provide any additional identifying power for the ATE, though they will become relevant when conducting inference in the next section.\footnote{This also connects to our earlier point: without Assumption \ref{ass:2}(d), the probability $P(A_i = 1|S_i)$ may vary across individuals, and yet, under our assumptions, the identified set for the ATE does not vary with $i$.}


Finally, the proof of Theorem \ref{thm:main} presented in Appendix \ref{sec:appendix1ATE} reveals that $\theta _{L}({\bf P})$ and $\theta _{H}({\bf P})$ only depend on the underlying data distribution ${\bf Q}$. That is, the treatment assignment mechanism ${\bf G}$ plays no role in determining the sharp bounds on the ATE. 

\subsection{Inference}\label{sec:inferenceATE}

Our goal in this section is to estimate the identified set for the ATE and provide a confidence interval for its true value. Following Theorem \ref{thm:main}, we propose the following estimators of the sharp bounds in \eqref{eq:bounds},
\begin{align}
\hat{\theta}_{L}~& \equiv ~\frac{1}{n}\sum_{i=1}^{n}\Big[ \frac{( Y_{i}D_{i}+Y_{L}( 1-D_{i}) ) A_{i}}{n_{A}( S_{i}) /n( S_{i}) }-\frac{( Y_{i}( 1-D_{i}) +Y_{H}D_{i}) ( 1-A_{i}) }{1-n_{A}( S_{i}) /n( S_{i}) }\Big] , \notag\\ 
\hat{\theta}_{H}~& \equiv ~\frac{1}{n}\sum_{i=1}^{n}\Big[ \frac{( Y_{i}D_{i}+Y_{H}( 1-D_{i}) ) A_{i}}{n_{A}( S_{i}) /n( S_{i}) }-\frac{( Y_{i}( 1-D_{i}) +Y_{L}D_{i}) ( 1-A_{i}) }{1-n_{A}( S_{i}) /n( S_{i}) }\Big] .\label{eq:BoundsHat}
\end{align}
Note that \eqref{eq:BoundsHat} is the sample analog of \eqref{eq:bounds}, where the treatment assignment probabilities $\{P(A_{i}=1|S_{i}):i=1,\dots,n\} $ are estimated by the sample treatment frequencies $\{n_{A}( S_{i}) /n( S_{i}):i=1,\dots,n \}$.\footnote{The definitions in \eqref{eq:BoundsHat} require that $n_{A}(S_{i}) /n( S_{i}) \in (0,1)$ for all $i=1,\dots,n$. Lemma \ref{lem:den_not_zero} in the appendix shows that this occurs with probability approaching one (w.p.a.1).} An alternative bounds estimator could be obtained by estimating the treatment assignment probabilities with their target values $\{\pi_{A}( S_{i}) :i=1,\dots,n \}$. In Section \ref{sec:DiscussionATE}, we argue that the bounds estimator in \eqref{eq:BoundsHat} should be preferred, as it generates inference for the ATE that is relatively more robust, efficient, and powerful (all in a sense made precise in that section).




We now derive the asymptotic properties of the bounds estimator in \eqref{eq:BoundsHat}. The literature on inference in partially identified models convincingly argues that uniform asymptotic results are necessary to obtain approximations that adequately represent the finite-sample experiment of interest; e.g., see \cite{imbens/manski:2004,andrews/soares:2010}. 
With this motivation in mind, we establish the uniform asymptotic distribution of our estimated bounds. Let $\mathcal{P}_1 =\mathcal{P}(Y_L,Y_H,\xi,\varepsilon) $ denote the set of probabilities ${\bf P}$ generated by $( {\bf  Q},{\bf  G})$ that satisfy Assumptions \ref{ass:1}(a)-(c) and \ref{ass:2}(a)-(c),(e). Theorem \ref{thm:AsyDist_ATE} in the appendix shows that
\begin{equation}
\sqrt{n}~\left(
\begin{array}{c}
\hat{\theta}_{L}-\theta_{L}({\bf P})\\
\hat{\theta}_{H}-\theta_{H}({\bf P})
\end{array}
\right)~\overset{d}{\to}~N( {\bf 0}_{2} , \Sigma_{\theta}({\bf P}))\text{~~~uniformly in ${\bf P} \in \mathcal{P}_1$}.
\label{eq:asyDistBoundsMainText}
\end{equation}
Two aspects of \eqref{eq:asyDistBoundsMainText} are worth highlighting. First, the asymptotic variance $\Sigma_{\theta}({\bf P})$ accounts for the use of CAR in treatment assignment and typically differs from the asymptotic variance obtained under SRS. Second, the uniformity of our results contrasts with the typical asymptotic approximations in the CAR literature, which are pointwise in nature. As discussed, this aspect of our derivations is crucial in the context of partial identification.

As a corollary of \eqref{eq:asyDistBoundsMainText}, we now establish that $\hat{\Theta} _{I} \equiv [\hat{\theta}_{L},\hat{\theta}_{H}]$ is a uniformly consistent estimator of the identified set for the ATE.

\begin{theorem}\label{thm:consist}
For any $\varepsilon >0$, 
\begin{equation}
\underset{n\to \infty }{\lim}~\inf_{{\bf P} \in \mathcal{P}_1 }~{\bf P}(~d_{H}(\hat{\Theta}_{I},\Theta _{I}({\bf P}))\leq \varepsilon ~)~=~1,  \label{eq:consistency}
\end{equation}
where $d_H(U,V)$ denotes the Hausdorff distance between sets $U,V \subset \mathbb{R}$.
\end{theorem}


Our next result in this section proposes a CI for the ATE under our assumptions. To this end, Theorem \ref{thm:AsyVarEstimation} provides a uniformly consistent estimator of the asymptotic variance $\Sigma_{\theta}({\bf P})$ in \eqref{eq:asyDistBoundsMainText}, given by
\begin{equation}
    \hat{\Sigma}_{\theta} ~\equiv~ \bigg( 
\begin{array}{cc}
\hat\sigma _{L}^{2} &\hat \sigma _{HL} \\ 
\hat\sigma _{HL} & \hat\sigma _{H}^{2}
\end{array}
 \bigg),
 \label{eq:Sigma_Hat}
\end{equation}
where $(\hat\sigma _{L}^{2} ,\hat\sigma _{H}^{2}, \hat\sigma _{HL})$ are defined in \eqref{eq:AsyDistATE_var1}. Following our earlier comments, we note that these estimators account for the use of CAR in treatment assignment.  With these estimators in place, we propose using the CI for the ATE given by $CI_{\alpha }^{2}$ in \cite{stoye:2009}, i.e.,
\begin{equation}
    \hat{C}_{\theta}\left( 1-\alpha \right) ~\equiv~\Big[~ \hat{\theta}_{L}-\frac{\hat{\sigma}_{L} \hat{c}_{L}}{\sqrt{n}}~,~\hat{\theta}_{H}+\frac{\hat{\sigma}_{H} \hat{c}_{H}}{\sqrt{n}}~\Big] ,
\label{eq:CSstoye}
\end{equation}
where $( \hat{c}_{L},\hat{c}_{H}) $ are the minimizers of $\hat{\sigma}_{L}c_{L}+\hat{\sigma}_{H}c_{H}$ subject to
\begin{align}
P\Big( -c_{L}\leq Z_{1}~\cap~\tfrac{\hat{\sigma}_{HL}}{\hat{\sigma}_{H}\hat{\sigma}_{L}}Z_{1}\leq c_{H}+\tfrac{\sqrt{n}( \hat{\theta}_{H}-\hat{\theta}_{L}) }{\hat{\sigma}_{H}}+Z_{2}\big({1-\tfrac{\hat{\sigma}_{HL}^2}{\hat{\sigma}_{H}^2\hat{\sigma}_{L}^2}} \big)^{1/2}\Big| X^{(n)} \Big)  
&\geq 1-\alpha  \notag\\
P\Big( -c_{L}-\tfrac{\sqrt{n}( \hat{\theta}_{H}-\hat{\theta}_{L}) }{\hat{\sigma}_{L}}-Z_{2}\big({1-\tfrac{\hat{\sigma}_{HL}^2}{\hat{\sigma}_{H}^2\hat{\sigma}_{L}^2}} \big)^{1/2}
\leq \tfrac{\hat{\sigma}_{HL}}{\hat{\sigma}_{H}\hat{\sigma}_{L}}Z_{1}~\cap~Z_{1}\leq c_{H} \Big| X^{(n)} \Big)  
&\geq 1-\alpha ,\label{eq:constrainsStoye}
\end{align}
with $(Z_{1},Z_{2}) \sim N({\bf 0}_2,{\bf I}_2)$. Under our assumptions, \cite{stoye:2009} shows that $CI_{\alpha }^{2}$ is the shortest confidence interval with the correct asymptotic nominal size.
The next result establishes that the CI in \eqref{eq:CSstoye} is asymptotically uniformly valid and exact.
\begin{theorem}\label{thm:Stoye}
The CI in \eqref{eq:CSstoye} satisfies
\begin{equation}
    \underset{n\to \infty }{\lim}~\inf_{{\bf P}\in \mathcal{P}_1 }~\inf_{\theta \in \Theta_{I}({\bf P})}~{\bf P}( \theta \in \hat{C}_{\theta}( 1-\alpha )) ~=~ 1-\alpha ,\label{eq:coverage}
\end{equation}
where $\Theta_{I}({\bf P})$ is as in \eqref{eq:bounds}.
\end{theorem}



The proof of Theorem \ref{thm:Stoye} follows from applying \citet[Proposition 1]{stoye:2009} to our CAR setup. The main ingredients for this result are the uniform convergence in distribution of our estimated bounds (Theorem \ref{thm:AsyDist_ATE}) and the uniform consistency of our estimators of the asymptotic variance (Theorem \ref{thm:AsyVarEstimation}). As previously emphasized, the uniform convergence results are novel in the CAR framework, yet essential for inference on partially identified parameters such as the ATE.

\subsection{Discussion}\label{sec:DiscussionATE}

As already mentioned, the bounds in \eqref{eq:BoundsHat} estimate the treatment assignment probabilities $\{P(A_{i}=1|S_{i}):i=1,\dots,n\} $ in \eqref{eq:bounds} with the sample treatment frequencies $\{n_{A}( S_{i}) /n( S_{i}):i=1,\dots,n \}$. In an RCT, it is not uncommon to know the target probabilities $\{\pi_{A}(s):s \in \mathcal{S} \}$. In principle, this information would also allow us to construct {\it alternative estimated bounds} that estimate the treatment assignment probabilities in \eqref{eq:bounds} with 
their target probabilities $\{\pi_{A}( S_i):i=1,\dots,n \}$. 
This section explains that the estimated bounds in \eqref{eq:BoundsHat} are demonstrably better than either of these alternative estimated bounds in several dimensions.

The first step to analyzing the limiting properties of the alternative estimated bounds is to derive their uniform asymptotic distribution. 
Under Assumptions \ref{ass:1}(a)-(c) and \ref{ass:2}(a)-(c),(e)-(f), Theorem \ref{thm:AsyDist_ATEother} in the Appendix derives the uniform asymptotic distribution of the alternative estimated bounds. This result reveals that the alternative estimated bounds have a uniform asymptotic distribution as in \eqref{eq:asyDistBoundsMainText}, but with ${\Sigma}_{\theta}({\bf P})$ replaced by a (weakly) larger asymptotic variance $\tilde{\Sigma}_{\theta}({\bf P})$ in the sense that $\tilde{\Sigma}_{\theta}({\bf P}) - {\Sigma}_{\theta}({\bf P})$ is positive semidefinite. In particular, $\tilde{\Sigma}_{\theta}({\bf P})$ depends on the parameter $\{\tau(s):s \in \mathcal{S}\}$ in Assumption \ref{ass:2}(f) characterizing the amount of dispersion of the CAR mechanism.\footnote{In fact, $\tilde{\Sigma}_{\theta}({\bf P}) - {\Sigma}_{\theta}({\bf P})$ can be positive definite if $\tau(s)>0$ for some $s\in \mathcal{S}$.} This result implies that the estimated bounds in \eqref{eq:BoundsHat} are better than the alternative estimated bounds in three ways.

First, the asymptotic analysis of the estimated bounds in \eqref{eq:BoundsHat} requires fewer assumptions than that of the alternative estimated bounds. In particular, the analysis of the former is established uniformly for $P \in \mathcal{P}_1$, which does not require Assumption \ref{ass:2}(f). This implies that the asymptotic behavior of the estimated bounds in \eqref{eq:BoundsHat} is robust to the details of the CAR mechanism. On the other hand, the asymptotic analysis of the alternative estimated bounds requires Assumption \ref{ass:2}(f), and so they do not enjoy this robustness property.

Second, the estimated bounds in \eqref{eq:BoundsHat} are (weakly) more efficient than the alternative estimated bounds. This follows from the fact that the estimated bounds in \eqref{eq:BoundsHat} are asymptotically distributed according to $N({\bf 0}_2,{\Sigma}_{\theta}({\bf P}))$, the alternative estimated bounds are asymptotically distributed according to $N({\bf 0}_2,\tilde{\Sigma}_{\theta}({\bf P}))$, and $\tilde{\Sigma}_{\theta}({\bf P}) - {\Sigma}_{\theta}({\bf P})$ is positive semidefinite. That is, estimating treatment probability using sample treatment frequencies is more efficient than plugging in the target probabilities. This finding resembles the semiparametric efficiency results for point-identified ATE estimator in \cite{hahn:1998} and \cite{hirano/imbens/ridder:2003}, but is novel in the context of partially identified inference for the ATE under imperfect compliance. 

Third, inference based on the estimated bounds in \eqref{eq:BoundsHat} is more powerful than the inference based on the alternative estimated bounds. This finding is based on results developed in our related work in \cite{bugni/gao/obradovic/velez:2024b}. In that paper, we study the asymptotic power properties of \cite{stoye:2009}'s CIs for sequences of local alternative hypotheses and show that asymptotically normal bounds estimators with higher variance-covariance matrices (in a positive semidefinite sense) have (weakly) lower rejection rates.\footnote{This result is not obvious in \cite{stoye:2009}'s framework, as the CI based on less efficient estimated bounds does not necessarily include the CI based on the more efficient estimated bounds.} In combination with the efficiency results described in the last paragraph, this allows us to conclude that for all sequences of local alternative hypotheses, the limiting rejection rate of the CI in \eqref{eq:CSstoye} is (weakly) larger than that of the CI related to the alternative estimated bounds.

\section{Average treatment effect on the treated (ATT)}\label{sec:att}

This section studies identification and inference on the ATT, given by
\begin{align}
\upsilon
~\equiv ~E[ Y_i(1)-Y_i(0)|D_i=1 ]. \label{eq:ATT_defn}
\end{align}
By definition, $\upsilon$ represents the treatment effect for the sub-population that (endogenously) adopted the treatment. This sub-population is composed of compliers assigned to treatment and always takers assigned to control.

By definition, the ATT is the average treatment effect conditional on the decision $D_i=1$. Since $D_i = D_i(A_i)$, this reveals that the ATT depends on both ${\bf Q}$ and ${\bf G}$. This contrasts with the ATE, which only depends on ${\bf Q}$. 

There are a few subtleties to the definition of ATT in the context of CAR. First, if we allow CAR mechanisms that determine arbitrary treatment assignment probabilities $\{ P( A_i=1|S^{( n) })\}_{i=1}^{n}$, it is possible for the ATT to vary with $i=1,\dots,n$. Assumption \ref{ass:2}(d) avoids this possibility by assuming that $P( A_i=1|S_i)$ does not depend on $i$. Second, the definition of the ATT requires $P(D_{i}=1)>0$ for all $i=1,\dots,n$. This result can be shown based on Assumptions \ref{ass:1}(d) and \ref{ass:2}(b).

\subsection{Identification}

The following result characterizes the identified set of the ATT.

\begin{theorem}\label{thm:mainATT} 
Under Assumptions \ref{ass:1} and \ref{ass:2}(a)-(d), the identified set for the ATT is $\Upsilon _{I}({\bf P})= [\upsilon_{L}({\bf P}),\upsilon_{H}({\bf P})],$ where 
{\small
\begin{align} 
\upsilon _{L}( {\bf P}) 
& \equiv \frac{1}{G} E\Bigg[  \left( \frac{Y_i A_i}{P(A_i =1 | S_i)} - \frac{Y_i(1-A_i)}{1- P(A_i = 1 | S_i)} \right) P(A_i =1 | S_i) +  \frac{(Y_i-Y_H)D_i(1-A_i)}{1 - P(A_i = 1 | S_i)}  \Bigg],  \notag \\ 
\upsilon_{H}( {\bf P})
& \equiv \frac{1}{G} E\Bigg[  \left( \frac{Y_i A_i}{P(A_i =1 | S_i)} - \frac{Y_i(1-A_i)}{ 1 - P(A_i = 1 | S_i)} \right) P(A_i =1 | S_i) +  \frac{(Y_i-Y_L)D_i(1-A_i)}{1 - P(A_i = 1 | S_i)}  \Bigg],  \label{eq:bounds_ATT}
\end{align}
}
and
{\small
\begin{align} 
G = E \left[ \left( \frac{D_i A_i}{ P(A_i =1 | S_i) } - \frac{D_i (1-A_i)}{1 - P(A_i = 1 | S_i) } \right)  P(A_i =1 | S_i)  + \frac{D_i (1-A_i)}{1 - P(A_i = 1 | S_i)}  \right]. 
\label{eq:bounds_ATT_G}
\end{align}
}
\end{theorem} 

At this point, we reiterate several remarks made after Theorem \ref{thm:main}. 
Theorem \ref{thm:mainATT} accommodates treatment assignment via CAR methods, implying that the data need not be i.i.d. Consequently, Theorem \ref{thm:mainATT} does not follow from existing results on the identification of the ATT under i.i.d.\ sampling, e.g., \cite{heckman:2000,huber:2017}. In fact, Theorem \ref{thm:mainATT} is the first characterization of the identified set of the ATT under the CAR framework with imperfect compliance.

When data are indeed i.i.d., the bounds in Theorem \ref{thm:mainATT} coincide with the sharp bounds on the ATT. As in the case of the ATE, this shows that the relaxation of i.i.d.\ sampling permitted by CAR treatment assignment does not alter the structure of the sharp bounds. We emphasize that this result does not follow from existing literature, which derives such bounds exclusively under i.i.d.\ assumptions. Accommodating CAR requires a distinct proof based on the joint data distribution.

The equalities in \eqref{eq:bounds_ATT} hold for any $i=1,\dots,n$. As with the ATE, all individuals in the sample have the same the identified set for the ATT, and establishing this requires proof when the data need not be i.i.d. We also note that Theorem \ref{thm:mainATT} remains unchanged if we restrict attention to Assumptions \ref{ass:1}(a)–(b),(d) and \ref{ass:2}(a)–(b),(d). In other words, Assumptions \ref{ass:1}(c) and \ref{ass:2}(c),(e) do not contribute to the identification of the ATT, but play a role when conducting inference in the next section. Compared to the result for the ATE, the sharp bounds for the ATT additionally require Assumptions \ref{ass:1}(d) and \ref{ass:2}(d). Assumption \ref{ass:1}(d) ensures the presence of compliers, which guarantees that $G$ in \eqref{eq:bounds_ATT_G} is positive. In turn, we require Assumption \ref{ass:2}(d) to obtain that all individuals share a common ATT.

Finally, the proof of Theorem \ref{thm:mainATT} presented in Appendix \ref{sec:appendix1ATT} reveals that $\upsilon _{L}( {\bf P}) $ and $\upsilon _{H}( {\bf P}) $ depend on both ${\bf Q}$ and ${\bf G}$. In contrast to the ATE, the treatment assignment mechanism ${\bf G}$ plays a role in determining the sharp bounds on the ATT. 


\subsection{Inference}\label{sec:inferenceATT}

We now turn to constructing an estimator of the identified set for the ATT and a confidence interval for its true value. The bounds in \eqref{eq:bounds_ATT} motivate the following:
\begin{align}
    \hat{\upsilon}_{L} ~&\equiv~ \frac{1}{\hat{G}n } \sum_{i=1}^n \left[ \left( \tfrac{Y_{i}A_{i}}{n_{A}\left( S_{i}\right) /n\left( S_{i}\right)} - \tfrac{Y_{i} \left(1-A_{i} \right)}{1- n_{A}\left( S_{i}\right) /n\left( S_{i}\right)}  \right) {\pi}_{A}\left( S_{i} \right)  + \tfrac{\left(Y_{i}-Y_{H}\right)D_{i}\left(1-A_{i}\right) }{1-n_{A}\left( S_{i}\right) /n\left( S_{i}\right)} \right]     \notag\\
    \hat{\upsilon}_{H} ~&\equiv~ \frac{1}{\hat{G}n } 
    \sum_{i=1}^n \left[ \left( \tfrac{Y_{i}A_{i}}{n_{A}\left( S_{i}\right) /n\left( S_{i}\right)} - \tfrac{Y_{i} \left(1-A_{i} \right)}{1- n_{A}\left( S_{i}\right) /n\left( S_{i}\right)}  \right){\pi}_{A}\left( S_{i} \right) 
    + \tfrac{\left(Y_{i}-Y_{L}\right)D_{i}\left(1-A_{i}\right) }{1-n_{A}\left( S_{i}\right) /n\left( S_{i}\right)} \right],
    \label{eq:BoundsHat_ATT}
\end{align}
and 
\begin{equation}
    \hat{G}~\equiv~ \frac{1}{n} \sum_{i=1}^n \left[ \left(\tfrac{D_{i}A_{i}}{n_{A}\left( S_{i}\right) /n\left( S_{i}\right)} - \tfrac{D_{i}\left(1-A_{i}\right)}{1-n_{A}\left( S_{i}\right) /n\left( S_{i}\right)}\right) {\pi}_{A}\left( S_{i} \right) + \tfrac{D_{i}\left(1-A_{i}\right)}{1-n_{A}\left( S_{i}\right) /n\left( S_{i}\right)} \right].
   \label{eq:G_hat}
\end{equation}
It is clear that the expression in \eqref{eq:BoundsHat_ATT} are the sample analog of those in \eqref{eq:bounds_ATT}, where the treatment assignment probabilities $\{P(A_{i}=1|S_i):i=1,\dots,n\} $ that appear in the denominator are replaced by the sample treatment frequencies $\{n_{A}( S_{i})/n(S_i):i=1,\dots,n \}$ and the ones in the numerator are replaced by $\{\pi_{A}( S_{i}):i=1,\dots,n \}$. This construction requires the researcher to know $\{\pi_{A}( s):s \in \mathcal{S}\}$, which is reasonable in the context of an RCT. If these were unknown, one would replace them with $\{n_{A}( S_{i})/n(S_i):i=1,\dots,n \}$ but, as we explain in Section \ref{sec:DiscussionATT}, this leads to inference for the ATT that is relatively less robust, efficient, and powerful (in a sense made precise in that section).

We now derive the uniform asymptotic properties of the bounds estimator in \eqref{eq:bounds_ATT}.
To this end, let $\mathcal{P}_2=\mathcal{P}(Y_{L},Y_{H},\xi ,\varepsilon )$ denote the set of probabilities $\mathbf{P}$ generated by $(\mathbf{Q},\mathbf{G})$ that satisfy Assumptions \ref{ass:1} and \ref{ass:2}(a)-(e).\footnote{Since $\mathcal{P}_2 \subset \mathcal{P}_1$, our analysis for ATT requires stronger assumptions than that of the ATE.} Theorem \ref{thm:AsyDist_ATT} shows that
\begin{equation}
\sqrt{n}\left(
\begin{array}{c}
\hat{\upsilon}_{L}-\upsilon _{L}(\mathbf{P})\\\hat{\upsilon}_{H}-\upsilon _{H}(\mathbf{P})
\end{array}
\right)~\overset{d}{\to }~N(\mathbf{0}_{2},\Sigma _{\upsilon }(\mathbf{P}))\text{~~~uniformly in $\mathbf{P}\in \mathcal{P}_2$.}
\label{eq:asydist_ATT_main_text}
\end{equation}
As in the previous section, we note that $\Sigma_{\upsilon}(\mathbf{P})$ accounts for the use of CAR in treatment assignment and generally does not coincide with the asymptotic variance under i.i.d.\ sampling.
Based on \eqref{eq:asydist_ATT_main_text}, we conclude that $\hat{\Upsilon}_{I}\equiv [\hat{\upsilon}_{L},\hat{\upsilon}_{H}]$ is a uniformly consistent estimate of the identified set of the ATT. 

\begin{theorem}\label{thm:consist_ATT}
For any $\varepsilon >0$, 
\begin{equation}
\underset{n\to \infty }{\lim}~\inf_{{\bf P} \in \mathcal{P}_2 }~{\bf P}(~d_{H}(\hat{\Upsilon}_{I},\Upsilon _{I}(\mathbf{P}))\leq \varepsilon
~)~=~1.  \label{eq:consist_ATT}
\end{equation}
where $d_H(A,B)$ denotes the Hausdorff distance between sets $A, B\subset \mathbb{R}$.
\end{theorem}

Next, we propose a CI for the ATT. For this purpose, Theorem \ref{thm:AsyVarEstimation_ATT} proposes a uniformly consistent estimator for the variance $\Sigma _{\upsilon }(\mathbf{P})$ in \eqref{eq:asydist_ATT_main_text}, given by
\begin{equation}
\hat{\Sigma}_{\upsilon } ~=~    \frac{1}{\hat{G}^{2}} \left( 
\begin{array}{cc}
\hat\varpi _{L}^{2} & \hat\varpi _{HL} \\ 
\hat\varpi _{HL}  & \hat\varpi _{H}^{2}
\end{array}
\right),
\end{equation}
where $\left(\hat\varpi _{L},\hat\varpi _{H},\hat\varpi _{HL}\right) $ are as in \eqref{eq:AsyVar_ATT} and $\hat{G}$ is as in \eqref{eq:G_hat}. We can then construct a CI for the ATT as in \eqref{eq:CSstoye} but with $( \hat{\theta}_{L},\hat{\theta}_{H},\hat{\Sigma}_{\theta  }) $ replaced by $( \hat{\upsilon}_{L},\hat{\upsilon}_{H},\hat{\Sigma}_{\upsilon }) $, i.e.,
\begin{equation}
    \hat{C}_{\upsilon}\left( 1-\alpha \right) ~\equiv~\Big[~ \hat{\upsilon}_{L}-\frac{\hat\varpi _{L} \hat{c}_{L}}{\hat{G}\sqrt{n}}~,~\hat{\upsilon}_{H}+\frac{\hat\varpi _{H} \hat{c}_{H}}{\hat{G}\sqrt{n}}~\Big] .
\label{eq:CSstoye_ATT}
\end{equation}
where $(\hat{c}_L,\hat{c}_H)$ are the minimizers of $\hat\varpi _{L} c_{L} + \hat\varpi _{H} c_{H}$ subject to
\begin{align*}
P\Big( -c_{L}\leq Z_{1}~\cap~\tfrac{\hat{\varpi}_{HL}}{\hat{\varpi}_{H}\hat{\varpi}_{L}}Z_{1}\leq c_{H}+\tfrac{\sqrt{n}( \hat{\theta}_{H}-\hat{\theta}_{L}) }{\hat{\varpi}_{H}/\hat{G}}+Z_{2}\big({1-\tfrac{\hat{\varpi}_{HL}^2}{\hat{\varpi}_{H}^2\hat{\varpi}_{L}^2}} \big)^{1/2}~\Big|~ X^{(n)} \Big)  
&\geq 1-\alpha  \notag\\
P\Big( -c_{L}-\tfrac{\sqrt{n}( \hat{\theta}_{H}-\hat{\theta}_{L}) }{\hat{\varpi}_{L}/\hat{G}}-Z_{2}\big({1-\tfrac{\hat{\varpi}_{HL}^2}{\hat{\varpi}_{H}^2\hat{\varpi}_{L}^2}} \big)^{1/2}
\leq \tfrac{\hat{\varpi}_{HL}}{\hat{\varpi}_{H}\hat{\varpi}_{L}}Z_{1}~\cap~Z_{1}\leq c_{H} ~\Big|~ X^{(n)} \Big) 
&\geq 1-\alpha ,
\end{align*}%
and $Z_{1},Z_{2}$ are i.i.d.\ $N( 0,1) $. The next result establishes that this CI is asymptotically uniformly valid.

\begin{theorem}\label{thm:StoyeATT} 
The CI in \eqref{eq:CSstoye_ATT} satisfies
\begin{equation}
\underset{n\to \infty }{\lim}~\inf_{\mathbf{P}\in \mathcal{P}_2}~\inf_{\upsilon \in \Upsilon _{I}(\mathbf{P})}~{\bf P}(\upsilon \in \hat{C}_{\upsilon}(1-\alpha ))~=~1-\alpha ,  \label{eq:coverageATT}
\end{equation}
where $\Upsilon _{I}(\mathbf{P})$ is as in \eqref{eq:bounds_ATT}.
\end{theorem}

The argument for this result is analogous to that of the ATE. The main ingredients for this proof are the uniform convergence in distribution of the ATT bounds (Theorem \ref{thm:AsyDist_ATT}) and the uniform consistency of our estimators of the asymptotic variance (Theorem \ref{thm:AsyVarEstimation_ATT}). We reiterate that the uniformity of our results is crucial for partially identified analysis and is new in the context of CAR.

\subsection{Discussion}\label{sec:DiscussionATT}

As previously pointed out, the bounds in \eqref{eq:BoundsHat_ATT} are constructed by replacing the treatment assignment probabilities $\{P(A_{i}=1|S_i):i=1,\dots,n\}$ in the denominator of \eqref{eq:bounds_ATT} with the sample treatment frequencies $\{n_{A}(S_{i})/n(S_i):i=1,\dots,n\}$, and those in the numerator of \eqref{eq:bounds_ATT} with $\{\pi_{A}(S_{i}):i=1,\dots,n\}$. In principle, one could consider alternative bounds that substitute either of these probabilities with different estimators. Our results show that all of these alternative bounds lead to inference for the ATT that is worse relative to that based on the bounds in \eqref{eq:BoundsHat_ATT}. For brevity, we focus on the alternative bounds that estimate both the probabilities in the numerator and denominator of \eqref{eq:bounds_ATT} using the sample treatment frequencies. Compared to the bounds in \eqref{eq:BoundsHat_ATT}, these alternative bounds have the advantage of not requiring knowledge of $\{P(A_{i}=1|S_i):i=1,\dots,n\}$ or $\{\pi_{A}(S_{i}):i=1,\dots,n\}$. However, we now argue that using the alternative bounds can be costly in terms of efficiency.

Under Assumptions \ref{ass:1} and \ref{ass:2}, Theorem \ref{thm:AsyDist_ATT_other} in the Appendix provides the uniform limiting distribution of the alternative estimated bounds where the probabilities in the numerator and denominator are both replaced by the sample treatment frequencies. This result shows that the alternative estimated bounds have a uniform asymptotic distribution as in \eqref{eq:asyDistBoundsMainText}, but with ${\Sigma}_{\upsilon}({\bf P})$ replaced by with a (weakly) larger asymptotic variance $\tilde{\Sigma}_{\upsilon}({\bf P})$ in a positive semidefinite sense. We can use this result to argue that the bounds in \eqref{eq:BoundsHat_ATT} are better than the alternative estimated bounds along the three dimensions discussed in Section \ref{sec:DiscussionATE}. We describe these briefly at the risk of some repetition.

First, the asymptotic analysis of the estimated bounds in \eqref{eq:BoundsHat_ATT} requires fewer assumptions than that of the alternative estimated bounds. The former is derived for $P \in \mathcal{P}_2$, while the latter also invokes Assumption \ref{ass:2}(f). Thus, the asymptotic behavior of the estimated bounds in \eqref{eq:BoundsHat_ATT} is robust to the details of the CAR mechanism, while the asymptotic behavior of the alternative estimated bounds is not.

Second, the estimated bounds in \eqref{eq:BoundsHat_ATT} are more efficient than either of the alternative estimated bounds. This follows directly from comparing the variance-covariance matrix of the asymptotic distribution of the two sets of bounds. That is, for the estimated bounds on the ATT, estimating treatment probability using target probabilities is more efficient than plugging in the sample treatment frequencies. Notice that this recommendation is the exact opposite of the one obtained in Section \ref{sec:DiscussionATE} for the estimated bounds on the ATE. We also note that these findings resemble the semiparametric efficiency results for point-identified ATT estimator in \cite{hirano/imbens/ridder:2003}, but are novel in our partially identified setting.

Finally, inference based on the estimated bounds in \eqref{eq:BoundsHat_ATT} is more powerful than the inference based on the alternative estimated bounds. That is, for all sequences of local alternative hypotheses, the limiting rejection rate of the CI in \eqref{eq:CSstoye_ATT} is (weakly) larger than the limiting rejection rate of the CI related to the alternative estimated bounds. This result is a consequence of the relative efficiency comparison described in the previous paragraph and our related work in \cite{bugni/gao/obradovic/velez:2024b}.

\section{Monte Carlo simulations}\label{sec:MC}

In this section, we illustrate the finite-sample performance of our inference methods using Monte Carlo simulations. This exercise has two goals. First, we aim to demonstrate that our asymptotic results are accurate in finite samples. This ensures that confidence intervals cover each point of the identified sets with a minimum prespecified probability. We particularly focus on the extremes of the identified set, where coverage is more challenging. Second, we seek to confirm that our recommended estimated bounds are preferable to alternative estimated bounds in terms of the statistical power of the related confidence intervals.

We consider three simulation designs. Each simulated dataset has $n=500$ i.i.d.\ individuals. All designs have four strata, i.e., $\mathcal{S} = \{1,2,3,4\}$, and we set $P(S_i=s) = 0.25$ for all $s \in \mathcal{S}$. As explained in Section \ref{sec:setup}, each individual can be a compiler (C), an always taker (AT), or a never taker (NT). We set $P(C|S_i=s) = 0.85$, $P(AT| S_i=s) = 0.05$, and $P(NT | S_i=s) = 0.10$ for all $s \in \mathcal{S}$. For each decision $d\in \{0,1\}$, type $t \in \{C, AT,NT\}$, and strata $s \in \mathcal{S}$, the potential outcome satisfies
\begin{equation}
    (Y_i(d) \mid T_i =t,~ S_i=s) ~\sim~ \text{Beta}(\alpha_{d,t,s}, 10-\alpha_{d,t,s}),
    \label{eq:MC_Yformula}
\end{equation}
where $\{\alpha_{d,t,s}: (d,t,s) \in \{0,1\} \times \{NT, C, AT\} \times\mathcal{S}\} $ varies with the design. Note that \eqref{eq:MC_Yformula} implies $Y_i(d) \in [0,1]$, and so we set $Y_L =0$ and $Y_H=1$. Given the realized strata, we simulate treatment assignment according to SRS or SBR, with target probability $(\pi_A(s):s \in \mathcal{S})$ that depends on the design. The description of the designs is completed as follows:
\begin{itemize}

    \item \textbf{Design 1:} For all $s \in \mathcal{S}$, $\alpha_{0,C,s} = 2$, $\alpha_{1,C,s} = 8$, $\alpha_{1,AT,s} = 3$, $\alpha_{0,NT,s} = 5$, and $\pi_A(s) = 0.5$.

    \item \textbf{Design 2:} For all $s \in \mathcal{S}$, $\alpha_{0,C,s} = 3 + (s-1)/3$, $\alpha_{1,C,s} = 4+(s-1)/3$, $\alpha_{1,AT,s} = 5+(s-1)/3$, $\alpha_{0,NT,s} = 2+(s-1)/3$, and $\pi_A(s) = 0.5$.

    \item \textbf{Design 3:} As in Design 2, but $(\pi_A(s): s\in \mathcal{S}) = (0.3,0.7,0.6,0.8)$. 
\end{itemize}
It is not hard to show that all designs satisfy Assumptions \ref{ass:1} and \ref{ass:2}.

For each design, we compute the identified sets for the ATE and ATT. For the ATE, the identified sets are $[0.425, 0.575]$ for Design 1 and $[0.040, 0.190]$ for Designs 2 and 3. For the ATT, the identified sets are $[0.463, 0.568]$ for Design 1, $[0.047, 0.153]$ for Design 2, and $[0.055, 0.145]$ for Design 3. We note that the LATE is 0.6 for Design 1 and 0.1 for Designs 2 and 3, indicating that this parameter may or may not be included in the identified sets of the ATE and ATT.

Tables \ref{tab:ate} and \ref{tab:att} provide results with $n=500$, $\alpha = 5\%$, and $5,000$ replications. We begin with Table \ref{tab:ate}, which presents the ATE results. In this case, recall that we recommend estimating the bounds using sample analog treatment assignment probabilities (i.e., ``sample'') rather than target probabilities (i.e., ``target''). Our asymptotic theory predicts that the rejection rates for $H: \theta = \theta_0$ for any $\theta_0$ in the identified set should not exceed $5\%$ as the sample size grows. Rejection rates in the interior of the identified set are typically much smaller than $5\%$, so we focus on the boundary, i.e., $\theta_0 = \theta_L$ or $\theta_0 = \theta_H$. The results show that the rejection rate at either boundary point is close to $5\%$, regardless of which CAR method is used or how treatment assignment probabilities are estimated. This is indicative that our asymptotic analysis is accurate with $n=500$. Next, we turn to the rejection rates for points outside of the identified set, such as $\theta_0 = \theta_L \times 0.9$ or $\theta_0 = \theta_H \times 1.1$. As expected, our simulations show that the rejection rate at either of these points is much larger than $5\%$. Notably, our recommended bounds estimator (i.e., ``sample'') delivers a higher rejection rate than the alternative bounds estimator (i.e., ``target''), especially with SRS.\footnote{\label{foot:SBR}Our results imply that both bounds estimators have equal asymptotic distribution under SBR.} Relatedly, the confidence interval of our recommended bounds estimator is, on average, shorter than that of the alternative bounds estimator.

\input{tableMC1}

We now turn to Table \ref{tab:att}, which provides results for the ATT. In this case, recall that we recommend estimating the bounds for the ATT by using the target probabilities for the treatment assignment probabilities in the numerator, and the sample analogs for those in the denominator (i.e., the ``target/sample'' combination). Any other choice of estimators for the treatment assignment probabilities may result in lower-quality inference. 
To verify this, one could consider alternative combinations of estimators for the treatment assignment probabilities. For brevity, we focus on the most natural alternative, resulting from using the sample analog to estimate both probabilities (i.e., the ``sample/sample'' combination).

The empirical findings for the ATT are qualitatively similar to those for the ATE. The rejection rate at either boundary point is close to $5\%$, regardless of how we implement the CAR or how we estimate the treatment assignment probabilities. The recommended bounds estimator produces confidence intervals that are slightly shorter and have a marginally higher rejection probability outside the identified set compared to the alternative bounds estimator. The main quantitative distinction with Table \ref{tab:ate} is that the differences between the recommended bounds estimator and the alternative one are very small. By inspecting our formal results, we can see that the differences between the asymptotic variances of the two bounds estimators are indeed small in magnitude across simulation designs.

\input{tableMC2}

\section{Empirical application}\label{sec:application}

In this section, we revisit \cite{dupas/karlan/robinson/ubfal:2018}, who conducted an RCT to assess the economic impact of expanding basic bank account access in three countries: Malawi, Uganda, and Chile.\footnote{The data are available at \url{https://www.aeaweb.org/articles?id=10.1257/app.20160597}.} Recently, \cite{bugni/gao:2023} focused on the RCT in Uganda and conducted inference on the (point-identified) LATE. We now use the analysis in this paper to perform inference on the ATE and ATT for this RCT.

We now briefly summarize the empirical setting of the RCT in Uganda; see \cite{dupas/karlan/robinson/ubfal:2018} for a more detailed description. \cite{dupas/karlan/robinson/ubfal:2018} randomly selected 2,159 Ugandan households without bank accounts in 2011 and assigned them to treatment or control groups. Treated households received a voucher for a free savings account and assistance with the paperwork, while control group households did not receive these vouchers. We use the binary variable $A_i \in \{0,1\}$ to indicate if household $i$ was given the voucher for the free savings account. Households were stratified by gender, occupation, and bank branch, creating 41 strata, i.e., $s \in \mathcal{S} = \{1, \ldots, 41\}$, and were assigned to treatment or control using SBR with $\pi_A(s) = 1/2$ for all $s \in \mathcal{S}$. We define the binary variable $D_i = D_i(A_i) \in \{0,1\}$ to indicate if household $i$ has opened and used the free savings account in this RCT. 

This RCT featured one-sided noncompliance. \citet[Table 1]{dupas/karlan/robinson/ubfal:2018} reveals that none of the 1,079 households in the control group accessed a free savings account. In contrast, among the 1,080 treated households, only 54\% actually opened a bank account, and only 42\% made at least one deposit during the RCT. We use three outcomes measured in 2010 US dollars: savings in formal financial institutions, savings in cash at home or in a secret place, and expenditures in the last month.

There are two relevant aspects of this empirical exercise worth discussing. First, due to one-sided noncompliance and the constant target assignment probability, the ATT and the LATE coincide (both point identified). As a corollary, our results for the ATT should align with those obtained for the LATE by \cite{bugni/gao:2023}. Second, our inference methods require knowledge of the logical lower and upper bounds for the outcomes of interest. Assuming that the outcome variables are bounded, we consistently estimate these bounds using the sample minimum and maximum. The minimum value of all variables is zero dollars. The maximum values for savings in formal financial institutions, savings in cash at home or in a secret place, and expenditures in the last month are 293.3, 366.6, and 256.6 dollars, respectively.

Table \ref{tab:app} presents the empirical results for the ATE and ATT. For the sake of comparison, we also include the results for the LATE obtained by \cite{bugni/gao:2023}. As expected, the identified set of the ATT is estimated to be a point, which coincides with the LATE estimator obtained by \cite{bugni/gao:2023}.\footnote{The differences between the CIs are caused by slight variations in the estimators for the standard errors and are asymptotically negligible.} In contrast, the ATE is partially identified.
We estimate that the ATE of opening and using these savings accounts on the amount saved in formal financial institutions is between 4.03 and 167.86 dollars, with a 95\% CI between 1.36 and 175.18 dollars. Additionally, we estimate that the ATE of opening and using these savings accounts on the amount saved in cash or at home is between -14.54 and 190.25 dollars, with a CI spanning between -17.51 and 199.70 dollars. Finally, we estimate that the ATE of opening and using these savings accounts on expenditures is between -18.26 and 125.09 dollars, with a CI spanning between -20.92 and 131.21 dollars. While these intervals are admittedly wide, it is relevant to note that the corresponding identified sets are sharp, meaning they cannot be restricted further without adding additional information. Put differently, the only way to shrink the identified set—and consequently the resulting estimates and confidence sets—is to impose additional assumptions on the model.

\input{tableApplication}

\section{Conclusions}\label{sec:conclusions}

This paper studied the identification and inference of the ATE and ATT in RCTs that have CAR and imperfect compliance. The interplay between these two aspects of the RCT makes our analysis novel relative to the existing literature. The imperfect compliance implies that our estimands are partially identified, contrasting with the CAR literature focusing on point-identified parameters. In fact, to our knowledge, ours is the first paper to consider partial identification analysis in RCTs with CAR. In turn, the treatment assignment using CAR implies that our data may not be i.i.d., which is non-standard in the identification analysis of the ATE and ATT.

We derive the identified set of the ATE and ATT. Due to CAR, this characterization does not follow from the existing literature under i.i.d.\ assumptions. In fact, our results appear to be the first characterization of these identified sets in the CAR literature with imperfect compliance. Our identified sets are intervals with endpoints given by sharp bounds. We provide estimators for these sharp bounds and demonstrate that they are asymptotically normally distributed, uniformly across a relevant class of probability distributions. This result enables us to (i) estimate the identified sets of the ATE and ATT uniformly consistently and (ii) provide uniformly valid confidence sets for the ATE and ATT. 

An important aspect of our inference is the estimator of the treatment assignment probabilities used to construct the estimated sharp bounds. The two natural options are (i) the target values of the treatment assignment probabilities (known in an RCT) or (ii) the sample analog of the treatment assignment probabilities. Our asymptotic analysis yields concrete practical recommendations. For the ATE, we recommend using the sample analog estimators for these probabilities. In the case of the ATT, we recommend using the target values in the numerator and the sample analogs in the denominator. There are three reasons behind these recommendations. First, our recommended option produces estimated bounds that are relatively more efficient. Second, and relatedly, our recommended option leads to confidence sets that reject local alternative hypotheses with relatively higher probability than the alternative option. Finally, our recommended estimated bounds can be obtained under weaker assumptions than the alternative options.

Using Monte Carlo simulations, we confirm that our asymptotic predictions and recommendations are relevant in finite samples. Finally, we illustrate our methodology with an empirical application to the RCT implemented in \cite{dupas/karlan/robinson/ubfal:2018}.

A natural extension of our results is considering parameters of interest beyond the ATE and the ATT. In line with this, we study the identification and inference for the average treatment effect on the untreated (ATU). Qualitatively speaking, the results resemble those obtained for the ATT. For brevity, we place them in Section \ref{sec:appendixE} of the Appendix. An interesting avenue for future research is to generalize the framework to encompass parameters representable via marginal treatment effects of \cite{bjorklund/moffit:1987} and \cite{heckman/vytlacil:1999}. We are actively pursuing research in this direction.

%% file: tableMC1.tex
\begin{table}[ht!]
    \centering
    \setlength{\tabcolsep}{4.0pt} 
     \renewcommand{\arraystretch}{0.5} 
        \scalebox{0.85}{\begin{tabular}{ccccccccc}
\midrule     
\toprule
        \multirow{2}[1]{*}{Design} & \multirow{2}[1]{*}{$[\theta_L,~ \theta_H]$} & \multirow{2}[1]{*}{CAR} & \multirow{1}[1]{*}{Estim.\ of} & \multicolumn{4}{c}{Rej.\ rates for $H:\theta = {\theta}_0$} & \multirow{2}{*}{Avg. CI} \\
        \cmidrule(lr){5-8}
        & & & ${P}(A_{i}|S_i)$ & $ \theta_0 = \theta_L$ & $ \theta_0 = \theta_H$ & $ \theta_0 = \theta_L$\text{\small $\times$}0.9  & $ \theta_0 = \theta_H$\text{\small $\times$}1.1 & length \\
\midrule    
\multirow{4}[1]{*}{1}  & \multirow{4}[1]{*}{[0.425, 0.575]}  & SRS & target & 5.4 & 5.3 & 22.8 & 34.4 & 0.309 \\
          & & SRS & sample & 5.7 & 4.3 & 59.3 & 99.2 & 0.211 \\
          & & SBR & target & 6.5 & 4.5 & 59.6 & 99.1 & 0.211 \\
          & & SBR & sample & 6.2 & 4.0 & 59.8 & 99.1 & 0.211 \\    
\midrule   
\multirow{4}[1]{*}{2} &  \multirow{4}[1]{*}{[0.040, 0.190]}  & SRS & target & 5.3 & 5.3 & 6.4 & 12.6 & 0.283 \\
          & & SRS & sample & 5.8 & 5.7 & 8.8 & 30.6 & 0.207 \\
          & & SBR & target & 5.4 & 5.8 & 8.6 & 30.9 & 0.207 \\
          & & SBR & sample & 5.4 & 5.5 & 8.3 & 30.9 & 0.207 \\   
\midrule            
\multirow{4}[1]{*}{3}  & \multirow{4}[1]{*}{[0.040, 0.190]} & SRS & target & 5.0 & 5.4 & 6.0 & 11.1 & 0.294 \\
          & & SRS & sample & 6.1 & 5.6 & 8.7 & 28.3 & 0.212 \\
          & & SBR & target & 6.1 & 5.8 & 9.0 & 27.5 & 0.212 \\
          & & SBR & sample & 6.1 & 5.9 & 8.8 & 27.4 & 0.212 \\          
\bottomrule\midrule     
    \end{tabular}}
     \caption{ \footnotesize Simulation results for inference on the ATE with $n=500$, $\alpha = 5\%$, and $5,000$ replications. CAR method can be implemented via SRS or SBR, and ${P}(A_i|S_i)$ in the bounds can be estimated using ``target'', i.e., target probabilities $\pi_A(S_i)$, or ``sample'', i.e., sample analogs $n_A(S_i)/n(S_i)$. We compute rejection rates (in \%) for $H:\theta = {\theta}_0$ with  ${\theta}_0 = \theta_L$, ${\theta}_0 = \theta_H$, ${\theta}_0 = \theta_L$\text{\small $\times$}0.9 and ${\theta}_0 = \theta_H$\text{\small $\times$}1.1, and average confidence interval length.  
     }
     \label{tab:ate}
\end{table}
%

%% file: tableMC2.tex
\begin{table}[ht!]
    \centering
    \setlength{\tabcolsep}{4.0pt} 
\renewcommand{\arraystretch}{0.5} 
        \scalebox{0.8}{\begin{tabular}{ccccccccc}\midrule     
        \toprule
        \multirow{2}[1]{*}{Design} & \multirow{2}[1]{*}{$[\upsilon_L, ~\upsilon_H]$} & \multirow{2}[1]{*}{CAR} & \multirow{1}[1]{*}{Estim.\ of ${P}(A_{i}|S_i)$} & \multicolumn{4}{c}{Rej.\ rates for $H:\upsilon = {\upsilon}_0$} & \multirow{2}{*}{Avg. CI} \\
        \cmidrule(lr){5-8}
        & & & in num/den & $ \upsilon_0 = \upsilon_L$ & $ \upsilon_0 = \upsilon_H$ & $ \upsilon_0 = \upsilon_L$\text{\small $\times$}0.9  & $ \upsilon_0 = \upsilon_H$\text{\small $\times$}1.1 & length \\
\midrule         
  \multirow{4}[1]{*}{1}  & \multirow{4}[1]{*}{[0.463, 0.568]}  & SRS & target/sample & 6.9 & 4.8 & 44.5 & 97.2 & 0.183 \\
          & & SRS & sample/sample & 6.7 & 4.9 & 43.8 & 97.1 & 0.184 \\
          & & SBR & target/sample & 6.2 & 4.5 & 44.6 & 96.7 & 0.184 \\
          & & SBR & sample/sample & 6.1 & 4.6 & 44.8 & 96.7 & 0.184 \\    
        \midrule  
      \multirow{4}[1]{*}{2}  & \multirow{4}[1]{*}{[0.047, 0.153]}   & SRS & target/sample & 6.5 & 5.4 & 9.8 & 22.4 & 0.169 \\
          & & SRS & sample/sample & 6.3 & 5.3 & 9.6 & 21.8 & 0.170 \\
          & & SBR & target/sample & 6.0 & 5.3 & 8.9 & 21.0 & 0.170 \\
          & & SBR & sample/sample & 6.0 & 5.2 & 8.9 & 21.1 & 0.170 \\   
        \midrule           
        \multirow{4}[1]{*}{3}  & \multirow{4}[1]{*}{[0.055, 0.145]}  & SRS & target/sample & 6.8 & 6.0 & 9.9 & 18.9 & 0.162 \\
          & & SRS & sample/sample & 6.7 & 5.7 & 9.9 & 18.4 & 0.163 \\
          & & SBR & target/sample & 6.8 & 5.6 & 10.3 & 18.4 & 0.162 \\
          & & SBR & sample/sample & 6.8 & 5.7 & 10.2 & 18.4 & 0.162 \\    
        \bottomrule\midrule     
    \end{tabular}}
     \caption{ \footnotesize Simulation results for inference on the ATT with $n=500$, $\alpha = 5\%$, and $5,000$ replications. CAR method can be implemented via SRS or SBR, and ${P}(A_i|S_i)$ in the numerator and denominator of the bounds can be estimated using ``target'', i.e., target probabilities $\pi_A(S_i)$, or ``sample'', i.e., sample analogs $n_A(S_i)/n(S_i)$.  We compute rejection rates (in \%) for $H:\upsilon={\upsilon }_0$ with  ${\upsilon }_0=\upsilon _L$, ${\upsilon }_0=\upsilon _H$, ${\upsilon }=\upsilon _L$\text{\small $\times$}0.9 and ${\upsilon }_0=\upsilon _H$\text{\small $\times$}1.1, and average confidence interval length.  %
     }\label{tab:att}
\end{table}
%

%% file: tableApplication.tex
\begin{table}[htbp]
\centering
\scalebox{0.95}{\begin{tabular}{ccccc}
\midrule    \toprule
\multirow{2}{*}{Parameter} & \multirow{2}{*}{Object}  & {Savings in formal}
& Savings in cash at  & Expenditures in  \\
& & fin.\ institutions & home or secret place & last month \\
\midrule
\multirow{2}[2]{*}{ATE} &
Estimated IS &
[4.03, 167.85] &
[-14.54, 190.25] &
[-18.26, 125.09] \\
& 95\% CI &
[1.36, 175.18] &
[-17.51, 199.70] &
[-20.92, 131.21] \\
\midrule
\multirow{2}[2]{*}{ATT} &
Estimated IS &
[17.57, 17.57] &
[-7.32, -7.32] &
[-2.43, -2.43] \\
& 95\% CI &
[9.61, 25.53] &
[-16.77, 2.12] &
[-10.01, 5.16] \\
\midrule
\multirow{2}[2]{*}{LATE} &
Estimator &
17.57 &
-7.32 &
-2.43 \\
& 95\% CI &
[9.61, 25.53] &
[-16.78, 2.13] &
[-10.05, 5.19] \\
\bottomrule\midrule    
\end{tabular}}%
\caption{\footnotesize Estimated identified set and CI based on data from \cite{dupas/karlan/robinson/ubfal:2018}. The ATE and ATT results were obtained by the recommended methods described in Sections \ref{sec:ate} and \ref{sec:att}, respectively.  The LATE results were copied from \citet[Section 8]{bugni/gao:2023}. ``Estimated IS'' denotes the estimate of the identified set, ``95\% CI'' denotes the CI with $\alpha = 5\%$, and ``Estimator'' denotes the LATE point estimate for LATE.}
\label{tab:app}
\end{table}

%% file: combined_appendix.tex
\section{Appendix on identification}\label{sec:appendix1}

Throughout all the appendices, we use LHS, RHS, and s.t.\ to abbreviate ``left-hand side'', ``right-hand side'', and ``such that'', respectively.
The results in Appendix \ref{sec:appendix1} are derived for an arbitrary sample size $n \in \mathbb{N}$ and a distribution ${\bf P}$ determined by \eqref{eq:outcome}, ${\bf Q}$, and ${\bf G}$.

\subsection{Preliminary results}

For any $( a,z) \in \{ 0,1\} \times \mathcal{Z}$, define
\begin{equation}
\pi _{D( a)}( z) ~\equiv~ P(D( a) =1|Z=z) .
\label{eq:pi_def}
\end{equation}
This object differs from $\pi _{D( a)}( s)$ in \eqref{eq:keyDefns} in that it conditions on $z \in \mathcal{Z}$ instead of $s \in \mathcal{S}$. By Assumption \ref{ass:1}, this object does not depend on $i=1,\dots,n$. 

\begin{lemma}\label{lem:QnMap} 
Let Assumptions \ref{ass:1}(a)-(c) and \ref{ass:2}(a)-(c) hold. Then, for any $i=1,\ldots ,n$ and $z\in \mathcal{Z}$,
\begin{align*}
P( D_{i}( 0) =0,D_{i}( 1) =0|Z_{i}=z) & ~=~1-\pi _{D( 1) }( z) \\
P( D_{i}( 0) =1,D_{i}( 1) =1|Z_{i}=z) & ~=~\pi _{D( 0) }( z) \\
P( D_{i}( 0) =0,D_{i}( 1) =1|Z_{i}=z) & ~=~\pi _{D( 1)}( z) -\pi _{D( 0) }( z) .
\end{align*}
\end{lemma}
\begin{proof}
Fix $z\in \mathcal{Z}$ and $i=1,\ldots ,n$ arbitrarily. To get the first statement,
\begin{equation*}
P( D_{i}( 0) =0,D_{i}( 1) =0|Z_{i}=z) ~\overset{(1)}{=}~P( D_{i}( 1) =0|Z_{i}=z) ~=~1-\pi _{D( 1)}( z) ,
\end{equation*}
where (1) holds by Assumption \ref{ass:1}(b). The other statements hold by analogous arguments.
\end{proof}

For any $( a,d,z,y) \in \{ 0,1\} \times \{ 0,1\} \times \mathcal{Z}\times \mathbb{R}$ and $i=1,\ldots ,n$, we define
\begin{align}
f_{a,d}( y,z) &~\equiv~ \left\{
\begin{tabular}{ll}
$dP( Y( d) =y|D( a) =d,Z=z) $ & if $\pi _{D( a) }( z) ^{d}( 1-\pi _{D( a)}( z) ) ^{1-d}>0$ \\
$0$ & if $\pi _{D( a) }( z) ^{d}( 1-\pi _{D( a)}( z) ) ^{1-d}=0$.
\end{tabular}
\right. 
\label{eq:f_def}
\end{align}
By Assumption \ref{ass:1}, this object does not depend on $i=1,\dots,n$. Note that \eqref{eq:pi_def} implies that $dP( D_{i}( a) =d,Z_{i}=z) =\pi _{D( a) }( z) ^{d}( 1-\pi _{D( a)}( z) ) ^{1-d} dP(Z=z) $, and so $z\in \mathcal{Z}$ and $\pi _{D( a)}( z) ^{d}( 1-\pi _{D( a) }( z) ) ^{1-d}>0$ implies that $dP( Y( d) =y|D( a) =d,Z=z) $ is well defined. 
The next result connects some of these objects.

\begin{lemma} \label{lem:QnYMap}
Let Assumptions \ref{ass:1}(a)-(c) and \ref{ass:2}(a)-(c) hold. For any $i=1,\ldots ,n$ and $( z,y) \in \mathcal{Z}\times \mathbb{R}$,
\begin{enumerate}
\item If $P( D_{i}( 0) =1,D_{i}( 1) =1|Z_{i}=z) =\pi _{D( 0) }( z) >0$,
\begin{equation}
dP( Y_{i}( 1) =y|D_{i}( 1) =1,D_{i}( 0) =1,Z_{i}=z) ~=~f_{0,1}( y,z) . \label{eq:QnYMap_1}
\end{equation}
\item If $P( D_{i}( 0) =0,D_{i}( 1) =0|Z_{i}=z) =1-\pi _{D( 1) }( z) >0$,
\begin{equation}
dP( Y_{i}( 0) =y|D_{i}( 1) =0,D_{i}( 0) =0,Z_{i}=z) ~=~f_{1,0}( y,z) . \label{eq:QnYMap_2}
\end{equation}
\item If $P( D_{i}( 0) =0,D_{i}( 1) =1|Z_{i}=z) =\pi _{D( 1) }( z) -\pi _{D( 0)}( z) >0$,
\begin{align}
\text{\small $dP( Y_{i}( 1) =y|D_{i}( 1) =1,D_{i}( 0) =0,Z_{i}=z)$}
&=
\tfrac{f_{1,1}( y,z) \pi _{D( 1)}( z) -f_{0,1}( y,z) \pi _{D( 0)}( z) }{\pi _{D( 1) }( z) -\pi _{D( 0) }( z) } \label{eq:QnYMap_3} \\
\text{\small $dP( Y_{i}( 0) =y|D_{i}( 1) =1,D_{i}( 0) =0,Z_{i}=z)$} 
& =
\tfrac{f_{0,0}( y,z) ( 1-\pi _{D( 0)}( z) ) -f_{1,0}( y,z) ( 1-\pi _{D( 1) }( z) ) }{\pi _{D( 1) }( z) -\pi _{D( 0) }( z) }.\label{eq:QnYMap_4}
\end{align}
\end{enumerate}
\end{lemma}
\begin{proof}
Fix $( z,y) \in \mathcal{Z}\times \mathbb{R}$ and $i=1,\ldots ,n$ arbitrarily. Also, set $z_i=z$. 

\noindent \underline{Part 1.} First, note that $P( D_{i}( 0) =1,D_{i}( 1) =1|Z_{i}=z) =\pi _{D( 0)}( z) >0$ implies that the LHS expression in \eqref{eq:QnYMap_1} is well defined. To show \eqref{eq:QnYMap_1}, note that
{\begin{align*}
dP( Y_{i}( 1) =y|D_{i}( 0) =D_{i}( 1) =1,Z_{i}=z) ~\overset{(1)}{=}~dP( Y_{i}( 1) =y|D_{i}( 0) =1,Z_{i}=z) ~\overset{(2)}{=}~f_{0,1}( y,z),
\end{align*}}
where (1) holds by Assumption \ref{ass:1}(b) and (2) holds by \eqref{eq:f_def} and $\pi _{D( 0) }( z) >0$.

\noindent \underline{Part 2.} This follows from an analogous argument to part 1, and is therefore omitted.

\noindent \underline{Part 3.} We only show \eqref{eq:QnYMap_3}, as \eqref{eq:QnYMap_4} follows from a similar derivation.

Note that $P( D_{i}( 0) =0,D_{i}( 1) =1|Z_{i}=z_{i}) =\pi _{D( 1)}( z_{i}) -\pi _{D( 0)}( z_{i}) >0$ implies that the LHS expressions in \eqref{eq:f_def} and \eqref{eq:QnYMap_1} are well defined.

We now consider two cases. In the first case, $\pi _{D( 0) }( z_{i}) >0$. By Lemma \ref{lem:QnMap}, this implies $P( D_{i}( 0) =1,D_{i}( 1) =1|Z_{i}=z_{i}) >0$. In that case, we have
\begin{align}
f_{1,1}( y,z_{i})~&=~dP( Y_{i}( 1) =y|D_{i}( 1) =1,Z_{i}=z_{i}) \notag \\
& ~\overset{(1)}{=}~\tfrac{\left\{
{\scriptsize\begin{array}{c}
dP( Y_{i}( 1) =y|D_{i}( 1) =D_{i}( 0) =1,Z_{i}=z_{i})  \pi _{D( 0)}( z_{i}) + \\
dP( Y_{i}( 1) =y|D_{i}( 1) =1,D_{i}( 0) =0,Z_{i}=z_{i})  ( \pi _{D( 1)}( z_{i}) -\pi _{D( 0)}( z_{i}) )
\end{array}}
\right\} }{\pi _{D( 1) }( z_{i}) } \notag \\
& ~\overset{(2)}{=}~\tfrac{f_{0,1}( y,z_{i})  \pi _{D( 0) }( z_{i}) +dP( Y_{i}( 1) =y|D_{i}( 1) =1,D_{i}( 0) =0,Z_{i}=z_{i})  ( \pi _{D( 1) }( z_{i}) -\pi _{D( 0)}( z_{i}) ) }{\pi _{D( 1)}( z_{i}) }, \label{eq:QnYMap_5}
\end{align}
where (1) holds by Assumption \ref{ass:1}(b) and Lemma \ref{lem:QnMap}, and (2) by \eqref{eq:QnYMap_1}. In the second case, $\pi _{D( 0)}( z_{i}) =0$, and we can repeat the derivation to get
\begin{equation*}
f_{1,1}( y,z_{i}) ~=~dP( Y_{i}( 1) =y|D_{i}( 1) =1,D_{i}( 0) =0,Z_{i}=z_{i}) . 
\end{equation*}
Note that this coincides with the RHS of \eqref{eq:QnYMap_5} for $\pi _{D( 0)}( z_{i}) =0$. Thus, we can use the RHS of \eqref{eq:QnYMap_5} to summarize both cases.
By solving for $dP( Y_{i}( 1) =y|D_{i}( 1) =1,D_{i}( 0) =0,Z_{i}=z_{i}) $ in \eqref{eq:QnYMap_5}, \eqref{eq:QnYMap_3} follows.
\end{proof}

\subsection{ATE}\label{sec:appendix1ATE}

\begin{proof}[Proof of Theorem \ref{thm:main}]
First, we show that the RHS of \eqref{eq:bounds} do not depend on $i=1,\dots ,n$. To this end, for any $i=1,\dots ,n$ and $\bar{Y}\in \{Y_{L},Y_{H}\}$, consider the following argument.
\begin{align}
E[ {(Y_{i}D_{i}+\bar{Y}(1-D_{i}))A_{i}}|S_{i}] 
&=E[E[ {(Y_{i}D_{i}+\bar{Y}(1-D_{i}))A_{i}}|S^{(n)}] |S_{i}]
\notag \\
&\overset{(1)}{=}E[E[ {(Y_{i}(1)D_{i}(1)+\bar{Y}(1-D_{i}(1)))A_{i}} |S^{(n)}] |S_{i}]  \nonumber \\
&\overset{(2)}{=}E[E[ {(Y_{i}(1)D_{i}(1)+\bar{Y}(1-D_{i}(1)))}|S^{(n)} ] E[ {A_{i}}|S^{(n)}] |S_{i}]  \nonumber \\
&\overset{(3)}{=}E[E[ {(Y_{i}(1)D_{i}(1)+\bar{Y}(1-D_{i}(1)))}|S_{i}] E[ {A_{i}}|S^{(n)}] |S_{i}]  \notag \\
&\overset{(4)}{=}E[ {(Y_{i}(1)D_{i}(1)+\bar{Y}(1-D_{i}(1)))}|S_{i}] P(A_{i}=1|S_{i}),\label{eq:representation1}
\end{align}
where (1) holds by $Y_{i}(1)D_{i}=Y_{i}D_{i}$ and $D_{i}(1)A_{i}=D_{i}A_{i}$, (2) by Assumption \ref{ass:2}(a), (3) by the i.i.d.\ condition in Assumption \ref{ass:1}, and (4) by $E[E[ {A_{i}}|S^{(n)}] |S_{i}]=P(A_{i}=1|S_{i})$. 
Next consider the following argument,
\begin{align}
E\left[ {(Y_{i}D_{i}+\bar{Y}(1-D_{i}))A_{i}}/P(A_{i}=1|S_{i})\right] 
&=E\left[ E\left[ {(Y_{i}D_{i}+\bar{Y}(1-D_{i}))A_{i}}|S_{i}\right] /P(A_{i}=1|S_{i})\right]  \notag\\
&\overset{(1)}{=}E\left[ {(Y_{i}(1)D_{i}(1)+\bar{Y}(1-D_{i}(1)))}\right]  \nonumber \\
&\overset{(2)}{=}E\left[ {(Y(1)D(1)+\bar{Y}(1-D(1)))}\right] , \label{eq:representation2} 
\end{align}
where (1) holds by \eqref{eq:representation1} and $P(A_{i}=1|S_{i}) = E[E[ {A_{i}}|S^{(n)}] |S_{i}] >0$ (which, in turn, follows from Assumptions \ref{ass:1}(a) and \ref{ass:2}(b), and \citet[Theorem D, page 104]{halmos:1974}) and (2) by the i.i.d.\ condition in Assumption \ref{ass:1}. Analogously, we can get 
\begin{equation}
E[Y(0)(1-D(0))+\bar{Y}D(0)]=E[ {(Y_{i}(1-D_{i})+\bar{Y}D_{i})(1-A_{i})}/P(A_{i}=0|S_{i})] .  \label{eq:representation3}
\end{equation}
The desired result follows from \eqref{eq:bounds}, and \eqref{eq:representation2} and \eqref{eq:representation3} with $\bar{Y}\in \{Y_{L},Y_{H}\}$.

Second, we show that
\begin{equation}
    \theta \in [\theta _{L}(\mathbf{P}),\theta _{H}(\mathbf{P})].\label{eq:representation4} 
\end{equation}
We only show the lower bound, as the upper bound is analogous. 
To this end, consider the following argument.
\begin{align*}
\theta  
&~=~E[Y_{i}(1)D_{i}+Y_{i}(1)(1-D_{i})]-E[Y_{i}(0)D_{i}+Y_{i}(0)(1-D_{i})] 
\notag \\
&~\overset{(1)}{\geq }~E[Y_{i}(1)D_{i}+Y_{L}(1-D_{i})]-E[Y_{i}(0)D_{i}+Y_{H}(1-D_{i})]  \notag \\
&~\overset{(2)}{=}~\theta _{L}(\mathbf{P}), 
\end{align*}
as desired, where (1) holds by $Y_{i}(d)\in [Y_{L},Y_{H}]$ and (2) by \eqref{eq:bounds}, \eqref{eq:representation2}, and \eqref{eq:representation3} with $\bar{Y}\in \{Y_{L},Y_{H}\}$.


Next, we establish that $[\theta _{L}(\mathbf{P}),\theta _{H}(\mathbf{P})]$ are the sharp bounds for the ATE. 
First, \eqref{eq:representation4} implies that $[\theta _{L}(\mathbf{P}),\theta _{H}(\mathbf{P})]$ is an outer identified set (i.e., a superset of the identified set). Lemma \ref{lem:main_inner} proposes a hypothetical distribution for all underlying variables that (i) is observationally equivalent to the data distribution $\mathbf{P}$ and (ii) produces ATEs that span $[\theta _{L}(\mathbf{P}),\theta _{H}(\mathbf{P})]$. As a corollary, $[\theta _{L}(\mathbf{P}),\theta _{H}(\mathbf{P})]$ is a subset of the identified set. By combining these observations, the desired result follows.
\end{proof}

\begin{lemma}\label{lem:main_inner} 
Let Assumptions \ref{ass:1}(a)-(c) and \ref{ass:2}(a)-(c) hold. For any $\beta _{0},\beta _{1}\in [ 0,1] $ and $ w^{(n)} = \{ (y_{1,i},y_{0,i},d_{1,i},d_{0,i},z_{i})\} _{i=1}^{n}$, let ${\bf R}(\cdot;\beta_0,\beta_1)$ denote the following measure for $( W^{( n) },A^{( n) })$ at $( w^{(n)},a^{( n) })$:
\begin{align}
& d{\bf R}( ( W^{( n) },A^{( n) }) =( w^{(n)},a^{( n) }) ;\beta _{0},\beta _{1})  ~\equiv~\notag \\
& 
\prod_{i=1}^{n}\left[
{\scriptsize\begin{array}{c}
\text{\small $f_{0,1}( y_{1,i},z_{i})  I[ y_{0,i}=Y_{L}\beta _{0}+Y_{H}( 1-\beta _{0}) ]  \pi _{D( 0)}( z_{i})  I[ d_{0,i}=d_{1,i}=1]$} \\
\text{\small $+\lambda _{1}( y_{1,i},z_{i})  \lambda _{0}( y_{0,i},z_{i})  ( \pi _{D( 1)}( z_{i}) -\pi _{D( 0)}( z_{i}) )  I[ d_{0,i}=0,d_{1,i}=1] +$} \\
\text{\small $I[ y_{1,i}=Y_{L}\beta _{1}+Y_{H}( 1-\beta _{1}) ]  f_{1,0}( y_{0,i},z_{i})  ( 1-\pi _{D( 1)}( z_{i}) )  I[ d_{0,i}=d_{1,i}=0 ]$}
\end{array}}
\right]{\text{\small $dP( Z_{i}=z_{i})$}} \notag\\
&~\times~ P( A^{( n) }=a^{( n) }|S^{( n) }=\{ S( z_{i}) \} _{i=1}^{n}) \label{eq:Rn_dist_defn}
\end{align}
where $\pi _{D( 0)}( z) $ and $\pi _{D( 1)}( z) $ are as in \eqref{eq:pi_def}, $f_{1,0}( y,z) $, $f_{0,1}( y,z) $, $ f_{0,0}( y,z) $, and $f_{1,1}( y,z) $ are as in \eqref{eq:f_def}, and
\begin{align}
\lambda _{1}( y,z)  &~\equiv~\left\{ 
\begin{tabular}{ll}
$\frac{f_{1,1}( y,z) \pi _{D( 1) }( z) -f_{0,1}( y,z) \pi _{D( 0)}( z) }{\pi _{D( 1)}( z) -\pi _{D( 0)}( z) }$ & if $\pi _{D( 1)}( z) -\pi _{D( 0)}( z) >0$ \\
$0$ & if $\pi _{D( 1) }( z) -\pi _{D( 0)}( z) =0$,
\end{tabular}
\right. \notag\\
\lambda _{0}( y,z)  &~\equiv~ \left\{ 
\begin{tabular}{ll}
$\frac{f_{0,0}( y,z) ( 1-\pi _{D( 0) }( z) ) -f_{1,0}( y,z) ( 1-\pi _{D( 1) }( z) ) }{\pi _{D( 1) }( z) -\pi _{D( 0)}( z) }$ & if $\pi _{D( 1)}( z) -\pi _{D( 0)}( z) >0$ \\
$0$ & if $\pi _{D( 1)}( z) -\pi _{D( 0)}( z) =0$.
\end{tabular}
\right.   \label{eq:mu_defn} 
\end{align}
Then, for any $\beta _{0},\beta _{1}\in [ 0,1] $,
\begin{enumerate}
\item ${\bf R}(( W^{( n) },A^{( n) });\beta _{0},\beta _{1}) $ is a probability measure.
\item ${\bf R}(( W^{( n) },A^{( n) });\beta _{0},\beta _{1}) $ induces a probability distribution for $X^{( n)}$ that is observationally equivalent to ${\bf P}(X^{( n) }) $, i.e., $ {\bf R}( X^{( n) };\beta _{0},\beta _{1}) ={\bf P}(X^{( n) }) $.
\item For every $i=1,\dots,n$, $E_{{\bf R}}[Y_{i}(1)-Y_{i}(0);\beta _{0},\beta _{1}] = E[Y( 1) D( 1) +( Y_{L}\beta
_{1}+Y_{H}( 1-\beta _{1}) ) ( 1-D( 1)) -Y( 0) ( 1-D( 0) ) - (Y_{L}\beta _{0}+Y_{H}( 1-\beta _{0}) ) D(
0)  ]] $. Therefore,
\begin{align*}
    &\underset{\beta _{0},\beta _{1} \in [0,1]}{\bigcup} E_{{\bf R}}[Y_{i}(1)-Y_{i}(0);\beta _{0},\beta _{1}]\\
    &~=~  \left[
\begin{array}{c}
E[ Y( 1) D( 1) +Y_{L}( 1-D( 1) ) -Y( 0) ( 1-D( 0) ) -Y_{H}D( 0)] , \notag\\
E[ Y( 1) D( 1) +Y_{H}( 1-D( 1) ) -Y( 0) ( 1-D( 0) ) -Y_{L}D( 0) ]
\end{array}
\right].
\end{align*}
\end{enumerate}
\end{lemma}
\begin{proof}
We fix $\beta _{0},\beta _{1}\in [ 0,1] $ arbitrarily throughout this proof.

\noindent \underline{Part 1.} The desired result follows from establishing:
\begin{align}
d{\bf R}(( W^{( n) },A^{( n) }) =(w^{(n)},a^{( n) }) ;\beta _{0},\beta _{1})& ~\geq ~0, \label{eq:main_inner1} \\
\sum\nolimits_{a^{( n) }}\int\nolimits_{w^{(n)}}d{\bf R}( ( W^{( n) },A^{( n) }) =(w^{(n)},a^{( n) }) ;\beta _{0},\beta _{1}) &~=~1.\label{eq:main_inner2}
\end{align}
First, \eqref{eq:main_inner1} follows from the fact that all objects in \eqref{eq:Rn_dist_defn} are non-negative. In particular, $\lambda _{1}( y,z) \geq 0$ and $\lambda _{0}( y,z) \geq 0$ follow from \eqref{eq:f_def}, \eqref{eq:mu_defn}, and the arguments in the proof of Lemma \ref{lem:QnYMap}.
Second, we now show \eqref{eq:main_inner2}. Note that \eqref{eq:Rn_dist_defn} implies
\begin{equation*}
{\bf R}(A^{( n) }=a^{( n) }|W^{( n) } = w^{(n)};\beta _{0},\beta _{1}) ~=~{\bf P}( A^{( n) }=a^{( n) }|S^{( n) }= (S( z_{i}) )_{i=1}^{n}) .
\end{equation*}
and, for $w_i = (y_{1,i},y_{0,i},d_{1,i},d_{0,i},z_{i})$,
\begin{align*}
&d{\bf R}(W_{i}=w_i ;\beta _{0},\beta _{1})~=~ \\
&\left[ 
{\small\begin{array}{c}
f_{0,1}( y_{1,i},z_{i})  I[ y_{0,i}=Y_{L}\beta _{0}+Y_{H}( 1-\beta _{0}) ]  \pi _{D( 0)}( z_{i})  I[ d_{0,i}=1,d_{1,i}=1]\\
+\lambda _{1}( y_{1,i},z_{i})  \lambda _{0}( y_{0,i},z_{i})  ( \pi _{D( 1)}( z_{i}) -\pi _{D( 0)}( z_{i}) )  I[ d_{0,i}=0,d_{1,i}=1] + \\
I[ y_{1,i}=Y_{L}\beta _{1}+Y_{H}( 1-\beta _{1}) ]  f_{1,0}( y_{0,i},z_{i})  ( 1-\pi _{D( 1) }( z_{i}) )  I[ d_{0,i}=0,d_{1,i}=0 ]
\end{array}}
\right] dP( Z_{i}=z_{i}) .
\end{align*}
To conclude \eqref{eq:main_inner2}, note that
\begin{equation*}
\sum\nolimits_{a^{( n) }}d{\bf R}(A^{( n) }=a^{( n) }|W^{( n) };\beta _{0},\beta _{1}) ~=~\sum\nolimits_{a^{( n) }}P( A^{( n) }=a^{( n) }|S^{( n) }=\{ S( z_{i}) \} _{i=1}^{n}) ~=~1 
\end{equation*}
and
\begin{align*}
&\sum_{d_{0,i},d_{1,i}\in \{ 0,1\} }\int_{y_{1,i},y_{0,i},z_{i}} 
\left[
{\scriptsize\begin{array}{c}
f_{0,1}( y_{1,i},z_{i})  I[ y_{0,i}=Y_{L}\beta _{0}+Y_{H}( 1-\beta _{0}) ]  \pi _{D( 0)}( z_{i})  I[ d_{0,i}=1,d_{1,i}=1] \\
+\lambda _{1}( y_{1,i},z_{i})  \lambda _{0}( y_{0,i},z_{i})  ( \pi _{D( 1)}( z_{i}) -\pi _{D( 0)}( z_{i}) )  I[ d_{0,i}=0,d_{1,i}=1]+ \\
 I[ y_{1,i}=Y_{L}\beta _{1}+Y_{H}( 1-\beta _{1}) ]  f_{1,0}( y_{0,i},z_{i})  ( 1-\pi _{D( 1)}( z_{i}) )  I[ d_{0,i}=0,d_{1,i}=0 ]
\end{array}}
\right]\\
&\times dP(Z_{i}=z_{i}) \\
&\overset{(1)}{=}\sum_{d_{0,i},d_{1,i}\in \{ 0,1\} }\left[
{\scriptsize\begin{array}{c}
\int_{z_{i}}\pi _{D( 0)}( z_{i}) dP(Z_{i}=z_{i})  I[ d_{0,i}=1,d_{1,i}=1] \\
+ \int_{z_{i}}( \pi _{D( 1)}( z_{i}) -\pi _{D( 0)}( z_{i}) ) dP(Z_{i}=z_{i}) dz_{i} I[ d_{0,i}=0,d_{1,i}=1] \\
+\int_{z_{i}}( 1-\pi _{D( 1)}( z_{i}) ) dP(Z_{i}=z_{i}) dz_{i} I[ d_{0,i}=0,d_{1,i}=0 ]
\end{array}}
\right]~\overset{(2)}{=}~1,
\end{align*}
where (1) holds by Lemma \ref{lem:QnMap} and \eqref{eq:f_def} and (2) by Lemma \ref{lem:QnMap} and Assumption \ref{ass:1}(b). 

\noindent \underline{Part 2.} We now derive $d{\bf R}(X^{( n) };\beta _{0},\beta _{1}) $ and verify that it coincides with $dP(X^{( n) }) $.

First, note that 
\begin{align} 
d{\bf R}(A^{( n) }=a^{( n) },Z^{( n) }=z^{( n) };\beta _{0},\beta _{1}) &\overset{(1)}{=}P(A^{( n) }=a^{( n) }|S^{( n) }=\{ S( z_{i}) \} _{i=1}^{n}) \prod\limits_{i=1}^{n}dP(Z_{i}=z_{i})\notag \\
&\overset{(2)}{=}dP(A^{( n) }=a^{( n) },Z^{( n) }=z^{( n) }) ,\label{eq:distEqual_1}
\end{align}
where (1) holds by \eqref{eq:Rn_dist_defn} and (2) by Assumption \ref{ass:2}(a).

Second, for any $( a^{( n) },z^{( n) }) $ s.t.\ $d{\bf R}(A^{( n) }=a^{( n) },Z^{( n) }=z^{( n) };\beta _{0},\beta _{1}) >0$, \eqref{eq:Rn_dist_defn} implies that
\begin{align} 
  \text{\small${\bf R}(D^{( n) }=d^{( n) }|Z^{( n) }=z^{( n) },A^{( n) }=a^{( n) },S^{( n) };\beta _{0},\beta _{1})=$} 
\text{\small $\prod_{i=1}^{n} \pi _{D( a_{i}) }( z_{i}) ^{d_{i}}( 1-\pi _{D( a_{i})}( z_{i}) ) ^{1-d_{i}}$ }, \label{eq:main_inner4}
\end{align}
Also, for any $( a^{( n) },z^{( n) }) $ s.t.\ $dP(A^{( n) }=a^{( n) },Z^{( n) }=z^{( n) }) >0$,
\begin{align}
& P(D^{( n) }=d^{( n) }|Z^{( n)}=z^{( n) },A^{( n) }=a^{( n) },S^{( n) }) \notag \\
& \overset{(1)}{=}~
\prod_{i=1}^{n}\sum_{b_{i}\in \{ 0,1\} }\left[
\begin{array}{c}
P( D_{i}( a_{i}) =d_{i},D_{i}( 1-a_{i}) =b_{i}|Z_{i}=z_{i}) I[ a_{i}=1]\\
+P( D_{i}( a_{i}) =d_{i},D_{i}( 1-a_{i}) =b_{i}|Z_{i}=z_{i}) I[ a_{i}=0]
\end{array}
\right]  \notag \\
& \overset{(2)}{=}~
\prod_{i=1}^{n} \pi _{D( a_{i})}( z_{i}) ^{d_{i}}( 1-\pi _{D( a_{i})}( z_{i}) ) ^{1-d_{i}} ,  \label{eq:main_inner5}
\end{align}
where (1) holds by Assumption \ref{ass:1} and (2) by Lemma \ref{lem:QnMap}. Note that \eqref{eq:distEqual_1}, \eqref{eq:main_inner4}, and \eqref{eq:main_inner5} imply that 
\begin{equation}
{\bf R}(D^{( n) },A^{( n) },Z^{( n) };\beta _{0},\beta _{1}) ~=~P(D^{( n) },A^{( n) },Z^{( n) }) . \label{eq:distEqual_2}
\end{equation}

Third, for any $( d^{( n) },a^{( n) },z^{( n) }) $ s.t.\ $d{\bf R}( D^{( n) }=d^{( n) },A^{( n) }=a^{( n) },Z^{( n) }=z^{( n) };\beta _{0},\beta _{1}) >0$. Then, \eqref{eq:Rn_dist_defn} implies
{\begin{align}
&d{\bf R}(D^{( n) }=d^{( n) },A^{( n) }=a^{( n) },Z^{( n) }=z^{( n) };\beta _{0},\beta _{1})= \notag \\
&\Big( \prod_{i=1}^{n}d{\bf R}(D_{i}( a_{i}) =d_{i}|Z_{i}=z_{i};\beta _{0},\beta _{1}) dP(Z_{i}=z_{i}) \Big)  P(A^{( n) }=a^{( n) }|S^{( n) }=\{ S( z_{i}) \} _{i=1}^{n}).\label{eq:main_inner5b}
\end{align}}
Note that \eqref{eq:main_inner5b} and $d{\bf R}( D^{( n) }=d^{( n) },A^{( n) }=a^{( n) },Z^{( n) }=z^{( n) };\beta _{0},\beta _{1}) >0$ implies that
\begin{equation}
    d{\bf R}(D_{i}( a_{i})=d_{i}|Z_{i}=z_{i};\beta _{0},\beta _{1})~>~0~~~\text{for all}~i=1,\dots,n.\label{eq:main_inner6}
\end{equation}
For $( d^{( n) },a^{( n) },z^{( n) }) $ mentioned earlier, we have 
\begin{align}
& d{\bf R}(Y^{( n) }=y^{( n) }|D^{( n) }=d^{( n) },A^{( n) }=a^{( n) },Z^{( n) }=z^{( n) };\beta _{0},\beta _{1})\notag\\
& \overset{(1)}{=}~\prod\limits_{i=1}^{n}\left[
\begin{array}{c}
\frac{d{\bf R}( Y_{i}( 1) =y_{i},D_{i}( 1) =1|Z_{i}=z_{i};\beta _{0},\beta _{1}) }{d{\bf R}( D_{i}( 1) =1|Z_{i}=z_{i};\beta _{0},\beta _{1}) }I[ d_{i}=1,a_{i}=1 ] \\
+\frac{d{\bf R}( Y_{i}( 0) =y_{i},D_{i}( 0) =0|Z_{i}=z_{i};\beta _{0},\beta _{1}) }{d{\bf R}( D_{i}( 0) =0|Z_{i}=z_{i};\beta _{0},\beta _{1}) }I[ d_{i}=0,a_{i}=0 ] \\
+\frac{d{\bf R}( Y_{i}( 1) =y_{i},D_{i}( 0) =1|Z_{i}=z_{i};\beta _{0},\beta _{1}) }{d{\bf R}( D_{i}( 0) =1|Z_{i}=z_{i};\beta _{0},\beta _{1}) }I[ d_{i}=1,a_{i}=0 ]\\
+\frac{d{\bf R}( Y_{i}( 0) =y_{i},D_{i}( 1) =0|Z_{i}=z_{i};\beta _{0},\beta _{1}) }{d{\bf R}( D_{i}( 1) =0|Z_{i}=z_{i};\beta _{0},\beta _{1}) }I[ d_{i}=0,a_{i}=1 ]
\end{array}
\right] \notag \\
& \overset{(2)}{=}~\prod\limits_{i=1}^{n}\left[
\begin{array}{c}
\frac{{\scriptsize\left(
\begin{array}{c}
d{\bf R}( Y_{i}( 1) =y_{i},D_{i}( 0) =0,D_{i}( 1) =1|Z_{i}=z_{i};\beta _{0},\beta _{1}) \\
+d{\bf R}( Y_{i}( 1) =y_{i},D_{i}( 0) =1,D_{i}( 1) =1|Z_{i}=z_{i};\beta _{0},\beta _{1})
\end{array}
\right)} }{ d{\bf R}( D_{i}( 0) =0,D_{i}( 1) =1|Z_{i}=z_{i};\beta _{0},\beta _{1}) +d{\bf R}( D_{i}( 0) =1,D_{i}( 1) =1|Z_{i}=z_{i};\beta _{0},\beta _{1}) }I[ d_{i}=1,a_{i}=1] \\
+\frac{\left(
{\scriptsize\begin{array}{c}
d{\bf R}( Y_{i}( 0) =y_{i},D_{i}( 0) =0,D_{i}( 1) =1|Z_{i}=z_{i};\beta _{0},\beta _{1}) \\+d{\bf R}( Y_{i}( 0) =y_{i},D_{i}( 0) =0,D_{i}( 1) =0|Z_{i}=z_{i};\beta _{0},\beta _{1})
\end{array}}
\right) }{ d{\bf R}( D_{i}( 0) =0,D_{i}( 1) =1|Z_{i}=z_{i};\beta _{0},\beta _{1}) +d{\bf R}( D_{i}( 0) =0,D_{i}( 1) =0|Z_{i}=z_{i};\beta _{0},\beta _{1}) }I[ d_{i}=0,a_{i}=0] \\
+\frac{d{\bf R}( Y_{i}( 1) =y_{i},D_{i}( 0) =1,D_{i}( 1) =1|Z_{i}=z_{i};\beta _{0},\beta _{1}) }{ d{\bf R}( D_{i}( 0) =1,D_{i}( 1) =1|Z_{i}=z_{i};\beta _{0},\beta _{1}) }I[ d_{i}=1,a_{i}=0] \\
+\frac{d{\bf R}( Y_{i}( 0) =y_{i},D_{i}( 0) =0,D_{i}( 1) =0|Z_{i}=z_{i};\beta _{0},\beta _{1}) }{ d{\bf R}( D_{i}( 0) =0,D_{i}( 1) =0|Z_{i}=z_{i};\beta _{0},\beta _{1}) }I[ d_{i}=0,a_{i}=1]
\end{array}
\right] \notag \\
& \overset{(3)}{=}~\prod\limits_{i=1}^{n}f_{a_{i},d_{i}}( y_{i},z_{i}) ,
\end{align}
where (1) holds by \eqref{eq:Rn_dist_defn}, \eqref{eq:main_inner6}, and the fact that $\{ D_{i}=d,A_{i}=a\} $ implies that $d{\bf R}(D_{i}( a) =d|Z_{i}=z;\beta _{0},\beta _{1}) =\pi _{D( a)}( z_{i}) ^{d_{i}}( 1-\pi _{D( a) }( z_{i}) ) ^{1-d_{i}}>0$ for all $i=1,\ldots ,n$, (2) by Assumption \ref{ass:1}(b), and (3) by \eqref{eq:Rn_dist_defn} and \eqref{eq:mu_defn}.

Next, consider any $( d^{( n) },a^{( n) },z^{( n) }) $ s.t. $dP( D^{( n) }=d^{( n) },A^{( n) }=a^{( n) },Z^{( n) }=z^{( n) }) >0$. Consider any arbitrary $ i=1,\ldots ,n$ arbitrarily. Note that $dP(D^{( n) }=d^{( n) },A^{( n) }=a^{( n) },Z^{( n) }=z^{( n) }) >0$ implies that
\begin{equation}
dP(D_{i}=d_{i},A_{i}=a_{i},Z^{( n) }=z^{( n) }) ~>~0.\label{eq:main_inner7}
\end{equation}%
Also, note that
\begin{align}
&dP(D_{i}=d_{i},A_{i}=a_{i},Z^{( n) }=z^{( n) }) \notag\\
&\overset{(1)}{=}~
P(D_{i}( a_{i}) =d_{i}|Z^{( n) }=z^{( n) }) P(A_{i}=a_{i}|S^{( n) }=\{ S( z_{i}) \} _{i=1}^{n}) dP( Z^{( n) }=z^{( n) }) \notag \\
&\overset{(2)}{=}~P(D_{i}( a_{i}) =d_{i}|Z_{i}=z_{i}) P(A_{i}=a_{i}|S^{( n) }=\{ S( z_{i}) \} _{i=1}^{n}) dP( Z^{( n) }=z^{( n) }) \notag \\
&\overset{(3)}{=}~\pi _{D( a_{i}) }( z_{i}) ^{d_{i}}( 1-\pi _{D( a_{i}) }( z_{i}) ) ^{1-d_{i}} P( A_{i}=a_{i}|S^{( n) }=\{ S( z_{i}) \} _{i=1}^{n}) dP( Z^{( n) }=z^{( n) }) , \label{eq:main_inner8}
\end{align}
where (1) holds by Assumption \ref{ass:2}(a), (2) by the i.i.d.\ condition in Assumption \ref{ass:1}, and (3) by \eqref{eq:pi_def}. Then, \eqref{eq:main_inner7} and \eqref{eq:main_inner8} imply that
\begin{equation}
\pi _{D( a_{i})}( z_{i}) ^{d_{i}}( 1-\pi _{D( a_{i}) }( z_{i}) ) ^{1-d_{i}}~>~0~~\text{for all}~i=1,\ldots ,n.\label{eq:main_inner9}
\end{equation}
Then,
\begin{align}
& dP(Y^{( n) }=y^{( n) }|D^{( n) }=d^{( n) },A^{( n) }=a^{( n) },Z^{( n) }=z^{( n) }) \notag \\
& \overset{(1)}{=}~dP(\{ Y_{i}( d_{i}) =y_{i}\} _{i=1}^{n}|\{ D_{i}( a_{i}) =d_{i}\} _{i=1}^{n},S^{( n) }=\{ S( z_{i}) \} _{i=1}^{n},Z^{( n) }=z^{( n) }) \notag \\
& \overset{(2)}{=}~\textstyle\prod_{i=1}^{n}dP(Y_{i}( d_{i}) =y_{i}|D_{i}( a_{i}) =d_{i},Z_{i}=z_{i}) \notag \\
& \overset{(3)}{=}~\textstyle\prod_{i=1}^{n}f_{a_{i},d_{i}}( y_{i},z_{i}) , \label{eq:main_inner10}
\end{align}
where (1) by Assumption \ref{ass:2}(a), (2) by the i.i.d.\ condition in Assumption \ref{ass:1}, and (3) by \eqref{eq:f_def} and \eqref{eq:main_inner9}. To conclude the proof, note that \eqref{eq:distEqual_2}, \eqref{eq:main_inner8}, and \eqref{eq:main_inner10} imply that ${\bf R}(X^{( n) };\beta _{0},\beta _{1})={\bf P}(X^{( n) })$,
as desired.

\noindent \underline{Part 3.}  
The desired results follow immediately from showing that:
\begin{align}
    E_{{\bf R}}[ Y_{i}( 1) ;\beta _{0},\beta _{1}] ~&=~ E[ Y( 1) D( 1) +( Y_{L}\beta _{1}+Y_{H}(1-\beta _{1})) ( 1-D( 1) ) ]\label{eq:main_inner11}\\
     E_{{\bf R}}[ Y_{i}( 0) ;\beta _{0},\beta _{1}] ~&=~ E [ Y( 0) ( 1-D( 0) ) +( Y_{L}\beta _{0}+Y_{H}( 1-\beta _{0}) ) D( 0) ] .\label{eq:main_inner12}
\end{align}
We only show \eqref{eq:main_inner11}, as \eqref{eq:main_inner12} follows from an analogous argument. To this end, 
\begin{align}
&E_{{\bf R}}[ Y_{i}( 1) ;\beta _{0},\beta _{1}]\notag\\
&~=~\textstyle\sum\nolimits_{a^{( n) }}\int_{w^{( n) }}y_{1,i}d{\bf R}((W^{(n)},A^{(n)})=(w^{(n)},a^{(n)});\beta _{0},\beta _{1}) dw^{( n) } \notag\\
&~\overset{(1)}{=}~
\text{\scriptsize $\int_{z_{i},y_{1,i}}\sum_{d_{0,i},d_{1,i}}y_{1,i}$}\left[
{\scriptsize\begin{array}{c}
f_{0,1}(y_{1,i},z_{i})  \pi _{D(0)}(z_{i})  I[d_{0,i}=d_{1,i}=1]\\
+\lambda _{1}(y_{1,i},z_{i})  (\pi _{D(1)}(z_{i})-\pi _{D(0)}(z_{i}))  I[d_{0,i}=0,d_{1,i}=1]+ \\
I[y_{1,i}=Y_{L}\beta _{1}+Y_{H}(1-\beta _{1})] 
(1-\pi _{D(1)}(z_{i}))  I[d_{0,i}=d_{1,i}=0]
\end{array}}
\right] \text{\scriptsize $dy_{1,i}dP(Z_{i}=z_{i})$}  \notag \\
&~\overset{(2)}{=}~
E[ Y_{i}( 1) D_{i}( 1) +( Y_{L}\beta _{1}+Y_{H}(1-\beta _{1})) ( 1-D_{i}( 1) ) ] \notag \\
&~\overset{(3)}{=}~E[ Y( 1) D( 1) +( Y_{L}\beta _{1}+Y_{H}(1-\beta _{1})) ( 1-D( 1) ) ] ,\notag 
\end{align}
as desired, where (1) holds by \eqref{eq:Rn_dist_defn}, $\sum_{a^{( n) }}P(A^{(n)}=a^{(n)}|S^{(n)}=\{S(z_{i})\}_{i=1}^{n})=1$, for all $ j=1,\ldots ,n$ with $j\not=i$,
\begin{equation*}
\text{\scriptsize $\int_{z_{j},y_{1,j}}$}\left[ 
{\scriptsize\begin{array}{c}
f_{0,1}(y_{1,j},z_{j})  I[y_{0,j}=Y_{L}\beta _{0}+Y_{H}(1-\beta _{0})]  \pi _{D(0)}(z_{j})  I[d_{0,j}=d_{1,j}=1] +\\
\lambda _{1}(y_{1,j},z_{j})  \lambda _{0}(y_{0,j},z_{j})  (\pi _{D(1)}(z_{j})-\pi _{D(0)}(z_{j}))  I[d_{0,j}=0,d_{1,j}=1]+ \\
I[y_{1,j}=Y_{L}\beta _{1}+Y_{H}(1-\beta _{1})]  f_{1,0}(y_{0,j},z_{j})  (1-\pi _{D(1)}(z_{j}))  I[d_{0,j}=d_{1,j}=0]
\end{array}}
\right] \text{\scriptsize $dy_{1,j}dP(Z_{j}=z_{j})$}=1,
\end{equation*}
and $\sum_{y_{0,i}}I[y_{0,i}=Y_{L}\beta _{0}+Y_{H}(1-\beta _{0})]=1$, $\int_{y_{0,i}}\lambda _{0}(y_{0,i},z_{i})dy_{0,i}=1$ if $(\pi _{D(1)}(z_{i})-\pi _{D(0)}(z_{i}))>0$, and $ \int_{y_{0,i}}f_{1,0}(y_{0,i},z_{i})dy_{0,i}=1$ if $(1-\pi _{D(1)}(z_{j}))>0$, (2) by \eqref{eq:pi_def} and \eqref{eq:f_def}, and (3) by the i.i.d.\ condition in Assumption \ref{ass:1}.
\end{proof}

\subsection{ATT}\label{sec:appendix1ATT}

\begin{proof}[Proof of Theorem \ref{thm:mainATT}]
Fix $i=1,\dots ,n$ arbitrarily throughout this proof. We begin by showing that $W_{i}\perp A_{i}|S_{i}$, where $W_{i}=(Y_{i}(1),Y_{i}(0),D_{i}(1),D_{i}(0),Z_{i})$. For any fixed $s\in \mathcal{S}$, and $w\in \mathbb{R}^{5}$ and $a\in \{0,1\}$,
\begin{align*}
P[W_{i}\leq w,A_{i}=a|S_{i}=s] &~=~P[P[W_{i}\leq w,A_{i}=a|S^{(n)}]|S_{i}=s] \\
&~\overset{(1)}{=}~P[P[W_{i} \leq w|S^{(n)}]P[A_{i}=a|S^{(n)}]|S_{i}=s] \\
&~\overset{(2)}{=}~P[P[W_{i} \leq w|S_{i}=s]P[A_{i}=a|S^{(n)}]|S_{i}=s] \\
&~=~P[W_{i}\leq w|S_{i}=s]P[A_{i}=a|S_{i}=s],
\end{align*}%
as desired, where (1) holds by Assumption \ref{ass:2}(a), and (2) by the i.i.d.\ condition in Assumption \ref{ass:1}.

Second, consider the following derivation.
\begin{align}
G&~\overset{(1)}{=}~E\left[ \left( \frac{D_{i}A_{i}}{P(A_{i}=1|S_{i})}-\frac{D_{i}(1-A_{i})}{1-P(A_{i}=1|S_{i})}\right) P(A_{i}=1|S_{i})+\frac{D_{i}(1-A_{i})}{1-P(A_{i}=1|S_{i})}\right]  \notag \\
&~\overset{(2)}{=}~E\left[ \left( \frac{D_{i}(1) A_{i}}{P(A_{i}=1|S_{i})}-\frac{D_{i}(0) (1-A_{i})}{1-P(A_{i}=1|S_{i})}\right) P(A_{i}=1|S_{i})+\frac{D_{i}(0) (1-A_{i})}{1-P(A_{i}=1|S_{i})}\right]  \notag \\
&~=E\left[ (D_{i}(1) -D_{i}(0) )A_{i}+D_{i}(0) \right]  \notag \\
&~\overset{(3)}{=}~E\left[ E[D_{i}(1)|S_{i}]P(A_{i}=1|S_{i})+E[D_{i}(0) |S_{i}](1-P(A_{i}=1|S_{i}))\right]  \notag \\
&~\overset{(4)}{=}~E\left[ E[D(1) |S]P(A=1|S_{j})+E[D(0) |S](1-P(A=1|S))\right]  \notag \\
&~=~E\left[ (D(1) -D(0) )A+D(0) \right] , \label{eq:mainATT_proof_1}
\end{align}%
where (1) holds by \eqref{eq:bounds_ATT_G}, (2) by the fact that $D_{i}=D_{i}(A_{i})$, (3) by $W_{i}\perp A_{i}|S_{i}$, and (4) by Assumption \ref{ass:2}(e) and the i.i.d.\ condition in Assumption \ref{ass:1}. Note that \eqref{eq:mainATT_proof_1} shows that $G=E [ (D_{i}(1) -D_{i}(0) )A_{i}+D_{i}(0) ]$, and that this expression is the same for all $i=1,\ldots ,n$.

Third, for $\bar{Y}\in \left\{ Y_{L},Y_{H}\right\} $, consider the following argument.
\begin{align}
& E\left[ \left( \frac{Y_{i}A_{i}}{P(A_{i}=1|S_{i})}-\frac{Y_{i}(1-A_{i})}{1-P(A_{i}=1|S_{i})}\right) P(A_{i}=1|S_{i})+\frac{(Y_{i}(1) -\tilde{Y})D_{i}(0) (1-A_{i})}{1-P(A_{i}=1|S_{i})}\right] \notag \\
& ~\overset{(1)}{=}~E\left[ 
\begin{array}{c}
Y_{i}(1) D_{i}(1) A_{i}+Y_{i}(0) (1-D_{i}(1) )A_{i} \\ 
-\frac{(Y_{i}(1)D_{i}(0)+Y_{i}(0)(1-D_{i}(0)))(1-A_{i})}{1-P(A_{i}=1|S_{i})}P(A_{i}=1|S_{i})+\frac{(Y_{i}(1) -\tilde{Y})D_{i}(0) (1-A_{i})}{1-P(A_{i}=1|S_{i})}
\end{array}
\right] \notag \\
& ~\overset{(2)}{=}~E\left[ 
\begin{array}{c}
E[Y_{i}(1) D_{i}(1) +Y_{i}(0) (1-D_{i}(1) )|S_{i}]P(A_{i}=1|S_{i})+E[(Y_{i}(1) -\tilde{Y})D_{i}(0) |S_{i}] \\ 
-\frac{E[(Y_{i}(1)D_{i}(0)+Y_{i}(0)(1-D_{i}(0)))|S_{i}]E[(1-A_{i})|S_{i}]}{ 1-P(A_{i}=1|S_{i})}P(A_{i}=1|S_{i})
\end{array}
\right] \notag \\
& ~\overset{(3)}{=}~E\left[ E[(Y_{i}(1) -Y_{i}(0) )(D_{i}(1) -D_{i}(0))|S_{i}]P(A_{i}=1|S_{i})+E[(Y_{i}(1) -\tilde{Y})D_{i}(0) |S_{i}]\right] \notag \\
& ~\overset{(4)}{=}~E\left[ E[(Y(1) -Y(0) )(D(1) -D(0))|S]P(A=1|S)+E[(Y(1) -\tilde{Y})D(0) |S]\right] \notag \\
& ~\overset{(5)}{=}~E[(Y(1)-Y(0))(D(1)-(0))A+(Y(1) -\tilde{Y})D(0) ], \label{eq:mainATT_proof_2}
\end{align}
where (1) holds by $Y_{i}=Y_{i}(D_{i})$ and $D_{i}=D_{i}(A_{i})$, (2) and (5) by $W_{i}\perp A_{i}|S_{i}$, (3) by $ P(A_{i}=1|S_{i})=E[E[A_{i}|S^{(n)}]|S_{i}]\in (0,1)$ (which, in turn, follows from Assumption \ref{ass:1}(a) and \ref{ass:2}(b), and \citet[Theorem D, page 104]{halmos:1974}, and (4) by Assumption \ref{ass:2}(e) and the i.i.d.\ condition in Assumption \ref{ass:1}. By \eqref{eq:mainATT_proof_1} and \eqref{eq:mainATT_proof_2}, we conclude that the expressions of $\upsilon_{L}(P)$ and $\upsilon_{H}(P)$ are the same for all $i=1,\ldots ,n$, and they satisfy 
\begin{align}
G\upsilon_{L}(\mathbf{P}) &~=~E[(Y_{i}(1)-Y_{i}(0))(D_{i}(1)-D_{i}(0))A_{i}+(Y_{i}(1) -Y_{H})D_{i}(0) ] \notag\\
G\upsilon_{H}(\mathbf{P}) &~=~E[(Y_{i}(1)-Y_{i}(0))(D_{i}(1)-D_{i}(0))A_{i}+(Y_{i}(1) -Y_{L})D_{i}(0) ]. \label{eq:mainATT_proof_3}
\end{align}

Fourth, we show that the ATT $\upsilon$ satisfies
\begin{equation}
G\upsilon~=~E[(Y_{i}(1)-Y_{i}(0))(D_{i}(1)-D_{i}(0))A_{i}+(Y_{i}(1) -Y_{i}(0))D_{i}(0) ]. \label{eq:mainATT_proof_4}
\end{equation}%
To this end, consider the following derivation.
\begin{equation}
P(D_{i}=1)~=~E[D_{i}]~\overset{(1)}{=}~E[(D_{i}(1) -D_{i}(0))A_{i}+D_{i}(0)A_{i}]~\overset{(2)}{=}~G, \label{eq:mainATT_proof_5}
\end{equation}%
where (1) holds by $D_{i}=D_{i}(A_{i})$ and (2) by \eqref{eq:mainATT_proof_1}. Note that Assumption \ref{ass:2}(d) and $P(A_{i}=1|S_{i})=E[E[A_{i}|S^{(n)}]|S_{i}]>0$ implies that $G=P(D_{i}=1)>0$. Then, we have
\begin{align*}
Gv &~\overset{(1)}{=}~P(D_{i}=1)E[Y_{i}(1)-Y_{i}(0)|D_{i}=1] \\
&~=~E[(Y_{i}(1)-Y_{i}(0))(D_{i}(1) -D_{i}(0))A_{i}+(Y_{i}(1)-Y_{i}(0))D_{i}(0) ],
\end{align*}
as desired, where (1) holds by \eqref{eq:mainATT_proof_5}.

Next, we establish that $[\upsilon_{L}(\mathbf{P}),\upsilon_{H}(\mathbf{P})]$ are the sharp bounds for the ATT. By \eqref{eq:mainATT_proof_3}, \eqref{eq:mainATT_proof_4}, \eqref{eq:mainATT_proof_5}, and $ Y_{i}(0)\in \lbrack Y_{L},Y_{H}]$, it follows immediately that $v(\mathbf{P} )\in [\upsilon_{L}(\mathbf{P}),\upsilon_{H}(\mathbf{P})]$. That is, $[\upsilon_{L}(\mathbf{ P}),\upsilon_{H}(\mathbf{P})]$ is an outer identified set. In turn, Lemma \ref{lem:main_innerATT} proposes a hypothetical distribution for all underlying variables that: (i) is observationally equivalent to the data distribution $\mathbf{P}$ and (ii) produces ATTs that span $[\upsilon _{L}(\mathbf{P}),\upsilon _{H}(\mathbf{P})]$. As a corollary, $[\upsilon _{L}(\mathbf{P}),\upsilon _{H}(\mathbf{P})]$ is a subset of the identified set. By combining these observations, the desired result follows.
\end{proof}

\begin{lemma}\label{lem:main_innerATT}
Let Assumptions \ref{ass:1} and \ref{ass:2}(a)-(e) hold.
For any $\beta _{0},\beta _{1}\in [ 0,1] $ and $w^{(n)} = \{ (y_{1,i},y_{0,i},d_{1,i},d_{0,i},z_{i})\} _{i=1}^{n}$, let $\mathbf{R}(\cdot;\beta_0,\beta_1)$ denote the measure for $( W^{( n) },A^{( n) })$ defined in Lemma \ref{lem:main_inner}. Then, for any $\beta _{0},\beta _{1}\in [ 0,1] $,

\begin{enumerate}
\item $\mathbf{R}(( W^{( n) },A^{( n) });\beta _{0},\beta _{1}) $ is a probability measure.

\item $\mathbf{R}(( W^{( n) },A^{( n) });\beta _{0},\beta _{1}) $ induces a  probability distribution for $X^{( n)}$ that is observationally equivalent to $\mathbf{P}(X^{( n) }) $, i.e., $\mathbf{R}( X^{( n) };\beta _{0},\beta _{1}) =\mathbf{P}(X^{( n) }) $.

\item For every $i=1,\dots ,n$, 
\begin{align*}
& E_{\mathbf{R}}[Y_{i}(1)-Y_{i}(0)|D_{i}=1;\beta _{0},\beta _{1}]= \\
& \frac{E[(Y_{i}(1)-Y_{i}(0))(D_{i}(1)-D_{i}(0))A_{i}+(Y_{i}(1)-Y_{L}\beta
_{1}-Y_{H}(1-\beta _{1})))D_{i}(0)]}{E[(D_{i}(1)-D_{i}(0))A_{i}+D_{i}(0)]}.
\end{align*}
Therefore, 
\begin{equation*}
\underset{\beta _{0},\beta _{1}\in [ 0,1]}{\bigcup }E_{\mathbf{R}}[Y_{i}(1)-Y_{i}(0)|D_{i}=1;\beta _{0},\beta _{1}]~=~[\upsilon _{L}(\mathbf{P}),\upsilon _{H}(\mathbf{P})]. 
\end{equation*}
\end{enumerate}
\end{lemma}
\begin{proof}
Parts 1-2 were shown in Lemma \ref{lem:main_inner}, so we focus on part 3.

\noindent \underline{Part 3.} It suffices to show that 
\begin{align}
& E_{\mathbf{R}}[Y_{i}(1)|D_{i}=1;\beta _{0},\beta _{1}]=\tfrac{E\left[ Y_{i}(1)[(D_{i}(1)-D_{i}(0))A_{i}+D_{i}(0)]\right] }{E[(D_{i}(1)-D_{i}(0))A_{i}+D_{i}(0)]}  \label{eq:ATT1} \\
& E_{\mathbf{R}}[Y_{i}(0)|D_{i}=1;\beta _{0},\beta _{1}]=\tfrac{E\left[ Y_{i}(0)(D_{i}(1)-D_{i}(0))A_{i}+(Y_{L}\beta _{1}+Y_{H}(1-\beta _{1}))D_{i}(0)\right] }{E[(D_{i}(1)-D_{i}(0))A_{i}+D_{i}(0)]}.
\label{eq:ATT0}
\end{align}

In turn, \eqref{eq:ATT1} and \eqref{eq:ATT0} follow from showing
\begin{align}
E_{\mathbf{R}}[D_{i};\beta _{0},\beta _{1}] &=E[(D_{i}(1)-D_{i}(0))A_{i}+D_{i}(0)]  \label{eq:ATT01_1} \\
E_{\mathbf{R}}[Y_{i}(1)D_{i};\beta _{0},\beta _{1}]
&=E[Y_{i}(1)[(D_{i}(1)-D_{i}(0))A_{i}+D_{i}(0)]] \label{eq:ATT01_2}\\
E_{\mathbf{R}}[Y_{i}(0)D_{i};\beta _{0},\beta _{1}] &=E\left[ Y_{i}(0)(D_{i}(1)-D_{i}(0))A_{i}+(Y_{L}\beta _{1}+Y_{H}(1-\beta _{1}))D_{i}(0)\right] .\label{eq:ATT01_3}
\end{align}

To show \eqref{eq:ATT01_1}, consider the following derivation.
\begin{align*}
&=E_{\mathbf{R}}[D_{i};\beta _{0},\beta _{1}]=E_{\mathbf{R}}[D_{i}(1)A_{i}+D_{i}(0)(1-A_{i});\beta _{0},\beta _{1}] \\
&=\sum\nolimits_{a^{(n)}}\int_{w^{(n)}}(d_{1,i}a_{i}+d_{0,i}(1-a_{i}))d\mathbf{R}((W^{(n)},A^{(n)})=(w^{(n)},a^{(n)});\beta _{0},\beta _{1})dw^{(n)}
\\
&\overset{(1)}{=}\int_{w^{(n)}}\left[ 
\begin{array}{c}
(d_{1,i}P(A_{i}=1|S^{(n)}=\{S(z_{i})\}_{i=1}^{n})+d_{0,i}P(A_{i}=0|S^{(n)}=\{S(z_{i})\}_{i=1}^{n})) \\ 
\times \prod\limits_{s=1}^{n}\left[ {\scriptsize 
\begin{array}{c}
1[y_{0,s}=Y_{L}\beta _{0}+Y_{H}(1-\beta _{0})]f_{0,1}(y_{1,s},z_{s})\pi_{D(0)}(z_{s})1[d_{0,s}=d_{1,s}=1] \\ 
+\lambda _{1}(y_{1,s},z_{s})\lambda _{0}(y_{0,s},z_{s})(\pi_{D(1)}(z_{s})-\pi _{D(0)}(z_{s}))1[d_{0,s}=0,d_{1,s}=1] \\ 
+1[y_{1,s}=Y_{L}\beta _{1}+Y_{H}(1-\beta _{1})]f_{1,0}(y_{0,s},z_{s})(1-\pi_{D(1)}(z_{s}))\times  \\ 
1[d_{0,s}=d_{1,s}=0]
\end{array}
}\right] \text{{\scriptsize $dP(Z_{s}=z_{s})$}}
\end{array}
\right] dw^{(n)} \\
&\overset{(2)}{=}\int_{z^{(n)}}\left[ 
\begin{array}{c}
P(A_{i}=1|S^{(n)}=\{S(z_{i})\}_{i=1}^{n})\pi _{D(1) }\left(z_{i}\right)  \\ 
+P(A_{i}=0|S^{(n)}=\{S(z_{i})\}_{i=1}^{n})\pi _{D(0)}(z_{i})
\end{array}
\right] \prod\limits_{s=1}^{n}\text{$dP(Z_{s}=z_{s})$}dz^{(n)} \\
&\overset{(3)}{=}\int_{z^{(n)}}\left[ 
\begin{array}{c}
P(A_{i}=1|S^{(n)}=\{S(z_{i})\}_{i=1}^{n})E(D_{i}(1)=1|Z_{i}=z_{i}) \\ 
+P(A_{i}=0|S^{(n)}=\{S(z_{i})\}_{i=1}^{n})E(D_{i}(0)=1|Z_{i}=z_{i})
\end{array}
\right] \prod\limits_{s=1}^{n}\text{$dP(Z_{s}=z_{s})$}dz^{(n)} \\
&\overset{(4)}{=}E[P(A_{i}=1|S^{(n)})E(D_{i}(1)|Z_{i})+P(A_{i}=0|S^{(n)})E(D_{i}(0)|Z_{i})] \\
&\overset{(5)}{=}E[P(A_{i}=1|S^{(n)})E(D_{i}(1)|S_{i})+P(A_{i}=0|S^{(n)})E(D_{i}(0)|S_{i})] \\
&=E[(D_{i}(1)-D_{i}(0))A_{i}+D_{i}(0)],
\end{align*}%
as desired, where (1) holds by \eqref{eq:Rn_dist_defn}, (2) by \eqref{eq:f_def} and \eqref{eq:mu_defn}, (3) by \eqref{eq:pi_def}, (4) by the i.i.d.\ condition in Assumption \ref{ass:1} and law of iterated expectations (LIE), and (5) by another use of LIE, which gives $ E[P(A_{i}=1|S^{(n)})|S_{i}]=P(A_{i}=1|S_{i})$.

To show \eqref{eq:ATT01_2}, consider the following derivation.
\begin{align*}
&E_{\mathbf{R}}[Y_{i}(1)D_{i};\beta _{0},\beta _{1}] \\
&=\text{{\footnotesize $\sum\nolimits_{a^{(n)}}%
\int_{w^{(n)}}y_{1,i}(d_{1,i}a_{i}+d_{0,i}(1-a_{i}))d\mathbf{R}((W^{(n)},A^{(n)})=(w^{(n)},a^{(n)});\beta _{0},\beta _{1})dw^{(n)}$}} \\
&\overset{(1)}{=}\int_{w^{(n)}}\left[ 
\begin{array}{c}
y_{1,i}\left[ 
\begin{array}{c}
d_{1,i}P(A_{i}=1|S^{(n)}=\{S(z_{i})\}_{i=1}^{n})+ \\ 
d_{0,i}P(A_{i}=0|S^{(n)}=\{S(z_{i})\}_{i=1}^{n})
\end{array}
\right] \times  \\ 
{\scriptsize \prod\limits_{s=1}^{n}}\left[ {\scriptsize 
\begin{array}{c}
f_{0,1}(y_{1,s},z_{s})1[y_{0,s}=Y_{L}\beta _{0}+Y_{H}(1-\beta _{0})]\pi_{D(0)}(z_{s})1[d_{0,s}=d_{1,s}=1] \\ 
+\lambda _{1}(y_{1,s},z_{s})\lambda _{0}(y_{0,s},z_{s})(\pi_{D(1)}(z_{s})-\pi _{D(0)}(z_{s}))1[d_{0,s}=0,d_{1,s}=1] \\ 
+1[y_{1,s}=Y_{L}\beta _{1}+Y_{H}(1-\beta _{1})]\times  \\ 
f_{1,0}(y_{0,s},z_{s})(1-\pi _{D(1)}(z_{s}))1[d_{0,s}=d_{1,s}=0]
\end{array}
}\right] \text{{\scriptsize $dP(Z_{s}=z_{s})$ }}
\end{array}
\right] dw^{(n)} \\
&\overset{(2)}{=}\int_{z^{(n)}}\left[ 
\begin{array}{c}
\left[\begin{array}{c}
E\left[ Y_{i}(1) |Z_{i}=z_{i},D_{i}(1) =1\right]P(A_{i}=1|S^{(n)}=\{S(z_{i})\}_{i=1}^{n})\pi _{D(1) }\left(z_{i}\right)  \\ 
+E\left[ Y_{i}(1) |Z_{i}=z_{i},D_{i}(0) =1\right]P(A_{i}=0|S^{(n)}=\{S(z_{i})\}_{i=1}^{n})\pi _{D(0)}(z_{i})
\end{array}\right]
\\
\prod\limits_{s=1}^{n}\text{$dP(Z_{s}=z_{s})$}
\end{array}\right]dz^{(n)}\\
&\overset{(3)}{=}\int_{z^{(n)}}\left[ 
\begin{array}{c}
E\left[ Y_{i}(1) |Z_{i}=z_{i},D_{i}(1) =1\right]E(D_{i}(1)=1|Z_{i}=z_{i})\times\\
P(A_{i}=1|S^{(n)}=\{S(z_{i})\}_{i=1}^{n}) \\ 
+E\left[ Y_{i}(1) |Z_{i}=z_{i},D_{i}(0) =1\right]E(D_{i}(0)=1|Z_{i}=z_{i})\times\\
P(A_{i}=0|S^{(n)}=\{S(z_{i})\}_{i=1}^{n})
\end{array}
\right] \prod\limits_{s=1}^{n}\text{$dP(Z_{s}=z_{s})$}dz^{(n)} \\
&\overset{(4)}{=}
E\left[
\begin{array}{c}
E\left[ Y_{i}(1) |Z_{i},D_{i}\left(1\right) =1\right] P(A_{i}=1|S^{(n)})E(D_{i}(1)|Z_{i})\\
+E\left[ Y_{i}\left(1\right) |Z_{i},D_{i}(0) =1\right]P(A_{i}=0|S^{(n)})E(D_{i}(0)|Z_{i})
\end{array}
\right]\\
&\overset{(5)}{=}
E\left[
\begin{array}{c}
E\left[ Y_{i}(1) |S_{i},D_{i}\left(1\right) =1\right] P(A_{i}=1|S^{(n)})E(D_{i}(1)|S_{i})\\
+E\left[ Y_{i}\left(1\right) |S_{i},D_{i}(0) =1\right]P(A_{i}=0|S^{(n)})E(D_{i}(0)|S_{i})
\end{array}
\right]
\\
&=E[Y_{i}(1)[(D_{i}(1)-D_{i}(0))A_{i}+D_{i}(0)]],
\end{align*}
where (1) holds by \eqref{eq:Rn_dist_defn}, (2) by \eqref{eq:f_def} and \eqref{eq:mu_defn}, (3) by \eqref{eq:pi_def} and \eqref{eq:f_def}, (4) by the i.i.d.\ condition in Assumption \ref{ass:1} and LIE, and (5) by another use of LIE, which gives $E[P(A_{i}=1|S^{(n)})|S_{i}]=P(A_{i}=1|S_{i})$. To conclude, we note that \eqref{eq:ATT01_3} follows from an analogous argument to the one used to show \eqref{eq:ATT01_2}.
\end{proof}


\section{Appendix on inference}\label{sec:appendix2}

This appendix uses LLN and CLT to denote the ``law of large numbers for triangular arrays'' (e.g., \citet[page 59]{durrett:2019}) and ``Lyapunov central limit theorem for triangular arrays'' (e.g., \citet[page 148]{durrett:2019}), respectively.
Also, we use ``w.p.a.1 to denote ``with probability approaching one''.

Recall from the main text that $\mathcal{P}_1$ is the set of probabilities $\mathbf{P}$ generated by $(\mathbf{Q},\mathbf{G})$ satisfying Assumptions \ref{ass:1}(a)-(c) and \ref{ass:2}(a)-(c), and $\mathcal{P}_2$ is the set of probabilities $\mathbf{P}$ generated by $(\mathbf{Q},\mathbf{G})$ that satisfy Assumptions \ref{ass:1} and \ref{ass:2}(a)-(e). In addition, we define $\mathcal{P}_3$ as the set of probabilities $\mathbf{P}$ generated by $(\mathbf{Q},\mathbf{G})$ satisfying Assumptions \ref{ass:1}(a)-(c) and \ref{ass:2}(a)-(c) and (f), and $\mathcal{P}_4$ as the set of probabilities $\mathbf{P}$ generated by $(\mathbf{Q},\mathbf{G})$ that satisfy Assumptions \ref{ass:1} and \ref{ass:2}. Note that $\mathcal{P}_4 \subset  \mathcal{P}_j\subset \mathcal{P}_1$ for $j=2,3$.

For $X^{(n) }=\left( \left\{ Y_{i},D_{i},A_{i},Z_{i}\right\}:i=1,\ldots ,n\right) \sim \mathbf{P}$ and $s \in \mathcal{S}$, we define
\begin{align}
    n_D(s)  \equiv \sum_{i=1}^n I[ D_i = 1, S_i = s ]~~~\text{ and }~~~n_{AD}(s) \equiv \sum_{i=1}^n I[ A_i = 1, D_i = 1, S_i = s ]\label{eq:n_{AD}(s)}
\end{align}
and
\begin{align}
T_{n,1}& \equiv \Big( \frac{1}{\sqrt{n}}\sum_{i=1}^{n}I[D_{i}=d,A_{i}=a,S_{i}=s](Y_{i}(d)-\mu_{\mathbf{P}} (d,a,s)):(d,a,s)\in\{0,1\}^{2}\times \mathcal{S}\Big)   \notag \\
T_{n,2}& \equiv \Big( [ \sqrt{n}( \tfrac{n_{AD}(s)}{n_{A}(s)}-\pi _{D(1),\mathbf{P}}(s)) ,\sqrt{n}( \tfrac{n_{D}(s)-n_{AD}(s)}{n(s)-n_{A}(s)}-\pi _{D(0),\mathbf{P}}(s)) ] ^{\prime }:s\in S\Big) \notag \\
T_{n,3}& \equiv \Big( \sqrt{n}( \tfrac{n(s)}{n}-p_{\mathbf{P}}(s)) :s\in S\Big)  \notag \\
T_{n,4}& \equiv \Big( {\small\text{$\frac{1}{\sqrt{n}}\sum_{i=1}^{n}I[D_{i}=d,A_{i}=a,S_{i}=s](Y_{i}(d)-\mu_{\mathbf{P}}(d,a,s))^{2}:(d,a,s)\in \{0,1\}^{2}\times \mathcal{S}$}}\Big)  \notag \\
T_{n,A}& \equiv \Big( {\sqrt{n}}(\tfrac{n_A(s)}{n(s)}-\pi_{A,{\mathbf{P}}}(s)):s\in \mathcal{S}\Big) ,\label{eq:Tn_defn}
\end{align}
with $((\mu_{\mathbf{P}} (d,a,s),\pi_{D(a),\mathbf{P}}(s),p_{\mathbf{P}}(s)):(d,a,s)\in\{0,1\}^{2}\times \mathcal{S})$ is as in \eqref{eq:keyDefns} and $(\pi_{A,{\mathbf{P}}}(s):s\in\mathcal{S})$ is as in Assumption \ref{ass:2}. We make the dependence on $\mathbf{P}$ explicit as this is relevant in proving uniform convergence results.

 \subsection{Preliminary results}

\begin{lemma}\label{lem:den_not_zero}
We have
\begin{equation*}
\underset{n\to \infty }{\lim }\inf_{\mathbf{P}\in \mathcal{P}_1}\mathbf{P}(n_{A}(s)/n(s)\in (0,1)\text{ for all }s\in \mathcal{S})~=~1.
\end{equation*}
\end{lemma}
\begin{proof}
Fix $s\in \mathcal{S}$ arbitrarily. By Assumption \ref{ass:2}(c), for any $\delta >0$,
\begin{equation}
\underset{n\to \infty }{\lim }\inf_{\mathbf{P}\in \mathcal{P}_1}\mathbf{P}\left( \left\vert n_{A}(s)/n(s)-\pi _{A}(s) \right\vert \leq \delta \right) ~=~1,  \label{eq:den_not_zero}
\end{equation}%
where $\pi _{A}(s) \in \left( \varepsilon ,1-\varepsilon \right) $. By \eqref{eq:den_not_zero} with $ \delta =\varepsilon /2$,
\begin{equation}
\underset{n\to \infty }{\lim }\inf_{\mathbf{P}\in \mathcal{P}_1 }\mathbf{P}\left( n_{A}(s)/n(s)\in (0,1)\right) ~=~1.  \label{eq:den_not_zero2}
\end{equation}%
The desired result follows from \eqref{eq:den_not_zero2} and that $\mathcal{S}$ is a finite set.
\end{proof}


\begin{lemma}\label{lem:AsyDist}
Consider any sequence of $\left\{ \mathbf{P}_{n}\in \mathcal{P}_1\right\} _{n\geq 1}$ s.t.\ for all $(d,a,s)\in
\{0,1\}^{2}\times \mathcal{S}$,
\begin{align}
 p_{\mathbf{P}_{n}}(s)& ~\to ~p(s) \notag\\
\pi _{D(a),\mathbf{P}_{n}}(s)& ~\to ~\pi _{D(a)}(s) \notag\\
 \mu _{\mathbf{P}_{n}}(d,a,s)& ~\to ~\mu(d,a,s) \notag\\
 \sigma _{\mathbf{P}_{n}}^{2}(d,a,s)&~\to ~\sigma ^{2}(d,a,s)\notag\\
\pi _{A,\mathbf{P}_{n}}(s)& ~\to ~\pi _{A}(s), \label{eq:keyconvergence}
\end{align}%
where objects are as defined as in \eqref{eq:keyDefns}. Then, along $\left\{ \mathbf{P}_{n}\right\} _{n\geq 1}$, 
\begin{equation*}
\left( 
\begin{array}{c}
T_{n,1} \\ 
T_{n,2} \\ 
T_{n,3}%
\end{array}%
\right) ~\overset{d}{\to }~N\left( \left( 
\begin{array}{c}
\mathbf{0} \\ 
\mathbf{0} \\ 
\mathbf{0}%
\end{array}%
\right) ,\left( 
\begin{array}{ccc}
\Sigma _{1} & \mathbf{0} & \mathbf{0} \\ 
\mathbf{0} & \Sigma _{2} & \mathbf{0} \\ 
\mathbf{0} & \mathbf{0} & \Sigma _{3}%
\end{array}%
\right) \right) ,
\end{equation*}%
where $(T_{n,1}',T_{n,2}',T_{n,3}')$ is as in \eqref{eq:Tn_defn} but with distribution $\mathbf{P}_{n}$, and
\begin{align*}
\Sigma _{1}& ~\equiv ~diag\left( \left[ 
{\scriptsize  \begin{array}{c}
I[(d,a)=(0,0)](1-\pi _{D(0)}(s))(1-\pi _{A}(s)) \\ 
+I[(d,a)=(1,0)]\pi _{D(0)}(s)(1-\pi _{A}(s)) \\ 
+I[(d,a)=(0,1)](1-\pi _{D(1)}(s))\pi _{A}(s) \\ 
+I[(d,a)=(1,1)]\pi _{D(1)}(s)\pi _{A}(s)
\end{array}}
\right] \text{\scriptsize $p(s)\sigma ^{2}(d,a,s):(d,a,s)\in \{0,1\}^{2}\times \mathcal{S}$}
\right)  \\
\Sigma _{2}& ~\equiv ~diag\left( \left[ 
\begin{array}{cc}
\frac{(1-\pi _{D(1)}(s))\pi _{D(1)}(s)}{\pi _{A}(s)} & 0 \\ 
0 & \frac{(1-\pi _{D(0)}(s))\pi _{D(0)}(s)}{(1-\pi _{A}(s))}%
\end{array}%
\right] \tfrac{1}{p(s)}:s\in \mathcal{S}\right)  \\
\Sigma _{3}& ~\equiv ~diag\left( (p(s):s\in \mathcal{S})-(p(s):s\in \mathcal{%
S})(p(s):s\in \mathcal{S})^{\prime }\right) .
\end{align*}
\end{lemma}
\begin{proof}
Let $\zeta _{j}\sim N(\mathbf{0},\Sigma _{j})$ for $j=1,2,3$, with $(\zeta _{1}^{\prime },\zeta _{2}^{\prime },\zeta _{3}^{\prime })$ are independent. Our goal is to show that $(T_{n,1}^{\prime },T_{n,2}^{\prime },T_{n,3}^{\prime })\overset{d}{\to }(\zeta _{1}^{\prime },\zeta _{2}^{\prime },\zeta _{3}^{\prime })$. We divide the argument into 3 steps.

\underline{Step 1.} For random vectors $T_{n,1}^{C}$ and $T_{n,1}^{D}$, we
show that 
\begin{align}
& (T_{n,1}^{\prime },T_{n,2}^{\prime },T_{n,3}^{\prime })~\overset{d}{=}~({T_{n,1}^{C}}^{\prime },T_{n,2}^{\prime },T_{n,3}^{\prime })\label{eq:ad_step1_eq1} \\
& T_{n,1}^{D}~\perp ~(T_{n,2}^{\prime },T_{n,3}^{\prime })^{\prime }\label{eq:ad_step1_eq2} \\
& T_{n,1}^{D}~\overset{d}{\to }~\zeta _{1}  \label{eq:ad_step1_eq3}\\
& T_{n,1}^{C}~=~T_{n,1}^{D}+o_{p}(1).  \label{eq:ad_step1_eq4}
\end{align}

For any $((y_{i})_{i=1}^{n},(d_{i})_{i=1}^{n},(a_{i})_{i=1}^{n},(s_{i})_{i=1}^{n})\in \mathbb{R}^{n}\times  \{0,1\}^{2n}\times \mathcal{S}^{n}$, consider first the following derivation. Provided that the conditioning event has a positive probability, 
\begin{align}
& d\mathbf{P}_{n}((Y_{i}(d_{i})=y_{i})_{i=1}^{n}|((D_{i},A_{i},S_{i})=(d_{i},a_{i},s_{i}))_{i=1}^{n})\notag \\
& \overset{(1)}{=}~\frac{d\mathbf{P}_{n}(((Y_{i}(d_{i}),D_{i}(a_{i}))=(y_{i},d_{i}))_{i=1}^{n}|(S_{i}=s_{i})_{i=1}^{n})}{\mathbf{P}_{n}((D_{i}(a_{i})=d_{i})_{i=1}^{n}|(S_{i}=s_{i})_{i=1}^{n})} \notag \\
& \overset{(2)}{=}~\frac{\prod_{i=1}^{n}\int_{z_{i}:S(z_{i})=s_{i}}d\mathbf{P}_{n}((Y_{i}(d_{i}),D(a_{i}),Z_{i})=(y_{i},d_{i},z_{i}))}{\prod_{i=1}^{n}\int_{z_{i}:S(z_{i})=s_{i}}\mathbf{P}_{n}((D_{i}(a_{i}),Z_{i})=(d_{i},z_{i}))}  \notag \\
& =~\textstyle\prod_{i=1}^{n}d\mathbf{P}_{n}(Y_{i}(d_{i})=y_{i}|D_{i}(a_{i})=d_{i},S_{i}=s_{i}),
\label{eq:ad_step1_eq5}
\end{align}
where (1) by $D_{i}=D_{i}(A_{i})$ and Assumption \ref{ass:2}(a) and (2) by $S_{i}=S(Z_{i})$ and the i.i.d.\ condition in Assumption \ref{ass:1}. Note that \eqref{eq:ad_step1_eq5} implies that $ ((Y_{i}(d_{i}))_{i=1}^{n}|((D_{i},A_{i},S_{i})=(d_{i},a_{i},s_{i}))_{i=1}^{n})$ are independent observations across $i=1,\dots ,n$ distributed according to $(Y(d_{i})|D(a_{i})=d_{i},S=s_{i})$. So, if we condition on  $((D_{i},A_{i},S_{i})=(d_{i},a_{i},s_{i}))_{i=1}^{n}$, $(Y_{i}(d_{i})-\mu _{\mathbf{P}_{n}}(d_{i},a_{i},s_{i}))_{i=1}^{n}$ is independent with $( Y_{i}(d_{i})-\mu_{\mathbf{P}_{n}}(d_{i},a_{i},s_{i}) |((D_{i},A_{i},S_{i})=(d_{i},a_{i},s_{i}))_{i=1}^{n})\overset{d}{=}(Y(d_{i})-\mu _{\mathbf{P}_{n}}(d_{i},a_{i},s_{i})|D(a_{i})=d_{i},S=s_{i})$.

Conditional on $((D_{i},A_{i},S_{i})=(d_{i},a_{i},s_{i}))_{i=1}^{n}$, consider the matrix 
\begin{equation}
\text{\small $(((1[D_{i}=d,A_{i}=a,S_{i}=s](Y_{i}(d)-\mu _{\mathbf{P}_{n}}(d,a,s))):(d,a,s)\in \{0,1\}^{2}\times \mathcal{S})^{\prime }:i=1,\dots ,n)$}.
\label{eq:ad_step1_eq6}
\end{equation}
For any row $i=1,\dots ,n$ of the matrix \eqref{eq:ad_step1_eq6}, we have the following. Row $i$ has one and only one indicator $(1[D_{i}=d,A_{i}=a,S_{i}=s]:(d,a,s)\in \{0,1\}^{2}\times \mathcal{S})$ that is turned on, corresponding to $(d,a,s)=(d_{i},a_{i},s_{i})$. For this entry, we have that $1[D_{i}=d,A_{i}=a,S_{i}=s](Y_{i}(d)-\mu_{\mathbf{P}_{n}}(d,a,s))=(Y_{i}(d_{i})-\mu_{\mathbf{P}_{n}} (d_{i},a_{i},s_{i}))$. The remaining observations in row $i$ are equal to zero and, thus, independent of $ (Y_{i}(d_{i})-\mu _{\mathbf{P}_{n}}(d_{i},a_{i},s_{i}))$. In this sense, the elements of the row are independent. By the derivation in \eqref{eq:ad_step1_eq5}, conditional on $((D_{i},A_{i},S_{i})=(d_{i},a_{i},s_{i}))_{i=1}^{n}$, the rows are independent. As a consequence, conditional on $ ((D_{i},A_{i},S_{i}))_{i=1}^{n}$, \eqref{eq:ad_step1_eq6} has the same distribution as 
\begin{equation}
(((1[D_{i}=d,A_{i}=a,S_{i}=s]\breve{Y}_{i}(d,a,s)):(d,a,s)\in \{0,1\}^{2}\times \mathcal{S})^{\prime }:i=1,\dots ,n),
\label{eq:ad_step1_eq7}
\end{equation}
where $((\breve{Y}_{i}(d,a,s):(d,a,s)\in \{0,1\}^{2}\times \mathcal{S})^{\prime }:i=1,\dots ,n)$ denotes a matrix of $4|\mathcal{S}|\times n$ independent random variables, independent of $((D_{i},A_{i},S_{i}))_{i=1}^{n} $, with $\breve{Y}_{i}(d,a,s)\overset{d}{=}(Y(d)-\mu _{\mathbf{P}_{n}}(d,a,s)|D(a)=d,S=s)$ for each $(d,a,s)\in \{0,1\}^{2}\times \mathcal{S}$. As a corollary, 
\begin{equation}
(T_{n,1}|((D_{i},A_{i},S_{i}))_{i=1}^{n})~\overset{d}{=}%
~(T_{n,1}^{B}|((D_{i},A_{i},S_{i}))_{i=1}^{n}),  \label{eq:ad_step1_eq8}
\end{equation}%
where 
\begin{equation*}
T_{n,1}^{B}~\equiv ~\Big( \frac{1}{\sqrt{n}}%
\sum_{i=1}^{n}1[D_{i}=d,A_{i}=a,S_{i}=s]\breve{Y}_{i}(d,a,s):(d,a,s)\in
\{0,1\}^{2}\times \mathcal{S}\Big) .
\end{equation*}

Consider the following classification of observations. Let $m=1$ represent $(d,a,s)=(0,0,1)$, $m=2$ represent $(d,a,s)=(1,0,1)$, $m=3$ represent $(d,a,s)=(0,1,1)$, $m=4$ represent $(d,a,s)=(1,1,1)$, $m=5$ represent $(d,a,s)=(0,0,2)$, and so on, until $m=4|\mathcal{S}|$, which represents $(d,a,s)=(1,1,|\mathcal{S}|)$. Let $M:(d,a,s)\to \mathcal{M}\equiv (1,\dots ,4|\mathcal{S}|)$ denote the function that maps each $(d,a,s)$ into the index $m\in \mathcal{M}$. For each $m\in \mathcal{M}$, let $N_{m}\equiv \sum_{i=1}^{n}1[ M(D_{i},A_{i},S_{i}) < m]$. Also, let $N_{4|\mathcal{S}|+1}=N_{|\mathcal{M}|+1}=n$. Conditional on $((D_{i},A_{i},S_{i}))_{i=1}^{n}$, $(N_{M(d,a,s)}:(d,a,s)\in \{0,1\}^{2}\times \mathcal{S})$ is nonstochastic. Now consider a reordering of the units $i=1,\dots ,n$ as follows: first by strata $s\in \mathcal{S}$, then by treatment assignment $a\in \{0,1\}$, and then by decision $d\in \{0,1\}$, i.e., in increasing order according to $m\in \mathcal{M}$. Let $T_{n,1}^{C}$ denote this reordered sum and $(\eta (i):i=1,\dots ,n)$ denote the permutation of the units described by this reordering. Since $\breve{Y}_{i}(d,a,s)\overset{d}{=}(Y(d)-\mu _{\mathbf{P}_{n}}(d,a,s)|D(a)=d,S=s)$, 
\begin{align}
&((\breve{Y}_{\eta (i)}(d,a,s):(d,a,s)\in \{0,1\}^{2}\times \mathcal{S} )^{\prime }:i=1,\dots ,n)\notag \\
&~\overset{d}{=}~((\breve{Y}_{i}(d,a,s):(d,a,s)\in \{0,1\}^{2}\times \mathcal{S})^{\prime }:i=1,\dots ,n).
\label{eq:ad_step1_eq10}
\end{align}
By \eqref{eq:ad_step1_eq10}, 
\begin{equation}
(T_{n,1}^{B}| ((D_{i},A_{i},S_{i}))_{i=1}^{n})~\overset{d}{=}~(T_{n,1}^{C}| ((D_{i},A_{i},S_{i}))_{i=1}^{n}),  \label{eq:ad_step1_eq11}
\end{equation}
where 
\begin{equation}
T_{n,1}^{C}~\equiv ~\Big( \frac{1}{\sqrt{n}}\sum_{i=N_{M(d,a,s)}+1}^{N_{(M(d,a,s)+1)}}\check{Y}_{i}(d,a,s):(d,a,s)\in \{0,1\}^{2}\times \mathcal{S}\Big)   \label{eq:ad_step1_eq11B}
\end{equation}%
and $((\check{Y}_{i}(d,a,s):(d,a,s)\in \{0,1\}^{2}\times \mathcal{S} )^{\prime }:i=1,\dots ,n)$ denotes a matrix of $4|\mathcal{S}|\times n$ independent random variables, independent of $((D_{i},A_{i},S_{i}))_{i=1}^{n}  $, with $\check{Y}_{i}(d,a,s)\overset{d}{=}(Y(d)-\mu _{\mathbf{P}_{n}}(d,a,s)| D(a)=d,S=s)$.
By \eqref{eq:ad_step1_eq8} and \eqref{eq:ad_step1_eq11}, 
\begin{equation}
(T_{n,1}| ((D_{i},A_{i},S_{i}))_{i=1}^{n})~\overset{d}{=}~(T_{n,1}^{C}| ((D_{i},A_{i},S_{i}))_{i=1}^{n}).  \label{eq:ad_step1_eq12}
\end{equation}

For any $(h_{1},h_{2},h_{3})\in \mathbb{R}^{4|\mathcal{S}|}\times \mathbb{R} ^{|\mathcal{S}|}\times \mathbb{R}^{|\mathcal{S}|}$, we have
{\begin{align*}
&\mathbf{P}_{n}(T_{n,1}\leq h_{1},T_{n,2}\leq h_{2},T_{n,3}\leq h_{3}) \notag\\
&~=~E_{ \mathbf{P}_{n}}[\mathbf{P}_{n}(T_{n,1}\leq h_{1},T_{n,2}\leq h_{2},T_{n,3}\leq h_{3}| ((D_{i},A_{i},S_{i}))_{i=1}^{n})] \\
& ~\overset{(1)}{=}~E_{\mathbf{P}_{n}}[\mathbf{P}_{n}(T_{n,1}\leq h_{1}| ((D_{i},A_{i},S_{i}))_{i=1}^{n})1[T_{n,2}\leq h_{2},T_{n,3}\leq h_{3}]] \\
& ~\overset{(2)}{=}~E_{\mathbf{P}_{n}}[\mathbf{P}_{n}(T_{n,1}^{C}\leq h_{1}| ((D_{i},A_{i},S_{i}))_{i=1}^{n})1[T_{n,2}\leq h_{2},T_{n,3}\leq h_{3}]] \\
& ~\overset{(3)}{=}~E_{\mathbf{P}_{n}}[\mathbf{P}_{n}(T_{n,1}^{C}\leq h_{1},T_{n,2}\leq h_{2},T_{n,3}\leq h_{3}| ((D_{i},A_{i},S_{i}))_{i=1}^{n})] \\
& ~=~\mathbf{P}_{n}(T_{n,1}^{C}\leq h_{1},T_{n,2}\leq h_{2},T_{n,3}\leq h_{3}),
\end{align*}}
where (1) and (3) hold because $(T_{n,2},T_{n,3})$ is a nonstochastic function of $((D_{i},A_{i},S_{i}):i=1,\dots,n)$, and (2) by \eqref{eq:ad_step1_eq12}. Since $(h_{1},h_{2},h_{3})$ was arbitrary, \eqref{eq:ad_step1_eq1} follows.

For each $(d,a,s)\in \{0,1\}^{2}\times \mathcal{S}$, let $M(d,a,s)\in \mathcal{M}$ be as defined earlier. Then, we define $F_{M(d,a,s)}\in [ 0,1]$ in the following recursive fashion. First, we set $F_{M(0,0,1)}=F_{1}=0 $. Second, given $F_{M(d,a,s)}$, define 
\begin{equation}
F_{M(d,a,s)+1}~\equiv ~\left[ 
{\scriptsize  \begin{array}{c}
1[(d,a)=(0,0)](1-\pi _{D(0)}(s))(1-\pi _{A}(s)) \\ 
+1[(d,a)=(1,0)]\pi _{D(0)}(s)(1-\pi _{A}(s)) \\ 
+1[(d,a)=(0,1)](1-\pi _{D(1)}(s))\pi _{A}(s) \\ 
+1[(d,a)=(1,1)]\pi _{D(1)}(s)\pi _{A}(s)
\end{array}}
\right] p(s)+F_{M(d,a,s)}.  \label{eq:ad_step1_eq20}
\end{equation}%
Note that \eqref{eq:ad_step1_eq20} implies that $F_{|\mathcal{M}|+1}=F_{4|\mathcal{S}|+1}=1$. Then, we define 
\begin{equation}
T_{n,1}^{D}~\equiv ~\Bigg( \frac{1}{\sqrt{n}}\sum_{i=\lfloor n F_{M(d,a,s)}\rfloor +1}^{\lfloor n F_{M(d,a,s)+1}\rfloor }\check{Y}_{i}(d,a,s):(d,a,s)\in \{0,1\}^{2}\times \mathcal{S}\Bigg) .
\label{eq:ad_step1_eq15}
\end{equation}%
Since $T_{n,1}^{D}$ is a nonstochastic function of $((\check{Y} _{i}(d,a,s):(d,a,s)\in \{0,1\}^{2}\times \mathcal{S})^{\prime }:i=1,\dots ,n) $, $(T_{n,2},T_{n,3})$ is a nonstochastic function of $ ((D_{i},A_{i},S_{i}))_{i=1}^{n}$, and $((\check{Y}_{i}(d,a,s):(d,a,s)\in \{0,1\}^{2}\times \mathcal{S})^{\prime }:i=1,\dots ,n)\perp ((D_{i},A_{i},S_{i}))_{i=1}^{n}$, \eqref{eq:ad_step1_eq2} holds.

For each $(u,d,a,s)\in [ 0,1]\times \{0,1\}^{2}\times \mathcal{S}$, we consider the partial-sum process in \citet[page 397]{durrett:2019}, 
\begin{equation}
L_{n}(u,d,a,s)~=~\frac{1}{\sqrt{n}}\sum_{i=1}^{\lfloor nu\rfloor }\check{Y} _{i}(d,a,s)+\frac{nu-\lfloor nu\rfloor }{\sqrt{n}}\check{Y}_{\lfloor
nu\rfloor +1}(d,a,s),  \label{eq:durett_interp}
\end{equation}%
where we set $\check{Y}_{n+1}(d,a,s)=0$. Note that $\check{Y}_{i}(d,a,s) \overset{d}{=}(Y(d)-\mu _{\mathbf{P}_{n}}(d,a,s)|D(a)=d,S=s)$ implies that $E_{\mathbf{P}_{n}}[\check{Y} _{i}(d,a,s)]=0$ and $V_{\mathbf{P}_{n}}[\check{Y}_{i}(d,a,s)]=\sigma _{\mathbf{P}_{n}}^{2}(d,a,s)$. By Lemma \ref{lem:Bugni_et_al2018_seq}, 
\begin{equation}
(L_{n}(u,d,a,s):u\in [ 0,1])~\overset{d}{\to }~\sigma (d,a,s)\times (B(u):u\in [ 0,1]),  \label{eq:ad_step1_eq16}
\end{equation}%
where $(B(u):u\in [ 0,1])$ denotes a Brownian motion with $B(0)=0$.
Then, 
\begin{align}
\frac{1}{\sqrt{n}}\sum_{i=\lfloor nF_{M(d,a,s)}\rfloor +1}^{\lfloor nF_{M(d,a,s)+1}\rfloor }\check{Y}_{i}(d,a,s)& ~\overset{(1)}{=}~L_{n}(F_{M(d,a,s)+1},d,a,s)-L_{n}(F_{M(d,a,s)},d,a,s)+o_{p}(1)  \notag \\
& ~\overset{(2)}{\overset{d}{\to }}~N(0,(F_{M(d,a,s)+1}-F_{M(d,a,s)})\sigma ^{2}(d,a,s)).
\label{eq:ad_step1_eq17}
\end{align}%
where (1) holds by \eqref{eq:durett_interp} and $|nu-\lfloor nu\rfloor |\leq 1$ and $\check{Y}_{\lfloor nu\rfloor +1}(d,a,s)=O_{p}(1)$ for all $ u\in [ 0,1]$, and (2) by \eqref{eq:ad_step1_eq16} and the Brownian scaling relation (e.g., \citet[Equation (7.1.1)]{durrett:2019}).

Since $((\check{Y}_{i}(d,a,s):(d,a,s)\in \{0,1\}^{2}\times \mathcal{S}%
)^{\prime }:i=1,\dots ,n)$ are independent random variables, we have
that for any $(d,a,s),(\tilde{d},\tilde{a},\tilde{s})\in \{0,1\}^{2}\times 
\mathcal{S}$ with $(d,a,s)\neq (\tilde{d},\tilde{a},\tilde{s})$, 
\begin{equation}
\frac{1}{\sqrt{n}}\sum_{i=\lfloor nF_{M(d,a,s)}\rfloor +1}^{\lfloor
nF_{M(d,a,s)+1}\rfloor }\check{Y}_{i}(d,a,s)~\perp ~\frac{1}{\sqrt{n}}%
\sum_{i=\lfloor nF_{M(\tilde{d},\tilde{a},\tilde{s})}\rfloor +1}^{\lfloor
nF_{M(\tilde{d},\tilde{a},\tilde{s})+1}\rfloor }\check{Y}_{i}(\tilde{d},%
\tilde{a},\tilde{s}).  \label{eq:ad_step1_eq18}
\end{equation}%
By \eqref{eq:ad_step1_eq17} and \eqref{eq:ad_step1_eq18}, 
\begin{equation}
T_{n,1}^{D}~\overset{d}{\to }~N(\mathbf{0}%
,diag((F_{M(d,a,s)+1}-F_{M(d,a,s)})\sigma ^{2}(d,a,s):(d,a,s)\in
\{0,1\}^{2}\times \mathcal{S})).  \label{eq:ad_step1_eq19}
\end{equation}%
Finally, note that \eqref{eq:ad_step1_eq20} and \eqref{eq:ad_step1_eq19}
imply \eqref{eq:ad_step1_eq3}.

As a preliminary result to step \eqref{eq:ad_step1_eq4}, note that for any $(d,a,s)\in \{0,1\}^{2}\times \mathcal{S}$. 
\begin{align}
\frac{N_{M(d,a,s)+1}}{n}-\frac{N_{M(d,a,s)}}{n}& ~\overset{(1)}{=}~
\left[ 
{\scriptsize \begin{array}{c}
1[(a,d)=(0,0)](1-\frac{n_{D}(s)-n_{AD}(s)}{n(s)-n_{A}(s)})(1-\frac{n_{A}(s)}{n(s)}) \\ 
+1[(a,d)=(0,1)](\frac{n_{D}(s)-n_{AD}(s)}{n(s)-n_{A}(s)})(1-\frac{n_{A}(s)}{n(s)}) \\ 
+1[(a,d)=(1,0)](1-\frac{n_{AD}(s)}{n_{A}(s)})\frac{n_{A}(s)}{n(s)} \\ 
+1[(a,d)=(1,1)]\frac{n_{AD}(s)}{n_{A}(s)}\frac{n_{A}(s)}{n(s)}
\end{array}}
\right] \frac{n(s)}{n}  \notag \\
& ~\overset{(2)}{=}~F_{M(d,a,s)+1}-F_{M(d,a,s)}+o_{p}(1),
\label{eq:ad_step1_eq21}
\end{align}%
where (1) holds by definition of $(N_{m}:m=1,\dots ,|\mathcal{M}|+1)$, and (2) by the definition of $(F_{m}:m=1,\dots ,|\mathcal{M}|+1)$, Lemma \ref{lem:A1and2_impliesold3}, $p_{\mathbf{P}_{n}}(s) \to p(s)$, and ${n(s)}/{n}=p_{\mathbf{P}_{n}}(s)+o_{p}(1)$, which holds by the LLN (that applies by the i.i.d.\ condition in Assumption \ref{ass:1}). By \eqref{eq:ad_step1_eq21}, $F_{1}=N_{1}/n=0$, and an inductive argument, we get that for all $(d,a,s)\in \{0,1\}^{2}\times \mathcal{S}$, 
\begin{equation}
({N_{M(d,a,s)+1}}/{n},{N_{M(d,a,s)}}/{n})~=~(F_{M(d,a,s)+1},F_{M(d,a,s)})+o_{p}(1).\label{eq:ad_step1_eq22}
\end{equation}
To conclude the step, note that \eqref{eq:ad_step1_eq4} follows from the
next derivation. 
\begin{align*}
T_{n,1}^{C}~& \overset{(1)}{=}~(L_{n}(N_{M(d,a,s)+1}/n,d,a,s)-L_{n}(N_{M(d,a,s)}/n,d,a,s):(d,a,s)\in \{0,1\}^{2}\times \mathcal{S}) \\
~& \overset{(2)}{=}~(L_{n}(F_{M(d,a,s)+1},d,a,s)-L_{n}(F_{M(d,a,s)},d,a,s):(d,a,s)\in \{0,1\}^{2}\times \mathcal{S})+o_{p}(1) \\
~& \overset{(3)}{=}~T_{n,1}^{D}+o_{p}(1),
\end{align*}%
where (1) holds by \eqref{eq:ad_step1_eq11B} and $N_{m}=\lfloor n(N_{m}/n)\rfloor $ for all $m=1,\dots ,|\mathcal{M}|+1$, (2) by \eqref{eq:ad_step1_eq22} and that $(L_{n}(u,d,a,s):u\in \lbrack 0,1])$ is tight (implied by Lemma \ref{lem:Bugni_et_al2018_seq}), and (3) by \eqref{eq:ad_step1_eq15}, \eqref{eq:durett_interp}, and $|nu-\lfloor nu\rfloor |\leq 1$ and $\check{Y}_{\lfloor nu\rfloor +1}(d,a,s)=O_{p}(1)$ for all $u\in [ 0,1]$.

\underline{Step 2.} We show that $(T_{n,2}^{\prime },T_{n,3}^{\prime }) \overset{d}{\to }(\zeta _{2}^{\prime },\zeta _{3}^{\prime })$.

By definition, $\zeta _{2}$ and $\zeta _{3}$ are continuously distributed. Then, $(h_{2}^{\prime },h_{3}^{\prime })$ is a continuity point of the CDF of $(\zeta _{2}^{\prime },\zeta _{3}^{\prime })$. For any $(h_{2}^{\prime },h_{3}^{\prime })$,
\begin{align}
&\mathbf{P}_{n}(T_{n,2}\leq h_{2},T_{n,3}\leq h_{3})\notag\\
~& =~E_{\mathbf{P}_{n}}[E_{\mathbf{P}_{n}}[1(T_{n,2}\leq h_{2})1(T_{n,3}\leq h_{3})| (A_{i})_{i=1}^{n}, (S_{i})_{i=1}^{n}]]  \notag \\
& \overset{(1)}{=}~E_{\mathbf{P}_{n}}[E_{\mathbf{P}_{n}}[1(T_{n,2}\leq h_{2})|((A_{i},S_{i}))_{i=1}^{n}]1(T_{n,3}\leq h_{3})]  \notag \\
& =~\left[ 
{\small \begin{array}{c}
E_{\mathbf{P}_{n}}[(\mathbf{P}_{n}(T_{n,2}\leq h_{2}|((A_{i},S_{i}))_{i=1}^{n})-P(\zeta _{2}\leq h_{2}))1(T_{n,3}\leq h_{3})] +\\ 
P(\zeta _{2}\leq h_{2})(\mathbf{P}_{n}(T_{n,3}\leq h_{3})-P(\zeta _{3}\leq h_{3}))+P(\zeta _{2}\leq h_{2})P(\zeta _{3}\leq h_{3})
\end{array}}
\right] ,  \label{eq:ad_step2_eq1}
\end{align}
where (1) holds because $T_{n,3}$ is nonstochastic conditional on $((A_{i},S_{i}))_{i=1}^{n}$. By \eqref{eq:ad_step2_eq1}, 
\begin{align}
& \text{ \small $|\lim \mathbf{P}_{n}(T_{n,2}\leq h_{2},T_{n,3}\leq h_{3})-P(\zeta _{2}\leq h_{2})P(\zeta _{3}\leq h_{3})|  \leq$} \notag \\
&  \text{\small $\lim E_{\mathbf{P}_{n}}[|\mathbf{P}_{n}(T_{n,2}\leq h_{2}|((A_{i},S_{i}))_{i=1}^{n})-P(\zeta _{2}\leq h_{2})|]+\lim |\mathbf{P}_{n}(T_{n,3}\leq h_{3})-P(\zeta _{3}\leq h_{3})|$}.
\label{eq:ad_step2_eq2}
\end{align}
To complete the proof, we now show that both limits on the RHS of \eqref{eq:ad_step2_eq2} are zero. We begin with the first limit. Fix $\varepsilon >0$ arbitrarily. It suffices to find $N\in \mathbb{N}$ s.t.\ $\forall n\geq N$, $E_{\mathbf{P}_{n}}[|\mathbf{P}_{n}(T_{n,2}\leq h_{2}|((A_{i},S_{i}))_{i=1}^{n})-P(\zeta _{2}\leq h_{2})|]\leq \varepsilon $. By Assumption \ref{ass:2}(c) and Lemma \ref{lem:A1and2_impliesold3}, there is a set of values of $ ((A_{i},S_{i}))_{i=1}^{n}$ denoted by $M_{n}$ s.t.\ $ \mathbf{P}_{n}(((A_{i},S_{i}))_{i=1}^{n}\in M_{n})\to 1$ and for all $ ((a_{i},s_{i}))_{i=1}^{n}\in M_{n}$, $\mathbf{P}_{n}(T_{n,2}\leq
h_{2}|((A_{i},S_{i}))_{i=1}^{n}=((a_{i},s_{i}))_{i=1}^{n})\to P(\zeta _{2}\leq h_{2})$.
Then, $\exists N\in \mathbb{N}$ s.t.\ $\forall n\geq N$ and $\forall ((a_{i},s_{i}))_{i=1}^{n}\in M_{n}$, 
\begin{align}
& |\mathbf{P}_{n}(T_{n,2}\leq
h_{2}|((A_{i},S_{i}))_{i=1}^{n}=((a_{i},s_{i}))_{i=1}^{n})-P(\zeta _{2}\leq
h_{2})|~\leq ~\varepsilon /2  \label{eq:ad_step2_eq3} \\
& \mathbf{P}_{n}(((A_{i},S_{i}))_{i=1}^{n}\in M_{n})~\geq ~1-\varepsilon /2.
\label{eq:ad_step2_eq4}
\end{align}
Then, 
\begin{align*}
& E_{\mathbf{P}_{n}}[|\mathbf{P}_{n}(T_{n,2}\leq h_{2}|((A_{i},S_{i}))_{i=1}^{n})-P(\zeta _{2}\leq h_{2})|] \\
& \overset{(1)}{\leq }~\mathbf{P}_{n}(((A_{i},S_{i}))_{i=1}^{n}\in M_{n})\varepsilon /2+\mathbf{P}_{n}(((A_{i},S_{i}))_{i=1}^{n}\in M_{n}^{c})~\overset{(2)}{\leq }~\varepsilon ,
\end{align*}
where (1) holds by \eqref{eq:ad_step2_eq3} and (2) by \eqref{eq:ad_step2_eq4}. This completes the proof for the first limit in \eqref{eq:ad_step2_eq2}. The argument for the second limit holds because $\zeta _{3}$ is continuously distributed and $T_{n,3}\overset{d}{\to }\zeta _{3}$, which holds by the i.i.d.\ assumption in Assumption \ref{ass:1}, \ref{eq:keyconvergence}, and the CLT.

\underline{Step 3.} We now combine steps 1 and 2 to complete the proof. Let $ (h_{1}^{\prime },h_{2}^{\prime },h_{3}^{\prime })$ be a continuity point of the CDF of $(\zeta _{1}^{\prime },\zeta _{2}^{\prime },\zeta _{3}^{\prime })$. Under these conditions, we have
{\small \begin{align*}
\lim \mathbf{P}_{n}(T_{n,1}\leq h_{1},T_{n,2}\leq h_{2},T_{n,3}\leq h_{3})& ~\overset{(1)}{=}~\lim \mathbf{P}_{n}(T_{n,1}^{C}\leq h_{1},T_{n,2}\leq h_{2},T_{n,3}\leq h_{3}) \\
& ~\overset{(2)}{=}~\lim \mathbf{P}_{n}(T_{n,1}^{D}\leq h_{1},T_{n,2}\leq h_{2},T_{n,3}\leq h_{3}) \\
& ~\overset{(3)}{=}~\lim \mathbf{P}_{n}(T_{n,1}^{D}\leq h_{1})\lim \mathbf{P} _{n}(T_{n,2}\leq h_{2},T_{n,3}\leq h_{3}) \\
&~ \overset{(4)}{=}~P(\zeta _{1}\leq h_{1})P(\zeta _{2}\leq h_{2})P(\zeta _{3}\leq h_{3}),
\end{align*}}
as desired, where (1) holds by \eqref{eq:ad_step1_eq1} in step 1, (2) by \eqref{eq:ad_step1_eq4} in step 1, (3) by \eqref{eq:ad_step1_eq2} in step 1,
and (4) by \eqref{eq:ad_step1_eq3} in step 1 and step 2.
\end{proof}

\begin{lemma}\label{lem:AsyDist2}
Assume the conditions in Lemma \ref{lem:AsyDist}. Then, along $\left\{ \mathbf{P}_{n}\right\} _{n\geq 1}$, 
\begin{align}
T_{n,4}
& \overset{p}{\to}\left(\left[
{\scriptsize\begin{array}{c}
I[ ( a,d) =( 0,0) ] ( 1-\pi _{A}( s) ) ( 1-\pi _{D( 0) }( s) ) \\
+I[ ( a,d) =( 0,1) ] ( 1-\pi _{A}( s) ) \pi _{D( 0) }( s) \\
+I[ ( a,d) =( 1,0) ] \pi _{A}( s) ( 1-\pi _{D( 1) }( s) ) \\
+I[ ( a,d) =( 1,1) ] \pi _{A}( s) \pi _{D( 1) }( s)
\end{array}}
\right] \text{ \scriptsize $p( s) \sigma ^{2}( d,a,s) :(d,a,s) \in \{0,1\}^2 \times \mathcal{S}$ } \right) ,\label{eq:AsyLimit2_eq2}
\end{align}
where $T_{n,4}$ is as in \eqref{eq:Tn_defn} but with distribution $\mathbf{P}_{n}$.
\end{lemma}
\begin{proof}
Fix $( d,a,s) \in \{0,1\}^{2}\times \mathcal{S}$ arbitrarily throughout this proof. By repeating arguments used in step 1 of the proof of Lemma \ref{lem:AsyDist}, we can show that
\begin{equation}
T_{n,4}(d,a,s)~\overset{d}{=}~T_{n,4}^{D}(d,a,s)+o_{p}( 1) , \label{eq:AsyDist2_1}
\end{equation}
where $T_{n,4}^{D}(d,a,s)\equiv\frac{1}{n}\sum_{i=\lfloor nF_{M( d,a,s) }\rfloor +1}^{\lfloor nF_{M( d,a,s) +1}\rfloor }U_{i}( d,a,s)$ and $( U_{i}( d,a,s) :i=1,\ldots ,n) $ is an i.i.d.\ sequence with $U_{i}( d,a,s) \overset{d}{=}\{(Y(d)-\mu_{\mathbf{P}_{n}} (d,a,s))^{2}|D(a)=d,S=s\}$. To show \eqref{eq:AsyLimit2_eq2}, consider the following argument.
\begin{align}
T_{n,4}^{D}(d,a,s) 
&~=~1[F_{M( d,a,s) +1}>F_{M( d,a,s)}]\tfrac{\lfloor nF_{M( d,a,s) +1}\rfloor -\lfloor nF_{M( d,a,s) }\rfloor }{n} \tfrac{\sum_{i=\lfloor nF_{M( d,a,s) }\rfloor +1}^{\lfloor nF_{M( d,a,s) +1}\rfloor }U_{i}( d,a,s)}{ \lfloor nF_{M( d,a,s) +1}\rfloor -\lfloor nF_{M( d,a,s) }\rfloor }\notag \\
&~\overset{(1)}{=}~\text{\small $1[F_{M( d,a,s) +1}>F_{M( d,a,s)}]( F_{M( d,a,s) +1}-F_{M( d,a,s) }+o( 1) ) (\sigma ^{2}(d,a,s)+o_{p}( 1) )$} \notag\\
&~\overset{(2)}{=}~\left[
{\scriptsize\begin{array}{c}
1[(a,d)=(0,0)](1-\pi _{A}(s))(1-\pi _{D(0)}(s)) \\ 
+1[(a,d)=(0,1)](1-\pi _{A}(s))\pi _{D(0)}(s) \\ 
+1[(a,d)=(1,0)]\pi _{A}(s)(1-\pi _{D(1)}(s)) \\ 
+1[(a,d)=(1,1)]\pi _{A}(s)\pi _{D(1)}(s)
\end{array}}
\right] p(s)\sigma ^{2}(d,a,s)+o_{p}( 1) ,\label{eq:AsyDist2_2}
\end{align}
where (1) holds by $ E_{\mathbf{P}_{n}}[U_{i}( d,a,s) ]=\sigma_{\mathbf{P}_{n}} ^{2}(d,a,s) \to \sigma^{2}(d,a,s)$ and the LLN, and (2) by \eqref{eq:ad_step1_eq20}. The desired result follows from \eqref{eq:AsyDist2_1} and \eqref{eq:AsyDist2_2}.
\end{proof}


\begin{lemma} \label{lem:A1and2_impliesold3}
Assume the conditions in Lemma \ref{lem:AsyDist}. Then, along $\left\{ \mathbf{P}_{n}\right\} _{n\geq 1}$, 
\begin{equation}
 \{ T_{n,2} \vert S^{(n)},A^{(n)}\} 
~\overset{d}{\to }~N( \mathbf{0},\Sigma _{D})\text{ w.p.a.1,} \label{eq:A3A_statement}
\end{equation}
where $T_{n,2}$ is as in \eqref{eq:Tn_defn} but with distribution $\mathbf{P}_{n}$, and 
\begin{equation}
\Sigma _{D}~\equiv ~diag\left(\left[
\begin{array}{cc}
\tfrac{\pi _{D(1)}(s)(1-\pi _{D(1)}(s))}{\pi _{A}(s)} & 0 \\
0 & \tfrac{\pi _{D(0)}(s)(1-\pi _{D(0)}(s))}{1-\pi _{A}(s)}
\end{array}
\right] \tfrac{1}{p(s)}:s\in \mathcal{S}\right).
\end{equation}
In addition, along $\left\{ \mathbf{P}_{n}\in \mathcal{P}_1\right\} _{n\geq 1}$, 
\begin{equation}
\left(\tfrac{n_{AD}(s)}{n_{A}( s) },\tfrac{ n_{D}(s)-n_{AD}(s) }{ n( s) -n_{A}(s) }\right) ~\overset{p}{\to}~( \pi _{D( 1) }( s) ,\pi _{D( 0) }( s) ) .
\label{eq:A3B_statement}
\end{equation} 
\end{lemma}
\begin{proof}

We only show \eqref{eq:A3A_statement}, as \eqref{eq:A3B_statement} follows from \eqref{eq:A3A_statement} and elementary convergence arguments. We split the proof of \eqref{eq:A3A_statement} into two steps. First, we show
\begin{align}
&\left. \left\{ \left(\left(  
{\scriptsize\begin{array}{c}
\sqrt{n_{A}( s) }\left( \tfrac{n_{AD}(s)}{n_A(s)}-\pi _{D(1),\mathbf{P}_{n}}(s)\right) \\
\sqrt{n( s) -n_{A}( s) }\left( \tfrac{n_{0D}(s)}{n(s)-n_A(s)}-\pi _{D(0),\mathbf{P}_{n}}(s)\right) 
\end{array}}
\right):s\in \mathcal{S}\right) \right\vert S^{(n)},A^{(n)}\right\} \notag \\
&~\overset{d}{\to}~N\left( \mathbf{0},diag\left( \left[
{\scriptsize
\begin{array}{cc}
\pi _{D(1)}(s)(1-\pi _{D(1)}(s)) & 0 \\
0 & \pi _{D(0)}(s)(1-\pi _{D(0)}(s))
\end{array}}
\right] :s\in \mathcal{S}\right) \right) \text{ w.p.a.1,}\label{eq:pf_A3_6}
\end{align}
where $n_{0D}( s) \equiv n_D( s) - n_{AD}( s) $
for any $s\in \mathcal{S}$. Second, we show
\begin{equation}
\left. \left\{ \left(\tfrac{\sqrt{n}}{\sqrt{n_{A}( s) }}, \tfrac{\sqrt{n}}{\sqrt{n( s) -n_{A}( s) }}\right) \right\vert S^{(n)},A^{(n)}\right\} \to \left( \tfrac{1}{\sqrt{p( s)\pi _{A}( s) }},\tfrac{1}{ \sqrt{p( s)(1-\pi _{A}( s)) }}\right) \text{ w.p.a.1.}
\label{eq:pf_A3_8}
\end{equation}
Then, \eqref{eq:A3A_statement} follows from \eqref{eq:pf_A3_6} and \eqref{eq:pf_A3_8} via elementary convergence arguments.

\underline{Step 1: Show \eqref{eq:pf_A3_6}.} Conditional on $(S^{(n)},A^{(n)})$, note that $(( n_{A}(s),n( s) ) :s\in \mathcal{S}) $ is non-stochastic, and so the only source of randomness in \eqref{eq:pf_A3_6} is $( ( n_{AD}(s),n_{0D}( s) ) :s\in \mathcal{S})$. Also, it is relevant to note that
\begin{align}
n_{AD}( s) &~=~\sum_{i=1}^{n}1[ A_{i}=1,D_{i}( 1) =1,S_{i}=s] ~=~\sum_{i=1}^{n}1[ A_{i}=1,S_{i}=s] D_{i}( 1) \notag\\
n_{0D}( s) &~=~\sum_{i=1}^{n}1[ A_{i}=0,D_{i}( 0) =1,S_{i}=s] ~=~\sum_{i=1}^{n}1[ A_{i}=0,S_{i}=s] D_{i}( 0) . \label{eq:pf_A3_2}
\end{align}
According to \eqref{eq:pf_A3_2}, each component of $(( n_{AD}(s),n_{0D}( s) ):s\in \mathcal{S})$ is determined by a different subset of individuals in the random sample.

As a next step, consider the following derivation for any $(d_{0,i}) _{i=1}^{n}\times ( d_{1,i}) _{i=1}^{n}\times ( a_{i}) _{i=1}^{n}\times (s_{i})_{i=1}^{n}\in \{ 0,1\} ^{n}\times \{ 0,1\} ^{n}\times \{ 0,1\} ^{n}\times \mathcal{S}^{n}$.
\begin{align}
&\mathbf{P}_{n}( ( ( D_{i}( 0) ,D_{i}( 1) ) ) _{i=1}^{n}=( ( d_{0,i},d_{1,i}) ) _{i=1}^{n} | ( ( A_{i},S_{i}) ) _{i=1}^{n}=( ( a_{i},s_{i}) ) _{i=1}^{n}) \notag \\
&\overset{(1)}{=}~\mathbf{P}_{n}( ( ( D_{i}( 0) ,D_{i}( 1) ) ) _{i=1}^{n}=( ( d_{0,i},d_{1,i}) ) _{i=1}^{n} | ( S_{i}) _{i=1}^{n}=( s_{i}) _{i=1}^{n}) \notag \\
&\overset{(2)}{=}~
\textstyle\prod_{i=1}^{n}\mathbf{P}_{n}( ( D( 0) ,D( 1) ) =( d_{0,i},d_{1,i}) | S=s_{i}) , \label{eq:pf_A3_4}
\end{align}
where (1) follows from Assumption \ref{ass:2}(a) and (2) follows from the i.i.d.\ condition in Assumption \ref{ass:1}. Conditionally on $( ( A_{i},S_{i}) ) _{i=1}^{n}=( ( a_{i},s_{i}) ) _{i=1}^{n}$, \eqref{eq:pf_A3_4} reveals that $( ( D_{i}( 0) ,D_{i}( 1) ) ) _{i=1}^{n}$ is an independent sample with $( (  D_{i}( 0) ,D_{i}( 1) ) |( ( A_{i},S_{i}) ) _{i=1}^{n}=( ( a_{i},s_{i}) ) _{i=1}^{n}) \overset{d}{=}( ( D( 0) ,D( 1) ) |S=s_{i}) $.

By \eqref{eq:pf_A3_2}, $(( n_{AD}(s),n_{0D}( s) ):s\in \mathcal{S}) $ are the sum of binary observations from different individuals. Conditional on $( S^{(n)},A^{(n)})$, \eqref{eq:pf_A3_2} and \eqref{eq:pf_A3_4} imply that
\begin{equation}
 \{ ( ( n_{AD}(s),n_{D}(s)-n_{AD}(s)):s\in \mathcal{S}) \vert S^{(n)},A^{(n)} \} ~\overset{d}{=}~( ( B_n( 1,s) ,B_n( 0,s) ) :s\in \mathcal{S} ) , \label{eq:pf_A3_5}
\end{equation}
where $( ( B_n( 1,s) ,B_n( 0,s) ) :s\in \mathcal{S}) $ are independent random variables with $ B_n( 1,s) \sim Bi( n_{A}( s) ,\pi _{D(1),\mathbf{P}_{n}}(s)) $ and $B_n( 0,s) \sim Bi( n( s) -n_{A}( s) ,\pi _{D(0),\mathbf{P}_{n}}(s)) $. For sequences of $(( S_{i},A_{i}) )_{i=1}^{n}$ with $n_{A}( s) \to \infty $ and $n( s) -n_{A}( s) \to \infty $ for all $s\in \mathcal{S}$, \eqref{eq:pf_A3_6} follows immediately from \eqref{eq:pf_A3_5}, the CLT, and $\pi _{D(a),\mathbf{P}_{n}}(s) \to \pi _{D(a)}(s)$. 

To conclude the step, it suffices to show that $n_{A}( s) \to \infty $ and $n( s) -n_{A}( s) \to \infty $ for all $s\in \mathcal{S}$ w.p.a.1. In turn, note that this is a consequence of $n( s) /n - p_{\mathbf{P}_{n}}(s)=o_p(1)$ by the LLN (which applies by the i.i.d.\ condition in Assumption \ref{ass:1}), $p_{\mathbf{P}_{n}}(s) - p(s)=o(1)$,  and $n_{A}( s) /n( s) \overset{p}{\to}\pi _{A}( s) \in ( 0,1) $ by Assumption \ref{ass:2}(c).

\underline{Step 2: Show \eqref{eq:pf_A3_8}.} Fix $s\in \mathcal{S}$ arbitrarily. By Assumptions \ref{ass:1}(a)-(c) and \ref{ass:2}(c),
\begin{align}
\left( \tfrac{\sqrt{n}}{\sqrt{n_{A}( s) }}, \tfrac{\sqrt{n}}{\sqrt{n( s) -n_{A}( s) }}\right) 
~&\overset{p}{\to}~\left( \tfrac{1}{\sqrt{p( s)  \pi _{A}( s) }},\tfrac{1}{\sqrt{p( s) (1-\pi _{A}( s)) }} \right) .\label{eq:pf_A3_7}
\end{align}
From \eqref{eq:pf_A3_7} and the fact that $( ( n_{A}(s),n( s) ) :s\in \mathcal{S}) $ is non-stochastic conditional on $( S^{(n)},A^{(n)})$, \eqref{eq:pf_A3_8} follows.
\end{proof}

 
\begin{lemma}\label{lem:Bugni_et_al2018_seq}
Assume the conditions in Lemma \ref{lem:AsyDist}. For any $( d,a,s) \in \{ 0,1\}^2 \times \mathcal{S}$ and $u\in [0,1]$, consider the partial-sum process in \citet[page 397]{durrett:2019},
\begin{equation}
L_{n}(u,d,a,s)~=~\frac{1}{\sqrt{n}}\sum_{i=1}^{\lfloor nu\rfloor }\check{Y} _{i}(d,a,s) + \frac{nu - \lfloor nu \rfloor}{\sqrt{n}} \check{Y} _{\lfloor nu \rfloor+1}(d,a,s),
\label{eq:durett_interp2}
\end{equation}
where $\check{Y} _{n+1}(d,a,s)= 0$ and $( \check{Y}_{i}(d,a,s):i=1,\dots ,n)$ is a vector of independent random variables, independent of $(( D_{i},A_{i},S_{i}):i=1,\dots,n )$, with $\check{Y}_{i}(d,a,s) \overset{d}{=}( Y_i(d)-\mu_{\mathbf{P}_{n}} (d,a,s)|D( a) =d,S=s) $. Then, along $\left\{ \mathbf{P}_{n}\right\} _{n\geq 1}$, 
\begin{equation}
    (L_{n}(u,d,a,s):u\in [0,1])~\overset{d}{\to}~\sigma(d,a,s) \times (B(u):u \in [0,1]),
    \label{eq:CLT0}
\end{equation}
where $(B(u) : u \in [0,1])$ is a standard Brownian motion with $B(0)=0$.
\end{lemma}
\begin{proof}
We divide the proof into two cases.

\underline{Case 1:} $\sigma (d,a,s)=0$. Fix $\varepsilon >0$ arbitrarily.
Then,
\begin{align}
\mathbf{P}_{n}\big(\sup_{u\in [ 0,1]}\vert L_{n}(u,d,a,s)\vert >\varepsilon \big)& ~{=}~\mathbf{P}_{n}\big(\max_{1\leq m\leq n}\vert n^{-1/2}\sum_{i=1}^{m}\check{Y}_{i}(d,a,s)\vert >\epsilon \big) \notag\\
& ~\overset{(1)}{\leq }~\varepsilon ^{-2}E_{\mathbf{P}_{n}}\big[ \big( n^{-1/2}\sum_{i=1}^{n}\check{Y}_{i}(d,a,s)\big) ^{2}\big] \nonumber \\
& ~\overset{(2)}{=}~\varepsilon ^{-2}\sigma _{\mathbf{P}_{n}}^{2}(d,a,s), \label{eq:CLT1}
\end{align}
where (1) holds by Kolmogorov's maximal inequality (\citet[Theorem 2.5.5]{durrett:2019})) and (2) because $(\check{Y} _{i}(d,a,s):i=1,\dots,n)$ is an i.i.d.\ sequence with $E_{\mathbf{P}_{n}}[ \check{Y }_{i}(d,a,s)] =0$ and $V_{\mathbf{P}_{n}}[ \check{Y}_{i}(d,a,s)] =\sigma _{\mathbf{P}_{n}}^{2}(d,a,s)$. Since $\sigma _{\mathbf{P}_{n}}^{2}(d,a,s)\to \sigma^{2} ( d,a,s) =0$, \eqref{eq:CLT1} implies that $(L_{n}(u,d,a,s):u\in [ 0,1 ] )\overset{p}{\to }0\times (B(u):u\in [ 0,1] )$, as desired.

\underline{Case 2:} $\sigma (d,a,s)>0$. Since $\sigma _{\mathbf{P}_{n}}(d,a,s)\to \sigma ( d,a,s)$, there is $N$ s.t.\ $\forall n\geq N$, $\sigma _{\mathbf{P}_{n}}(d,a,s)\geq \sigma ( d,a,s) /2>0$. We focus on $n\geq N$ for the rest of the proof. Then,
\begin{equation}
L_{n}(u,d,a,s)~=~\sigma (d,a,s)\tfrac{L_{n}(u,d,a,s)}{\sigma_{\mathbf{P}_{n}}(d,a,s)} +\delta _{\mathbf{P}_{n}}( u) , \label{eq:CLT2}
\end{equation}
where
\begin{equation}
\delta _{\mathbf{P}_{n}}( u) ~\equiv~ \sigma (d,a,s)\tfrac{L_{n}(u,d,a,s)}{\sigma _{\mathbf{P}_{n}}(d,a,s)}( \tfrac{\sigma _{\mathbf{P}_{n}}(d,a,s)}{\sigma (d,a,s)}-1) . \label{eq:CLT3}
\end{equation}

The proof is completed by showing that
\begin{equation}
( L_{n}(u,d,a,s)/\sigma _{\mathbf{P}_{n}}(d,a,s):u\in [ 0,1] ) ~\overset{d}{\to }~( B(u):u\in [ 0,1] ) .
\label{eq:CLT5}
\end{equation}
To see why, note that \eqref{eq:CLT3}, \eqref{eq:CLT5}, and $\sigma _{\mathbf{P}_{n}}(d,a,s)\to \sigma ( d,a,s) >0$ implies that $(\delta _{\mathbf{P}_{n}}( u) :u\in [ 0,1 ]) \overset{p}{\to }(0:u\in [ 0,1] )$, which, when combined with \eqref{eq:CLT2} and \eqref {eq:CLT5}, implies \eqref{eq:CLT0}. To show \eqref{eq:CLT5}, note that 
\begin{equation}
\frac{L_{n}(u,d,a,s)}{\sigma _{\mathbf{P}_{n}}(d,a,s)}~=~\sum_{i=1}^{\lfloor nu\rfloor }W_{n,i}+(nu-\lfloor nu\rfloor )W_{n,\lfloor nu\rfloor +1},
\label{eq:durett_interp3}
\end{equation}%
where, for all $i=1,\ldots ,n$, 
\begin{equation}
W_{n,i}~\equiv ~\frac{\check{Y}_{i}(d,a,s)}{\sqrt{n}\sigma _{\mathbf{P}_{n}}(d,a,s)},  \label{eq:CLT4}
\end{equation}%
and $W_{n,n+1}=0$. Then, \eqref{eq:CLT5} is a corollary of \citet[Theorem 8.2.4]{durrett:2019} under the following conditions:
\begin{enumerate}[(i)]
\item $\{(W_{n,i},\mathcal{F}_{n,i}):i=1,\dots ,n\}$ is a martingale difference array, where $\mathcal{F}_{n,i}$ denotes the $\sigma $-algebra generated by $(W_{n,j}:j=1,\dots ,i)$.
\item For all $u\in [ 0,1]$, $V_{n,\lfloor nu\rfloor }\equiv\sum_{i=1}^{\lfloor nu\rfloor }E_{\mathbf{P}_{n}}[W_{n,i}^{2}]~\to~u.$
\item For all $\epsilon >0$, $\sum_{i=1}^{n}E_{\mathbf{P}_{n}}[W_{n,i}^{2}I\{|W_{n,i}|>\epsilon \}| \mathcal{F}_{n,i-1}]~=~o_{p}(1).$
\end{enumerate}
We now verify these conditions. For (i), it suffices to check that $W_{n,i}\in \mathcal{F}_{n,i}$ and $E_{\mathbf{P}_{n}}[W_{n,i}| \mathcal{F} _{n,i-1}]=0$ for any $i=1,\dots ,n$. The first one follows by definition of $\mathcal{F}_{n,i}$. The second one follows from \eqref{eq:CLT4} and that $(\check{Y}_{i}(d,a,s):i=1,\dots ,n)$ is an i.i.d.\ sequence with $E_{\mathbf{P}_{n}}[\check{Y}_{i}(d,a,s)]=0$, which gives $E_{\mathbf{P} _{n}}[W_{n,i}| \mathcal{F}_{n,i-1}]=E_{\mathbf{P}_{n}}[W_{n,i}]=0$.
To verify condition (ii), note that \eqref{eq:CLT4} implies $E_{\mathbf{P}_{n}}[W_{n,i}^{2}] = \sigma _{\mathbf{P}_{n}}^{-2}(d,a,s)E_{\mathbf{P}_{n}}[\check{Y}_{i}^{2}(d,a,s)]/n = 1/n$. From this, we get that $V_{n,\lfloor nu\rfloor }=\lfloor nu\rfloor /n\to u$, as desired.
Before proving condition (iii), consider the following preliminary argument. 
\begin{align}
E_{\mathbf{P}_{n}}[|\check{Y}_{i}(d,a,s)|^{2+\delta }]& ~\overset{(1)}{=}~E_{\mathbf{P}_{n}}[|Y_{i}(d)-\mu _{\mathbf{P}_{n}}(d,a,s)|^{2+\delta }|
D_{i}(a)=d,S_{i}=s]  \notag \\
& ~\overset{(2)}{\leq }~2^{1+\delta }(E_{\mathbf{P}_{n}}[|Y_{i}(d)|^{2+\delta }| D_{i}(a)=d,S_{i}=s]+|\mu _{\mathbf{P}_{n}}(d,a,s)|^{2+\delta })  \notag \\
& ~\overset{(3)}{\leq }~2^{2+\delta }E_{\mathbf{P}_{n}}[|Y_{i}(d)|^{2+\delta}| D_{i}(a)=d,S_{i}=s]  \notag \\
& ~\overset{(4)}{\leq }~2^{2+\delta }\max\{|Y_L|,|Y_H|\}^{2+\delta }.  \label{eq:CLT6}
\end{align}
where (1) holds by $\check{Y}_{i}(d,a,s)\overset{d}{=}\{Y_{i}(d)-\mu_{\mathbf{P}_{n}}(d,a,s)|D_{i}(a)=d,S_{i}=s\}$, (2) by Minkowski's inequality, (3) by H\"{o}lder's inequality, and (4) by Assumption \ref{ass:1}(a). Then, condition
(iii) holds by the next derivation. 
\begin{align*}
\sum_{i=1}^{n}E_{\mathbf{P}_{n}}\{W_{n,i}^{2}1[|W_{n,i}|>\epsilon \}| 
\mathcal{F}_{n,i-1}]
& ~\overset{(1)}{=}~nE_{\mathbf{P}_{n}}[W_{n,i}^{2}1[|W_{n,i}|>\epsilon ]] \\
& ~\leq ~nE_{\mathbf{P}_{n}}[|{W_{n,i}|^{2+\delta }}{\epsilon ^{-\delta }}1[|W_{n,i}|>\epsilon ]] \\
& ~\overset{(2)}{\leq }~\tfrac{2^{4+2\delta }\max\{|Y_L|,|Y_H|\}^{2+\delta }}{n^{\delta /2}\sigma
^{2+\delta }(d,a,s)\epsilon ^{\delta }}\to 0,
\end{align*}%
as desired, where (1) holds by \eqref{eq:CLT4} and that $(\check{Y}_{i}(d,a,s):i=1,\dots ,n)$ is an i.i.d.\ sequence, (2) by \eqref{eq:CLT4},
\eqref{eq:CLT6}, and $\sigma _{\mathbf{P}_{n}}(d,a,s)\geq \sigma (d,a,s)/2>0$
(as $n\geq N$).
\end{proof}

\subsection{ATE}


\begin{proof}[Proof of Theorem \ref{thm:consist}]
Fix $\varepsilon > 0$ arbitrarily. By definition of $\Theta_I({\bf P})$ and $\hat{\Theta}_I$, 
\[
\{ |\theta_L({\bf P}) - \hat{\theta}_L| \leq \varepsilon \} ~\cap~ \{ |\theta_H({\bf P}) - \hat{\theta}_H| \leq \varepsilon \}~~ \subseteq~~ \{ d_H(\hat{\Theta}_{I}, \Theta_I({\bf P})) \leq \varepsilon \}.
\]
Thus, it suffices to show
\begin{align*}
\liminf_{n \to \infty} \inf_{{\bf P} \in \mathcal{P}_1} \mathbf{P}(|\theta_L({\bf P}) - \hat{\theta}_L| \leq \varepsilon) =  \liminf_{n \to \infty} \inf_{{\bf P} \in \mathcal{P}_1} \mathbf{P}(|\theta_H({\bf P}) - \hat{\theta}_H| \leq \varepsilon) = 1 . 
\end{align*}
To conclude the proof, note that both equalities follow from Theorem \ref{thm:AsyDist_ATE}.
\end{proof}

\begin{proof}[Proof of Theorem \ref{thm:Stoye}.]
The result follows from \citet[Proposition 1]{stoye:2009}. To apply this result, it suffices to verify its conditions. Assumption 1(i) follows from Theorems \ref{thm:AsyDist_ATE} and \ref{thm:AsyVarEstimation}, Assumption 1(ii) holds by Theorem \ref{thm:AsyDist_ATE} and Assumption \ref{ass:1}(a), and Assumption 3 holds by $\hat{\theta}_{H}\geq \hat{\theta}_{L}$ and \citet[Lemma 3]{stoye:2009}.
\end{proof}

\begin{lemma}\label{lem:expressionBounds} 
Let
\begin{align}
\check{\theta}_{L}& ~\equiv ~\frac{1}{n}\sum_{i=1}^{n}\bigg[\frac{(Y_{i}D_{i}+Y_{L}(1-D_{i}))A_{i}}{\hat{\pi}_A(S_i)}-\frac{(Y_{i}(1-D_{i})+Y_{H}D_{i})(1-A_{i})}{1-\hat{\pi}_A(S_i)}\bigg] \notag\\
\check{\theta}_{H}& ~\equiv ~\frac{1}{n}\sum_{i=1}^{n}\bigg[\frac{(Y_{i}D_{i}+Y_{H}(1-D_{i}))A_{i}}{\hat{\pi}_A(S_i)}-\frac{(Y_{i}(1-D_{i})+Y_{L}D_{i})(1-A_{i})}{1-\hat{\pi}_A(S_i)}\bigg],
\label{eq:expressionBounds1}
\end{align}
where $\{\hat{\pi}_A(s):s \in \mathcal{S}\} $ is an arbitrary estimator of $\{{\pi}_A(s):s \in \mathcal{S}\} $ that satisfies
\begin{equation}
\underset{n\to  \infty }{\lim }\inf_{\mathbf{P}\in \mathcal{P}_1}\mathbf{P}\big( \hat{\pi}_A(s)\in (0,1)\text{ for all }s \in \mathcal{S}\big)~=~1.
\label{eq:expressionBounds2}
\end{equation}

Then,
\begin{equation}
\sqrt{n}\left( 
\begin{array}{c}
\check{\theta}_{L}-\theta _{L}(\mathbf{P}) \\ 
\check{\theta}_{H}-\theta _{H}(\mathbf{P})
\end{array}
\right) =\left( 
\begin{array}{c}
U_{n}+V_{L,n}+W_{L,n} \\ 
U_{n}+V_{H,n}+W_{H,n}
\end{array}
\right) +\delta _{n},  \label{eq:expressionBounds3}
\end{equation}
where $\theta _{L}(\mathbf{P})$ and $\theta _{H}(\mathbf{P})$ are as in \eqref{eq:bounds}, $\delta _{n}=o_{p}(1)$ uniformly in $\mathbf{P}\in \mathcal{P}_1$, and
\begin{align}
U_{n}& \equiv \sum_{s\in \mathcal{S}}\left[ 
\begin{array}{c}
T_{n,1}\left( 1,1,s\right) /\pi _{A}(s) -T_{n,1}\left(0,0,s\right) /(1-\pi _{A}(s) ) \\ 
+p(s)\mu (1,1,s)T_{n,2}\left( 1,s\right) +p(s)\mu (0,0,s)T_{n,2}\left(2,s\right)  \\ 
+[\pi _{D(1)}(s) \mu (1,1,s)-(1-\pi _{D(0) }(s) )\mu (0,0,s)]T_{n,3}(s) 
\end{array}
\right]   \notag \\
V_{L,n}& \equiv \sum_{s\in \mathcal{S}}\left[ 
\begin{array}{c}
-p(s)Y_{L}T_{n,2}\left( 1,s\right) -p(s)Y_{H}T_{n,2}\left( 2,s\right)  \\ 
+[(1-\pi _{D(1)}(s) )Y_{L}-\pi _{D(0)
}(s) Y_{H}]T_{n,3}(s) 
\end{array}
\right]   \notag \\
W_{L,n}& \equiv \sum_{s\in \mathcal{S}}\sqrt{n}(\tfrac{n_{A}(s)}{n(s)}-\hat{\pi}_A(s) )p(s)\left[ 
\begin{array}{c}
\pi _{D(1)}(s) \mu (1,1,s)/\pi _{A}(s) +\\ 
(1-\pi _{D(0) }(s) )\mu (0,0,s)/(1-\pi _{A}(s) ) \\ 
+(1-\pi _{D(1)}(s) )Y_{L}/\pi _{A}(s) \\ 
+\pi _{D(0) }(s) Y_{H}/(1-\pi _{A}(s) )
\end{array}
\right] ,  \label{eq:expressionBounds5}
\end{align}
with $(T_{n,1},T_{n,2},T_{n,3})$ as in \eqref{eq:Tn_defn}, and $V_{H,n}$ and $W_{H,n}$ are defined as in $V_{L,n}$ and $W_{L,n}$ but with $Y_{L}$ and $ Y_{H}$ interchanged. 
\end{lemma}
\begin{proof}
Let $E_{n}=\{ n_{A}(s)/n(s)\in ( 0,1) \text{ and }\hat{\pi}_A(s)\in (0,1)\text{ for all }s\in \mathcal{S}\} $ and 
\begin{align*}
\tilde{U}_{n}&  ~\equiv~\sum_{s\in \mathcal{S}}\left[ 
\begin{array}{c}
T_{n,1}\left( 1,1,s\right) /\frac{n_{A}(s)}{n(s)}-T_{n,1}\left( 0,0,s\right)/(1-\frac{n_{A}(s)}{n(s)}) \\ 
+p(s)\mu (1,1,s)T_{n,2}\left( 1,s\right) +p(s)\mu (0,0,s)T_{n,2}\left(2,s\right)  \\ 
+[\frac{n_{AD}(s)}{n_{A}(s)}\mu (1,1,s)-(1-\frac{n_{D}(s)-n_{AD}(s)}{n(s)-n_{A}(s)})\mu (0,0,s)]T_{n,3}(s) 
\end{array}
\right]  \\
\tilde{V}_{L,n}&  ~\equiv~\sum_{s\in \mathcal{S}}\left[ 
\begin{array}{c}
-p(s)Y_{L}\sqrt{n}(\frac{n_{AD}(s)}{n_{A}(s)}-\pi _{D(1)}(s))
-p(s)Y_{H}\sqrt{n}(\frac{n_{D}(s)-n_{AD}(s)}{n(s)-n_{A}(s)}-\pi _{D(0)}(s)) \\ 
+[(1-\frac{n_{AD}(s)}{n_{A}(s)})Y_{L}-\frac{n_{D}(s)-n_{AD}(s)}{n(s)-n_{A}(s)}Y_{H}]\sqrt{n}(\frac{n(s)}{n}-p(s))
\end{array}
\right]   \notag \\
\tilde{W}_{L,n}& ~\equiv~ \sum_{s\in \mathcal{S}}\left[ 
\begin{array}{c}
\sqrt{n}(\frac{n_{A}(s)}{n(s)}-\hat{\pi}_A(s)) )\times  \\ 
\left[ 
\begin{array}{c}
\frac{T_{n,1}\left( 1,1,s\right) }{\sqrt{n}}/(\frac{n_{A}(s)}{n(s)}\hat{\pi}_{A}(s))+\frac{T_{n,1}\left( 0,0,s\right) }{\sqrt{n}}/((1-\frac{n_{A}(s)}{n(s)})(1-\hat{\pi}_{A}(s))) \\ 
+\frac{n(s)}{n}\frac{n_{AD}(s)}{n_{A}(s)}\mu (1,1,s)/\hat{\pi}_{A}(s) \\ 
+\frac{n(s)}{n}(1-\frac{n_{D}(s)-n_{AD}(s)}{n(s)-n_{A}(s)})\mu (0,0,s)/(1-\hat{\pi}_{A}(s)) \\ 
+\frac{n(s)}{n}(1-\frac{n_{AD}(s)}{n_{A}(s)})Y_{L}/\hat{\pi}_{A}(s)+\frac{n(s)}{n}\frac{n_{D}(s)-n_{AD}(s)}{n(s)-n_{A}(s)}Y_{H}/(1-\hat{\pi}_{A}(s))
\end{array}
\right] 
\end{array}
\right] .
\end{align*}
Some algebra shows that \eqref{eq:expressionBounds3} holds with $\delta _{n}=\delta _{n,1}+\delta _{n,2}$, where
\begin{align*}
\delta _{n,1} &~\equiv~ \left( 
\begin{array}{c}
( U_{n}-\tilde{U}_{n}) +( V_{L,n}-\tilde{V}_{L,n})+( W_{L,n}-\tilde{W}_{L,n})  \\ 
( U_{n}-\tilde{U}_{n}) +( V_{H,n}-\tilde{V}_{H,n})+( W_{H,n}-\tilde{W}_{H,n}) 
\end{array}
\right) \\
\delta _{n,2} &~\equiv~ 1\left[ E_{n}^{c}\right] \left( \sqrt{n}\left(
\begin{array}{c}
\tilde{\theta}_{L}-\theta _{L}(\mathbf{P}) \\ 
\tilde{\theta}_{H}-\theta _{H}(\mathbf{P})
\end{array}
\right) -\left(
\begin{array}{c}
\tilde{U}_{n}+\tilde{V}_{L,n}+\tilde{W}_{L,n} \\ 
\tilde{U}_{n}+\tilde{V}_{H,n}+\tilde{W}_{H,n}
\end{array}
\right) \right) .
\end{align*}
To complete the proof, it suffices to show $\delta _{n,1}=o_{p}(1) $ and $\delta _{n,2}=o_{p}(1)$ uniformly in $\mathbf{P}\in \mathcal{P}_1$. The first one follows from Lemma \ref{lem:AsyDist}, which implies that $U_{n}-\tilde{U}_{n}=o_{p}(1)$, $V_{L,n}-\tilde{V}_{L,n}=o_{p}(1)$, $W_{L,n}-\tilde{W}_{L,n}=o_{p}(1)$, $V_{H,n}-\tilde{V}_{H,n}=o_{p}(1)$, and $W_{H,n}-\tilde{W}_{H,n}=o_{p}(1)$, all uniformly in $\mathbf{ P}\in \mathcal{P}_1$. The second one follows from the fact that $ E_{n}\subseteq \left\{ \delta _{n,2}=0\right\} $ and $\lim_{n\to  \infty }\inf_{\mathbf{P}\in \mathcal{P}_1}P(E_{n})=1$ which, in turn, follows from Lemma \ref{lem:den_not_zero} and \eqref{eq:expressionBounds2}.
\end{proof}

\begin{theorem}\label{thm:AsyDist_ATE} 
We have
\begin{equation}
\sqrt{n}\left( 
\begin{array}{c}
\hat{\theta}_{L}-\theta _{L}(\mathbf{P}) \\ 
\hat{\theta}_{H}-\theta _{H}(\mathbf{P})
\end{array}
\right) ~\overset{d}{\to }~N\left( {\bf 0}_2 ,\left( 
\begin{array}{cc}
\sigma _{L}^{2}(\mathbf{P})  & \sigma _{HL}\left( \mathbf{P}\right)  \\ 
\sigma _{HL}(\mathbf{P})  & \sigma _{H}^{2}\left( \mathbf{P}\right) 
\end{array}%
\right) \right) ,  \label{eq:AsyDist_ATE1}
\end{equation}%
uniformly in $\mathcal{P}_1$, where
\begin{align}
\sigma _{L}^{2}(\mathbf{P}) & =\text{\small$\sum_{s\in \mathcal{S}}p(s)$}\left[ 
{\scriptsize\begin{array}{c}
(\sigma ^{2}(1,1,s)+(\mu (1,1,s)-Y_{L})^{2}(1-\pi _{D(1)}(s)))\pi_{D(1)}(s)/\pi _{A}(s)+ \\ 
(\sigma ^{2}(0,0,s)+(\mu (0,0,s)-Y_{H})^{2}\pi _{D(0)}(s))(1-\pi_{D(0)}(s))/(1-\pi _{A}(s)) \\ 
+\left[ \left( 
\begin{array}{c}
\pi _{D(1)}(s)\mu (1,1,s)+(1-\pi _{D(1)}(s))Y_{L} \\ 
-(1-\pi _{D(0)}(s))\mu (0,0,s)-\pi _{D(0)}(s)Y_{H}
\end{array}
\right) -\theta _{L}(\mathbf{P})\right] ^{2}
\end{array}}
\right]   \notag \\
\sigma _{HL}(\mathbf{P}) & =\text{\small$\sum_{s\in \mathcal{S}}p(s)$}\left[ 
{\scriptsize
\begin{array}{c}
(\sigma ^{2}(1,1,s)+(\mu (1,1,s)-Y_{L})(\mu (1,1,s)-Y_{H})
(1-\pi _{D(1)}(s)))\pi _{D(1)}(s)/\pi _{A}(s)+ \\ 
(\sigma ^{2}(0,0,s)+(\mu (0,0,s)-Y_{H})(\mu (0,0,s)-Y_{L})\pi
_{D(0)}(s))(1-\pi _{D(0)}(s))/(1-\pi _{A}(s)) \\ 
~+ 
\left[\left(
\begin{array}{c}
\pi _{D(1)}(s)\mu (1,1,s)+(1-\pi _{D(1)}(s))Y_{L} \\ 
-(1-\pi _{D(0)}(s))\mu (0,0,s)-\pi _{D(0)}(s)Y_{H}%
\end{array}%
\right)  -\theta _{L}(\mathbf{P})\right] \times  \\ 
\left[\left( 
\begin{array}{c}
\pi _{D(1)}(s)\mu (1,1,s)+(1-\pi _{D(1)}(s))Y_{H} \\ 
-(1-\pi _{D(0)}(s))\mu (0,0,s)-\pi _{D(0)}(s)Y_{L}
\end{array}
\right) -\theta _{H}(\mathbf{P})\right]
\end{array}}
\right],\label{eq:AsyDist_ATE2} 
\end{align}
and $\sigma _{H}^{2}(\mathbf{P}) $ is as $\sigma _{L}^{2}(\mathbf{P}) $ but with $(\theta_{L}(\mathbf{P}),Y_{H})$ replaced by $(\theta_{H}(\mathbf{P}),Y_{L})$.
Moreover, $\sigma _{L}^{2}(\mathbf{P}) $ and $\sigma_{H}^{2}(\mathbf{P}) $ are positive and finite, uniformly in $\mathcal{P}_1$.
\end{theorem}
\begin{proof}
We begin by proving the convergence in \eqref{eq:AsyDist_ATE1}, uniformly in $\mathcal{P}_1$. We proceed by contradiction i.e., suppose the uniform convergence fails. Let $\zeta _{n}(\mathbf{P})$ and $\zeta (\mathbf{P})$ denote the distributions on the LHS and RHS of \eqref{eq:AsyDist_ATE1}, respectively. Then, there must exist a subsequence $\left\{
w_{n}\right\} _{n\geq 1}$ of $\left\{
n\right\} _{n\geq 1}$ s.t., for some $h\in \mathbb{R} ^{2}$ and $\{ \mathbf{P}_{w_{n}}\} _{n\geq 1}$,
\begin{equation}
\lim_{n\to \infty }\left\vert \mathbf{P}_{w_{n}}\left( \zeta _{w_{n}}\left(  \mathbf{P}_{w_{n}}\right) \leq h\right) -\mathbf{P}_{w_{n}}\left( \zeta \left( \mathbf{P}_{w_{n}}\right) \leq h\right) \right\vert >0.
\label{eq:thm:AsyDist_ATE3}
\end{equation}
We can then find a further subsequence $\left\{ k_{n}\right\} _{n\geq 1}$ of $\left\{ w_{n}\right\} _{n\geq 1}$ s.t.\ $\{(\theta _{L}(\mathbf{P}_{k_{n}}),\theta _{H}(\mathbf{P}_{k_{n}}))_{n\geq 1}$ and $\{ ( \sigma _{\mathbf{P}_{k_{n}}}^{2}(d,a,s),\mu _{\mathbf{P}_{k_{n}}}(d,a,s),\pi _{D(a),\mathbf{P}_{k_{n}}}(s),\pi _{A,\mathbf{P}_{k_{n}}}(s),p_{\mathbf{P}_{k_{n}}}(s)) :( d,a,s) \in \{ 0,1\} \times \{ 0,1\} \times S\} $ converge. Since $\sigma _{L}^{2}\left( \mathbf{P}_{k_{n}}\right) $, $\sigma _{H}^{2}\left( \mathbf{P} _{k_{n}}\right) $, and $\sigma _{HL}\left( \mathbf{P}_{k_{n}}\right) $ are continuous functions of these objects, $( \sigma _{L}^{2}\left( \mathbf{P}_{k_{n}}) ,\sigma _{L}^{2}( \mathbf{P}_{k_{n}}) ,\sigma _{HL}( \mathbf{P}_{k_{n}}\right) ) $ converges as well. Denote the respective limits by $(\theta _{L},\theta _{H})$, $\{ ( {\sigma}^{2}(d,a,s),{\mu}(d,a,s),{\pi}_{D(a)}(s),{\pi}_{A}(s),{p}(s)) :( d,a,s) \in \{ 0,1\} ^2 \times S\} $, and $\left( {\sigma}_{L}^{2},{\sigma}_{H}^{2},{\sigma}_{HL}\right) $. By some algebra, $\left( {\sigma}_{L}^{2},{\sigma}_{H}^{2},{\sigma}_{HL}\right) $ coincides with \eqref{eq:AsyDist_ATE2}.

Since the normal CDF is continuous in its variance-covariance matrix,
\begin{equation}
\lim_{n\to \infty }\mathbf{P}_{k_{n}}\left( \zeta \left( \mathbf{P}_{k_{n}}\right) \leq h\right) =P\left( {\zeta}\leq h\right) ,
\label{eq:thm:AsyDist_ATE4}
\end{equation}
where
\begin{equation*}
{\zeta}\sim N\left( {\bf 0}_{2} ,\left( 
\begin{array}{cc}
{\sigma}_{L}^{2} & {\sigma}_{HL} \\ 
{\sigma}_{HL} & {\sigma}_{H}^{2}
\end{array}
\right) \right) .
\end{equation*}

Under the current assumptions, Lemmas \ref{lem:AsyDist} and \ref{lem:expressionBounds} (with $\hat{\pi}_n(S_i) = n_A(S_i)/n(S_i)$ for all $i=1,\dots,n$) imply that
\begin{equation}
\lim_{n\to \infty }\mathbf{P}_{k_{n}}\left( \zeta _{k_{n}}\left( \mathbf{P}_{k_{n}}\right) \leq h\right) =P\left( {\zeta}\leq h\right) .
\label{eq:thm:AsyDist_ATE5}
\end{equation}
By \eqref{eq:thm:AsyDist_ATE4} and \eqref{eq:thm:AsyDist_ATE5}, 
\begin{equation}
\lim_{n\to \infty }\left\vert \mathbf{P}_{k_{n}}\left( \zeta _{k_{n}}\left( \mathbf{P}_{k_{n}}\right) \leq h\right) -\mathbf{P}_{k_{n}}\left( \zeta \left( \mathbf{P}_{k_{n}}\right) \leq h\right) \right\vert =0.
\label{eq:thm:AsyDist_ATE6}
\end{equation}
Since $\left\{ k_{n}\right\} _{n\geq 1}$ is a subsequence of $\left\{w_{n}\right\} _{n\geq 1}$, \eqref{eq:thm:AsyDist_ATE6} contradicts \eqref{eq:thm:AsyDist_ATE3}, as desired.

To complete the proof, it suffices to show that $\sigma _{L}^{2}\left( \mathbf{P}\right) $ and $\sigma _{H}^{2}(\mathbf{P}) $ are positive and finite, uniformly for all $\mathbf{P}\in \mathcal{P}_1$. By Assumptions \ref{ass:1}(c) and \ref{ass:2}(c), we have that 
\begin{align}
&p(s) \sigma ^{2}(1,1,s)\pi _{D(1)}(s)/\pi _{A}(s)\geq \xi /\left( 1-\varepsilon \right)\text{ for some } s \in \mathcal{S}\notag, \text{ or }\\
&p(s) \sigma ^{2}(0,0,s)\pi _{D(0)}(s)/\pi _{A}(s)\geq \xi /\left( 1-\varepsilon \right) \text{ for some } s \in \mathcal{S}.\label{eq:thm:AsyDist_ATE7}
\end{align}%
By \eqref{eq:AsyDist_ATE2} and \eqref{eq:thm:AsyDist_ATE7}, $\sigma _{L}^{2}(\mathbf{P}) \geq \xi /\left( 1-\varepsilon \right) >0$ and $\sigma _{H}^{2}(\mathbf{P}) \geq \xi /\left( 1-\varepsilon \right) >0$. Finally, Assumptions \ref{ass:1}(a) and \ref{ass:2}(c) imply $\sigma _{L}^{2}(\mathbf{P}) \leq ( 10/\varepsilon +12)  \max\{|Y_L|,|Y_H|\}^{2}$ and $\sigma _{H}^{2}(\mathbf{P}) \leq ( 10/\varepsilon +12) \max\{|Y_L|,|Y_H|\}^{2}$.
\end{proof}

\begin{theorem}\label{thm:AsyVarEstimation} 
For all $(d,a,s)\in \{0,1\}\times \{0,1\}\times \mathcal{S}$, 
\begin{align}
\hat{p}(s)& = n(s)/n, \notag\\
\hat{\pi}_{D(1)}(s)& = {n_{AD}(s)}/{n_{A}(s)},\notag \\
\hat{\pi}_{D(0)}(s)& ={(n_{D}(s)-n_{AD}(s))}/{(n(s)-n_{A}(s))},\notag\\
\hat{\mu}(d,a,s)& = \left\{ 
\begin{array}{cc}
\frac{\frac{1}{n}\sum_{i=1}^{n}I[D_{i}=d,A_{i}=a,S_{i}=s]Y_{i}}{\frac{1}{n}\sum_{i=1}^{n}I[D_{i}=d,A_{i}=a,S_{i}=s]} & \text{\small if $\sum_{i=1}^{n}I[D_{i}=d,A_{i}=a,S_{i}=s]>0$,} \\ 
0 &  \text{\small if $\sum_{i=1}^{n}I[D_{i}=d,A_{i}=a,S_{i}=s]=0$,} 
\end{array}%
\right. \notag\\
\hat{\sigma}^{2}(d,a,s)& = \left\{ 
\begin{array}{cc}
\frac{\frac{1}{n}\sum_{i=1}^{n}I[D_{i}=d,A_{i}=a,S_{i}=s](Y_{i}-\hat{\mu}(d,a,s))^{2}}{\frac{1}{n}\sum_{i=1}^{n}I[D_{i}=d,A_{i}=a,S_{i}=s]} & \text{\small if $\sum_{i=1}^{n}I[D_{i}=d,A_{i}=a,S_{i}=s]>0$,} \\ 
0 &  \text{\small if $\sum_{i=1}^{n}I[D_{i}=d,A_{i}=a,S_{i}=s]=0$.} 
\end{array}%
\right. \label{eq:AsyDistATE_var8}
\end{align}%
The definitions in \eqref{eq:AsyDistATE_var1} and \eqref{eq:AsyDistATE_var8} implicitly use that $n_{A}(s)/n(s)\in (0,1) $. By Lemma \ref{lem:den_not_zero}, this occurs w.p.a.1 uniformly in $\mathcal{P}_1$. Let
\begin{align}
\hat{\sigma}_{L}^{2}& =\sum_{s\in \mathcal{S}}\hat{p}(s)\left[ 
{\scriptsize\begin{array}{c}
( \hat{\sigma}^{2}(1,1,s)+(\hat{\mu}(1,1,s)-Y_{L})^{2}(1-\hat{\pi}_{D(1)}(s)))\hat{\pi}_{D(1)}(s)/(n_{A}(s)/n(s)) +\\ 
(\hat{\sigma}^{2}(0,0,s)+(\hat{\mu}(0,0,s)-Y_{H})^{2}\hat{\pi}_{D(0)}(s))(1-\hat{\pi}_{D(0)}(s))/(1-n_{A}(s)/n(s)) \\ 
+ \left[ \left(
\begin{array}{c}
\hat{\pi}_{D(1)}(s)\hat{\mu}(1,1,s)+(1-\hat{\pi}_{D(1)}(s))Y_{L} \\ 
-(1-\hat{\pi}_{D(0)}(s))\hat{\mu}(0,0,s)-\hat{\pi}_{D(0)}(s)Y_{H}
\end{array}
\right) -\hat{\theta}_{L}\right] ^{2}
\end{array}}
\right]  \notag \\
\hat{\sigma}_{HL}& =\sum_{s\in \mathcal{S}}\hat{p}(s)\left[ 
{\scriptsize\begin{array}{c}
( \hat{\sigma}^{2}(1,1,s)+(\hat{\mu}(1,1,s)-Y_{L})(\hat{\mu}(1,1,s)-Y_{H})(1-\hat{\pi}_{D(1)}(s)))  \hat{\pi}_{D(1)}(s)/(n_{A}(s)/n(s))+ \\ 
(\hat{\sigma}^{2}(0,0,s)+(\hat{\mu}(0,0,s)-Y_{H})(\hat{\mu}(0,0,s)-Y_{L})\hat{\pi}_{D(0)}(s)) (1-\hat{\pi}_{D(0)}(s))/(1-n_{A}(s)/n(s)) \\ 
+
\left[ \left(
\begin{array}{c}
\hat{\pi}_{D(1)}(s)\hat{\mu}(1,1,s)+(1-\hat{\pi}_{D(1)}(s))Y_{L} \\ 
-(1-\hat{\pi}_{D(0)}(s))\hat{\mu}(0,0,s)-\hat{\pi}_{D(0)}(s)Y_{H}%
\end{array}%
\right)-\hat{\theta}_{L} \right] \times \\ 
\left[ \left(
\begin{array}{c}
\hat{\pi}{D(1)}(s)\hat{\mu}(1,1,s)+(1-\hat{\pi}_{D(1)}(s))Y_{H} \\ 
-(1-\hat{\pi}_{D(0)}(s))\hat{\mu}(0,0,s)-\hat{\pi}_{D(0)}(s)Y_{L}
\end{array}
\right]- \hat{\theta}_{H}\right]
\end{array}}
\right] ,  \label{eq:AsyDistATE_var1}
\end{align}
and $\hat{\sigma}_{H}^{2}$ is defined as in $\hat{\sigma}_{L}^{2}$ but with $(\hat{\theta}_L,Y_{L},Y_{H}) $ replaced by $(\hat{\theta}_H,Y_{H},Y_{L})$. Then, 
\begin{equation*}
(\hat{\sigma}_{L}^{2},\hat{\sigma}_{H}^{2},\hat{\sigma}_{HL})~\overset{p}{\to  }~(\sigma _{L}^{2}(\mathbf{P}) ,\sigma _{H}^{2}(\mathbf{P}) ,\sigma _{HL}(\mathbf{P}) )
\end{equation*}
uniformly in $\mathcal{P}_1$.
\end{theorem}
\begin{proof}
Fix $\delta >0$ arbitrarily. We proceed by contradiction i.e., suppose that the desired uniform convergence fails. Then, there must exist a subsequence $\left\{ w_{n}\right\} _{n\geq 1}$ of $\mathbb{N}$ s.t., for some $h\in \mathbb{R}^{2}$ and $\{ \mathbf{P}_{w_{n}}\} _{n\geq 1}$,
\begin{equation}
\lim_{n\to  \infty }\mathbf{P}_{w_{n}}\left( \left\Vert \hat{\sigma}_{L}^{2}-\sigma _{L}^{2}\left( \mathbf{P}_{w_{n}}\right) ,\hat{\sigma}_{H}^{2}-\sigma _{H}^{2}\left( \mathbf{P}_{w_{n}}\right) ,\hat{\sigma}_{HL}-\sigma _{HL}\left( \mathbf{P}_{w_{n}}\right) \right\Vert >\delta\right) >0.  \label{eq:AsyDist3_1}
\end{equation}
We can then find a subsequence $\left\{ k_{n}\right\} _{n\geq 1}$ of $ \left\{ w_{n}\right\} _{n\geq 1}$ s.t., for all $(d,a,s)\in
\{0,1\}^{2}\times \mathcal{S}$, the convergence in \eqref{eq:keyconvergence} occurs but with $\mathbf{P}_{n}$ replaced by $\mathbf{P}_{k_{n}}$.
By Assumption \ref{ass:2}(c) and the convergence in \eqref{eq:keyconvergence},
\begin{align}
\left( \theta _{L}(\mathbf{P}_{k_{n}}),\theta _{H}(\mathbf{P} _{k_{n}})\right)  &\to  (\theta _{L},\theta _{H})
\label{eq:AsyDist3_3} \\
\left( \sigma _{L}^{2}\left( \mathbf{P}_{k_{n}}\right) ,\sigma _{L}^{2}\left( \mathbf{P}_{k_{n}}\right) ,\sigma _{HL}\left( \mathbf{P} _{k_{n}}\right) \right)  &\to  \left( {\sigma }_{L}^{2},{\sigma } _{H}^{2},{\sigma }_{HL}\right) .\label{eq:AsyDist3_4}
\end{align}

First, note that Theorem \ref{thm:AsyDist_ATE} and \eqref{eq:AsyDist3_3} implies that along $\left\{ \mathbf{P}_{k_{n}}\in \mathcal{P}\right\} _{n\geq 1}$, 
\begin{equation}
(\hat{\theta}_{L},\hat{\theta}_{H})\overset{p}{\to  }(\theta _{L},\theta _{H}).  \label{eq:AsyDist3_5}
\end{equation}

Second, Lemma \ref{lem:AsyDist} implies that $( T_{n,2},T_{n,3}) =O_{p}(1)$. This and the convergence in \eqref{eq:keyconvergence} implies that for all $\left( a,s\right) \in \{ 0,1\} \times \mathcal{S}$ and along $\left\{ \mathbf{P}_{k_{n}}\in \mathcal{P}_1\right\} _{n\geq 1}$, 
\begin{equation}
(\hat{p}(s),\hat{\pi}_{D(a)}(s))\overset{p}{\to  }(p(s),\pi _{D(a)}(s)).\label{eq:AsyDist3_6}
\end{equation}

Third, we show that for any $s\in \mathcal{S}$ and $b\in \left\{ 1,2\right\} 
$ and along $\left\{ \mathbf{P}_{k_{n}}\in \mathcal{P}_1\right\} _{n\geq 1}$,  
\begin{align}
\hat{p}(s)\hat{\pi}_{D(1)}(s)\hat{\mu}(1,1,s)^{b}&\overset{p}{\to  } p(s)\pi _{D(1)}(s)\mu (1,1,s)^{b} \notag\\
\hat{p}(s)\left( 1-\hat{\pi}_{D(0)}(s)\right) \hat{\mu}(0,0,s)^{b}&\overset{ p}{\to  }p(s)\left( 1-\pi _{D(1)}(s)\right) \mu (0,0,s)^{b}  \label{eq:AsyDist3_7}
\end{align}
We only show the first one, as the second one follows from analogous arguments. There are two cases: $\mathbf{P}(D(1)=1,S=s)=0$ and $\mathbf{P} (D(1)=1,S=s)>0 $. In the first case, $\mu (1,1,s)=0$ and $\mathbf{P}_{k_{n}}(D(1)=1,S=s)\to  0$. Since $\hat{p}(s)\hat{\pi}_{D(1)}(s)\to  p(s)\pi _{D(1)}(s)=\mathbf{P}(D(1)=1,S=s)=0$ and $\hat{\mu}(1,1,s)$ is either a weighted average of $Y_{i}\in \left[ Y_{L},Y_{H}\right] $ or zero, we get that $\hat{p}(s)\hat{\pi}_{D(1)}(s)\hat{\mu}(1,1,s)^{b}\overset{p}{ \to  }0=p(s)\pi _{D(1)}(s)\mu (1,1,s)^{b}$, as desired. In the second case, $\mathbf{P}_{k_{n}}(D(1)=1,S=s)\to  \mathbf{P}(D(1)=1,S=s)>0$. By \eqref{eq:AsyDist3_6}, Assumption \ref{ass:2}(c), and the convergence in \eqref{eq:keyconvergence}, 
\begin{equation}
\frac{n_{AD}(s) }{n_{A}(s) }\frac{n_{A}\left( s\right) }{n(s) }\frac{n(s) }{n}=\hat{\pi}_{D(1)}(s) \frac{n_{A}(s) }{n(s) }\hat{p}(s)\overset{p}{\to  }p(s)\pi _{D(1)}(s)\pi _{A}(s)>0,\label{eq:AsyDist3_8}
\end{equation}
which implies that $1\left[ n_{AD}(s) >0\right] =1+o_{p}(1) $. Therefore,%
\begin{align}
&\hat{\mu}(1,1,s)-\mu \left( 1,1,s\right) \notag \\
&\overset{(1)}{=}\left\{ 
\begin{array}{cc}
\frac{\frac{1}{n}\sum_{i=1}^{n}I[D_{i}=1,A_{i}=1,S_{i}=s]Y_{i}(1)}{n_{AD}(s) /n}-\mu \left( 1,1,s\right)  & \text{if }n_{AD}\left( s\right) >0 \\ 
-\mu \left( 1,1,s\right)  & \text{if }n_{AD}(s) =0
\end{array}
\right.   \notag \\
&\overset{(2)}{=}\tfrac{\frac{1}{n}\sum_{i=1}^{n}I[D_{i}=1,A_{i}=1,S_{i}=s] \left( Y_{i}(1)-\mu _{k_{n}}\left( 1,1,s\right) \right) }{\left( \hat{\pi}_{D(1)}(s)\frac{n_{A}(s) }{n(s) }\hat{p}(s)-p(s)\pi _{D(1)}(s)\pi _{A}(s)+p(s)\pi _{D(1)}(s)\pi _{A}(s) \right)}+\mu _{k_{n}}\left( 1,1,s\right) -\mu
\left( 1,1,s\right) +o_{p}(1)  \notag \\
&\overset{(3)}{=}o_{p}(1),\label{eq:AsyDist3_9}
\end{align}
where (1) holds by \eqref{eq:AsyDistATE_var8} and $n_{AD}(s) =\sum_{i=1}^{n}I[D_{i}=1,A_{i}=1,S_{i}=s]$, (2) by $1 \left[ n_{AD}(s) >0\right] =1+o_{p}(1)$, and (3) by \eqref{eq:AsyDist3_8}, Lemma \ref{lem:AsyDist}, and the convergence in \eqref{eq:keyconvergence}, which implies that $T_{n,1}=O_{p}(1)$. By combining \eqref{eq:AsyDist3_6} and \eqref{eq:AsyDist3_8}, we have that $\hat{p}(s)\hat{\pi}_{D(1)}(s)\hat{\mu}(1,1,s)^{b}\overset{p}{\to  }p(s)\pi _{D(1)}(s)\mu (1,1,s)^{b}$, as desired.

Fourth, we show that for any $s\in \mathcal{S}$ and along $\left\{ \mathbf{P}_{k_{n}}\right\} _{n\geq 1}$, 
\begin{align}
& \hat{p}(s)\hat{\pi}_{D(1)}(s)\hat{\sigma}^{2}(1,1,s)~\overset{p}{\to  }~p(s)\pi _{D(1)}(s)\sigma ^{2}(1,1,s)  \notag \\
& \hat{p}(s)\hat{\pi}_{D(1)}(s)\hat{\sigma}^{2}(0,0,s)~\overset{p}{\to  }~p(s)( 1-\pi _{D(1)}(s)) \sigma ^{2}(0,0,s).\label{eq:AsyDist3_10}
\end{align}%
We only show the first one, as the second one follows from analogous arguments. There are two cases: $\mathbf{P}(D(1)=1,S=s)=0$ and $\mathbf{P}(D(1)=1,S=s)>0 $. In the first case, $\sigma ^{2}(1,1,s)=0$ and $\mathbf{P} _{k_{n}}(D(1)=1,S=s)\to  0$. Since $\hat{p}(s)\hat{\pi} _{D(1)}(s)\to  p(s)\pi _{D(1)}(s)=\mathbf{P}(D(1)=1,S=s)=0$ and $\hat{ \sigma}^{2}(1,1,s)$ is either a weighted average of $(Y_{i}-\hat{\mu} (1,1,s))^{2}$ (with $(Y_{i}-\hat{\mu}(1,1,s))^{2}\leq 4\max \left\{ \left\vert Y_{L}\right\vert ,\left\vert Y_{H}\right\vert \right\} ^{2}$) or zero, we conclude that $\hat{p}(s)\hat{\pi}_{D(1)}(s)\hat{\sigma}^{2}(1,1,s)\overset{p}{\to  }0=p(s)\pi _{D(1)}(s)\hat{\sigma}^{2}(1,1,s)$, as desired. In the second case, $\mathbf{P}_{k_{n}}(D(1)=1,S=s)\to \mathbf{P}(D(1)=1,S=s)>0$. By repeating previous arguments, we conclude that $1\left[ n_{AD}(s) >0\right] =1+o_{p}(1)$. Therefore,
\begin{align}
&\hat{\sigma}^{2}(1,1,s)-\sigma ^{2}(1,1,s) \notag\\
&\overset{(1)}{=}\left\{ 
\begin{array}{cc}
\frac{\frac{1}{n}\sum_{i=1}^{n}I[D_{i}=1,A_{i}=1,S_{i}=s](Y_{i}-\hat{\mu}(1,1,s))^{2}}{n_{AD}(s) /n}-\sigma ^{2}(1,1,s) & \text{if }n_{AD}(s) >0 \\ 
-\sigma ^{2}(1,1,s) & \text{if }n_{AD}(s) =0
\end{array}
\right.   \notag \\
&\overset{(2)}{=}\tfrac{T_{k_n,4}(1,1,s)}{\left( \hat{\pi}_{D(1)}(s)\frac{n_{A}(s) }{n(s) }\hat{p}(s)-p(s)\pi _{A}(s)\pi _{D(1)}(s)+p(s)\pi _{D(1)}(s)\pi _{A}(s) \right) }-\sigma ^{2}(1,1,s) +o_{p}(1) 
\overset{(3)}{=}o_{p}(1),\label{eq:AsyDist3_11}
\end{align}
where (1) holds by \eqref{eq:AsyDistATE_var8} and $n_{AD}(s) =\sum_{i=1}^{n}I[D_{i}=1,A_{i}=1,S_{i}=s]$, (2) by $\left[ n_{AD}(s) >0\right] =1+o_{p}(1)$, and (3) by \eqref{eq:AsyDist3_8}, Lemma \ref{lem:AsyDist2}, and the convergence in \eqref{eq:keyconvergence}. By \eqref{eq:AsyDist3_6} and \eqref{eq:AsyDist3_8}, $\hat{p}(s)\hat{\pi}_{D(1)}(s)\hat{\sigma}^2(1,1,s)\overset{p}{\to  }p(s)\pi _{D(1)}(s)\sigma^2 (1,1,s)$, as desired.

Finally, by \eqref{eq:AsyDistATE_var1}, \eqref{eq:AsyDist3_5}, \eqref{eq:AsyDist3_6}, \eqref{eq:AsyDist3_7}, \eqref{eq:AsyDist3_10}, and \eqref{eq:AsyDist3_11}, we get
\begin{equation}
\lim_{n\to  \infty }\mathbf{P}_{k_{n}}\left( \left\Vert \hat{\sigma}_{L}^{2}-\sigma _{L}^{2},\hat{\sigma}_{H}^{2}-\sigma _{H}^{2},\hat{\sigma}_{HL}-\sigma _{HL}\right\Vert >\delta /2\right) =0.  \label{eq:AsyDist3_12}
\end{equation}
By the convergence in \eqref{eq:keyconvergence} and \eqref{eq:AsyDist3_12},
\begin{equation}
\lim_{n\to  \infty }\mathbf{P}_{k_{n}}\left( \left\Vert \hat{\sigma}_{L}^{2}-\sigma _{L}^{2}\left( \mathbf{P}_{k_{n}}\right) ,\hat{\sigma}_{H}^{2}-\sigma _{H}^{2}\left( \mathbf{P}_{k_{n}}\right) ,\hat{\sigma}_{HL}-\sigma _{HL}\left( \mathbf{P}_{k_{n}}\right) \right\Vert >\delta \right) =0.  \label{eq:AsyDist3_13}
\end{equation}
Since $\left\{ k_{n}\right\} _{n\geq 1}$ is a subsequence of $\left\{ w_{n}\right\} _{n\geq 1}$, \eqref{eq:AsyDist3_13} contradicts \eqref{eq:AsyDist3_1}, as desired.
\end{proof}

\subsection{ATT}

\begin{lemma}\label{lem:expressionBoundsATT} 
Let 
\begin{align}
\check{\upsilon}_{L}~& \equiv ~\frac{1}{n\check{G}}\sum_{i=1}^{n}\left[
\left( \tfrac{Y_{i}A_{i}}{\tilde{\pi}_{A}\left( S_{i}\right) }-\tfrac{%
Y_{i}\left( 1-A_{i}\right) }{1-\tilde{\pi}_{A}\left( S_{i}\right) }\right) 
\hat{\pi}_{A}\left( S_{i}\right) +\tfrac{\left( Y_{i}-Y_{H}\right)
D_{i}\left( 1-A_{i}\right) }{1-\tilde{\pi}_{A}\left( S_{i}\right) }\right]  
\notag \\
\check{\upsilon}_{H}~& \equiv ~\frac{1}{n\check{G}}\sum_{i=1}^{n}\left[
\left( \tfrac{Y_{i}A_{i}}{\tilde{\pi}_{A}\left( S_{i}\right) }-\tfrac{%
Y_{i}\left( 1-A_{i}\right) }{1-\tilde{\pi}_{A}\left( S_{i}\right) }\right) 
\hat{\pi}_{A}\left( S_{i}\right) +\tfrac{\left( Y_{i}-Y_{L}\right)
D_{i}\left( 1-A_{i}\right) }{1-\tilde{\pi}_{A}\left( S_{i}\right) }\right] ,
\label{eq:expressionBoundsATT1}
\end{align}%
where 
\begin{equation}
\check{G}~\equiv ~\frac{1}{n}\sum_{i=1}^{n}\left[ \left( \tfrac{D_{i}A_{i}}{%
\tilde{\pi}_{A}\left( S_{i}\right) }-\tfrac{D_{i}\left( 1-A_{i}\right) }{1-%
\tilde{\pi}_{A}\left( S_{i}\right) }\right) \hat{\pi}_{A}\left( S_{i}\right)
+\tfrac{D_{i}\left( 1-A_{i}\right) }{1-\tilde{\pi}_{A}\left( S_{i}\right) }%
\right] .  \label{eq:expressionBoundsATT2}
\end{equation}%
and $\{\hat{\pi}_{A}(s):s\in \mathcal{S}\}$ and $\{\tilde{\pi}_{A}(s):s\in 
\mathcal{S}\}$ are arbitrary estimators of $\{{\pi }_{A}(s):s\in \mathcal{S}%
\}$. Then, 
\begin{equation}
\sqrt{n}\left( 
\begin{array}{c}
\check{\upsilon}_{L}-\upsilon _{L}(\mathbf{P}) \\ 
\check{\upsilon}_{H}-\upsilon _{H}(\mathbf{P})%
\end{array}%
\right) ~=~\frac{1}{G}\left( 
\begin{array}{c}
U_{n}+V_{L,n}+W_{L,n}+R_{L,n} \\ 
U_{n}+V_{H,n}+W_{H,n}+R_{H,n}%
\end{array}%
\right) +\delta _{n},  \label{eq:expressionBoundsATT3}
\end{equation}%
where $\upsilon _{L}(\mathbf{P})$ and $\upsilon _{H}(\mathbf{P})$ are as in %
\eqref{eq:bounds_ATT}, $\delta _{n}=o_{p}(1)$ uniformly in $\mathbf{P}\in 
\mathcal{P}_{2}$, and 
\begin{align}
U_{n}& =\sum_{s\in \mathcal{S}}\left[ {\scriptsize 
\begin{array}{c}
T_{n,1}(1,1,s)+T_{n,1}(0,1,s)+T_{n,1}(1,0,s)-\pi
_{A}(s)T_{n,1}(0,0,s)/(1-\pi _{A}(s)) \\ 
+p(s)\pi _{A}(s)[\mu (1,1,s)-\mu (0,1,s)]T_{n,2}(1,s) \\ 
+p(s)[(1-\pi _{A}(s))\mu (1,0,s)+\pi _{A}(s)\mu (0,0,s)]T_{n,2}(2,s)+ \\ 
\left[ 
\begin{array}{c}
\left[ 
\begin{array}{c}
\mu (1,1,s)\pi _{D(1)}(s)+\mu (0,1,s)(1-\pi _{D(1)}(s)) \\ 
-\mu (1,0,s)\pi _{D(0)}(s)-\mu (0,0,s)(1-\pi _{D(0)}(s))%
\end{array}%
\right] \pi _{A}(s) \\ 
+\mu (1,0,s)\pi _{D(0)}(s)%
\end{array}%
\right] T_{n,3}(s)%
\end{array}%
}\right]   \notag \\
V_{L,n}& =-\sum_{s\in \mathcal{S}}\left[ {\scriptsize 
\begin{array}{c}
p(s)\pi _{A}(s)\frac{N_{L}}{G}T_{n,2}(1,s)+p(s)[Y_{H}+\frac{N_{L}}{G}(1-\pi
_{A}(s))]T_{n,2}(2,s)+ \\ 
\big[Y_{H}\pi _{D(0)}(s)+\frac{N_{L}}{G}(\pi _{D(1)}(s)\pi _{A}(s)+\pi
_{D(0)}(s)(1-\pi _{A}(s)))\big]T_{n,3}(s)%
\end{array}%
}\right]   \notag \\
W_{L,n}& =\sum_{s\in \mathcal{S}}p(s)\sqrt{n}(\hat{\pi}_{A}(s)-\pi _{A}(s))%
\left[ {\scriptsize 
\begin{array}{c}
\mu (1,1,s)\pi _{D(1)}(s)+\mu (0,1,s)(1-\pi _{D(1)}(s)) \\ 
-\mu (1,0,s)\pi _{D(0)}(s)-\mu (0,0,s)(1-\pi _{D(0)}(s)) \\ 
-\frac{N_{L}}{G}(\pi _{D(1)}(s)-\pi _{D(0)}(s))%
\end{array}%
}\right]   \notag \\
R_{L,n}& =\sum_{s\in \mathcal{S}}p(s)\sqrt{n}(\tilde{\pi}_{A}(s)-\tfrac{%
n_{A}(s)}{n\left( s\right) })\left[ {\scriptsize 
\begin{array}{c}
-\mu (1,1,s)\pi _{D(1)}(s)-\mu (0,1,s)(1-\pi _{D(1)}(s))+ \\ 
\mu (1,0,s)\pi _{D(0)}(s)-Y_{H}\pi _{D\left( 0\right) }\left( s\right)
/(1-\pi _{A}\left( s\right) ) \\ 
-\mu (0,0,s)(1-\pi _{D(0)}(s))\pi _{A}\left( s\right) /(1-\pi _{A}\left(
s\right) ) \\ 
+\frac{N_{L}}{G}\left( \pi _{D\left( 1\right) }\left( s\right) -\pi
_{D\left( 0\right) }\left( s\right) \right) 
\end{array}%
}\right]   \notag \\
N_{L}& =\sum_{s\in \mathcal{S}}p(s)\left[ {\scriptsize 
\begin{array}{c}
\left[ 
\begin{array}{c}
\mu (1,1,s)\pi _{D(1)}(s)+\mu (0,1,s)(1-\pi _{D(1)}(s)) \\ 
-\mu (1,0,s)\pi _{D(0)}(s)-\mu (0,0,s)(1-\pi _{D(0)}(s))%
\end{array}%
\right] \pi _{A}(s) \\ 
+(\mu (1,0,s)-Y_{H})\pi _{D(0)}(s)%
\end{array}%
}\right]   \notag \\
G& =\sum_{s\in \mathcal{S}}p(s)[\pi _{D(1)}(s)\pi _{A}(s)+\pi
_{D(0)}(s)(1-\pi _{A}(s))],  \label{eq:expressionBoundsATT4}
\end{align}%
with $(T_{n,1},T_{n,2},T_{n,3})$ as in \eqref{eq:Tn_defn}, and $V_{H,n}$, $%
W_{H,n}$, $R_{H,n}$ are defined as in $V_{L,n}$, $W_{L,n}$, $R_{L,n}$ but
with $(N_{L},Y_{H})$ replaced by $(N_{H},Y_{L})$.
\end{lemma}
\begin{proof}
We only show the result for $\sqrt{n}(\check{\upsilon}_{L}-\upsilon _{L}(\mathbf{P}))$, as the one for $\sqrt{n}(\check{\upsilon}_{H}-\upsilon _{H}(\mathbf{P}))$ is analogous. Note that $\check{\upsilon}_{L}=\check{N}_{L}/\check{G}$ and $\upsilon _{L}(\mathbf{P})=\tilde{N}_{L}/\tilde{G}$, 
\begin{align*}
\check{N}_{L}& ~\equiv~ \sum_{s\in \mathcal{S}}\tfrac{n(s) }{n}\left[ 
\begin{array}{c}
\left( \tfrac{\sum_{i=1}^{n}1[A_{i}=1,S_{i}=s]Y_{i}}{n(s)\tilde{\pi}_{A}(s) }-\tfrac{\sum_{i=1}^{n}1[A_{i}=0,S_{i}=s]Y_{i}}{n(s)(1-\tilde{\pi}_{A}(s)}\right) \hat{\pi}_{A}(s) \\ 
+\tfrac{\sum_{i=1}^{n}\left( Y_{i}-Y_{H}\right) 1[D_{i}=1,A_{i}=0,S_{i}=s]}{n(s)(1-\tilde{\pi}_{A}(s)}
\end{array}
\right] \\
\check{G}& ~\equiv~ \sum_{s\in \mathcal{S}}\tfrac{n(s) }{n}\left[\left( \tfrac{n_{AD}(s) }{n(s)\tilde{\pi}_{A}(s) }-\tfrac{n_{D}(s) -n_{AD}(s) }{n(s)(1-\tilde{\pi}_{A}(s)}\right) \hat{\pi}_{A}(s) +\tfrac{n_{D}(s) -n_{AD}(s) }{n(s)(1-\tilde{\pi}_{A}\left(s\right) )}\right]
\end{align*}
and 
\begin{align*}
\tilde{N}_{L}& ~=~E\left[
(Y_{i}(1)-Y_{i}(0))(D_{i}(1)-D_{i}(0))A_{i}+(Y_{i}(1)-Y_{H})D_{i}(0)\right] \\
\tilde{G}& ~=~E[(D_{i}(1)-D_{i}(0))A_{i}+D_{i}(0)].
\end{align*}%
Also, we have 
\begin{align}
N_{L}& ~\overset{(1)}{=}~E\left[ {
\begin{array}{c}
E[Y_{i}(1)-Y_{i}(0)|D_{i}(1)>D_{i}(0),S_{i}]P(D_{i}(1)>D_{i}(0)|S_{i})\pi _{A}(S_{i}) \\ 
+E[Y_{i}(1)-Y_{H}|D_{i}(0)=1,S_{i}]P(D_{i}(0)=1|S_{i}) 
\end{array}
}\right] \notag \\
G& ~\overset{(2)}{=}~E\left[ P(D_{i}(1)>D_{i}(0)|S_{i})\pi
_{A}(S_{i})+P(D_{i}(0)=1|S_{i})\right] ,  \label{eq:expressionBoundsATT6}
\end{align}%
where (1) and (2) hold by the i.i.d.\ condition in Assumption \ref{ass:1} and \eqref{eq:keyDefns}.

We now proceed to show \eqref{eq:expressionBoundsATT1}. First, note that
\begin{align*}
\tilde{N}_{L}~&\overset{(1)}{=}~E\left[ {
\begin{array}{c}
E[Y_{i}(1)-Y_{i}(0)|D_{i}(1)>D_{i}(0),S_{i}]P(D_{i}(1)>D_{i}(0)|S_{i})P(A_{i}=1|S_{i}) \\ 
+E[Y_{i}(1)-Y_{H}|D_{i}(0)=1,S_{i}]P(D_{i}(0)=1|S_{i})
\end{array}
}\right]  \\
\tilde{G}~&\overset{(2)}{=}~E\left[ P(D_{i}(1)>D_{i}(0)|S_{i})P(A_{i}=1|S_{i})+P(D_{i}(0)=1|S_{i})\right] .
\end{align*}%
where (1) and (2) hold because $W_{i}\perp A_{i}|S_{i}$, as shown in the proof of Theorem \ref{thm:mainATT}. Second, we note that 
\begin{equation}
\sqrt{n}({\check{N}_{L}}/{\check{G}}-{\tilde{N}_{L}}/{\tilde{G}})~=~\sqrt{n}({\check{N}_{L}}/{\check{G}}-{N_{L}}/{G})+\delta _{n,1},
\label{eq:expressionBoundsATT7}
\end{equation}%
where $\delta _{n,1}=\sqrt{n}({N_{L}}/{G}-\tilde{N}_{L}/\tilde{G})=o_{p}(1)$, uniformly in $\mathbf{P}\in \mathcal{P}_{2}$. To get this result, we use
Assumptions \ref{ass:1}(a), (d) and \ref{ass:2}(b) and (e). 

Third, we note that
\begin{equation}
\sqrt{n}({\check{N}_{L}}/{\check{G}}-{N_{L}}/{G})~=~[\sqrt{n}(\check{N}_{L}-N_{L})-(N_{L}/G)\sqrt{n}(\check{G}-G)]/G+\delta _{n,2},
\label{eq:expressionBoundsATT8}
\end{equation}%
where $\delta _{n,2}\equiv [ \sqrt{n}(\check{N}_{L}-{N_{L}})-({N_{L}} /G)\sqrt{n}(\check{G}-G)](1/\check{G}-1/G)=o_{p}(1)$, uniformly in $\mathbf{P%
}\in \mathcal{P}_{2}$ (based on Lemma \ref{lem:AsyDist}). 

Fourth, a lengthy derivation shows that 
\begin{equation}
\sqrt{n}(\check{G}-G)~=~\sum_{s\in \mathcal{S}}\left[ 
\begin{array}{c}
T_{n,3}(s)\left[ \pi _{D(1) }(s) \pi _{A}(s) +\pi _{D(0) }(s) ( 1-\pi _{A}(s) ) \right]  \\ 
+T_{n,2}(s,1)p(s) \pi _{A}(s)  
+T_{n,2}(s,2)p(s) ( 1-\pi _{A}(s) )  \\ 
-p(s) \left( \pi _{D(1) }(s) -\pi _{D(0) }(s) \right) \sqrt{n}\left( \tilde{\pi} _{A}(s) -n_{A}(s) /n(s)\right)  \\ 
+p(s) \left( \pi _{D(0) }(s) -\pi _{D(1) }(s) \right) \sqrt{n}\left( \hat{\pi} _{A}(s) -\pi _{A}(s)\right) 
\end{array}%
\right] +\delta _{n,3},
\label{eq:expressionBoundsATT9}
\end{equation}%
where $\delta _{n,3}=o_{p}(1)$, uniformly in $\mathbf{P}\in \mathcal{P}_{2}$ (based on Lemma \ref{lem:AsyDist}).

Finally, another lengthy derivation shows that
\begin{equation}
\sqrt{n}(\check{N}_{L}-N_{L})~=~\sum_{s\in \mathcal{S}}\left\{ {\scriptsize 
\begin{array}{c}
T_{n,1}(1,1,s)+T_{n,1}(0,1,s)+T_{n,1}(1,0,s)-\frac{\pi _{A}(s)T_{n,1}(0,0,s) }{(1-\pi _{A}(s))}+ \\ 
p(s)\pi _{A}(s)(\mu (1,1,s)-\mu (0,1,s))T_{n,2}(s,1)+ \\ 
p(s)[(1-\pi _{A}(s))\mu (1,0,s)+\pi _{A}(s)\mu (0,0,s)-Y_{H}]T_{n,2}(s,2)+ \\ 
\left[ 
\begin{array}{c}
\left[ 
\begin{array}{c}
\mu (1,1,s)\pi _{D(1)}(s)+\mu (0,1,s)(1-\pi _{D(1)}(s)) \\ 
-\mu (1,0,s)\pi _{D(0)}(s)-\mu (0,0,s)(1-\pi _{D(0)}(s))%
\end{array}%
\right] \pi _{A}(s) \\ 
+(\mu (1,0,s)-Y_{H})\pi _{D(0)}(s)%
\end{array}%
\right] T_{n,3}(s) \\ 
+p(s)\left[ 
\begin{array}{c}
\mu (1,1,s)\pi _{D(1)}(s)+\mu (0,1,s)(1-\pi _{D(1)}(s)) \\ 
-\mu (1,0,s)\pi _{D(0)}(s)-\mu (0,0,s)(1-\pi _{D(0)}(s))
\end{array}%
\right] \times  \\ 
\sqrt{n}(\hat{\pi}_{A}(s)-\pi _{A}(s)) \\ 
+p(s)\left[ 
\begin{array}{c}
-\mu _{\mathbf{P}}\left( 1,1,s\right) \pi _{D(1) }(s) -\mu _{\mathbf{P}}\left( 0,1,s\right) \left( 1-\pi _{D(1) }(s) \right)  \\ 
+\mu _{\mathbf{P}}\left( 1,0,s\right) \pi _{D(0) }(s) -\mu _{\mathbf{P}}\left( 0,0,s\right) \frac{\pi _{A}(s)
\left( 1-\pi _{D(0) }(s) \right) }{(1-\pi _{A}(s) )} \\ 
-Y_{H}\frac{\pi _{D(0) }(s) }{(1-\pi _{A}(s) )}%
\end{array}%
\right] \times  \\ 
\sqrt{n}(\tilde{\pi}_{A}(s)-n_{A}(s)/n(s)%
\end{array}%
}\right\} \text{{\small $+\delta _{n,4}$}},
\label{eq:expressionBoundsATT10}
\end{equation}%
where $\delta _{n,4}=o_{p}(1)$, uniformly in $\mathbf{P}\in \mathcal{P}_{2}$ (based on Lemma \ref{lem:AsyDist}). By setting $\delta _{n}=\delta _{n,1}+\delta _{n,2}-\delta _{n,3}N_L/G^2+\delta _{n,4}/G$ and combining \eqref{eq:expressionBoundsATT6}, \eqref{eq:expressionBoundsATT7}, \eqref{eq:expressionBoundsATT8}, \eqref{eq:expressionBoundsATT9}, and \eqref{eq:expressionBoundsATT10}, \eqref{eq:expressionBoundsATT1} follows.
\end{proof}

\begin{theorem}\label{thm:AsyDist_ATT} 
We have 
\begin{equation}
\sqrt{n}\left( 
\begin{array}{c}
\hat{\upsilon}_{L}-\upsilon _{L}(\mathbf{P}) \\ 
\hat{\upsilon}_{H}-\upsilon _{H}(\mathbf{P})
\end{array}
\right) ~\overset{d}{\to  }~
 N\left( {\bf 0 }_2,
\frac{1}{G^{2}} \left( 
\begin{array}{cc}
\varpi _{L}^{2}(\mathbf{P})  & \varpi _{HL}\left( \mathbf{P}%
\right)  \\ 
\varpi _{HL}(\mathbf{P})  & \varpi _{H}^{2}\left( \mathbf{P}%
\right) 
\end{array}%
\right) \right) ,  \label{eq:AsyDist_ATT}
\end{equation}%
uniformly in $\mathcal{P}_2$, where 
\begin{align*}
\varpi _{L}^{2}(\mathbf{P})  
&=\sum_{s\in \mathcal{S}}p(s)\left[
{\scriptsize
\begin{array}{c}
\pi _{D(1)}(s)\pi _{A}(s)\sigma ^{2}(1,1,s)+\left( 1-\pi _{D(1)}(s)\right)
\pi _{A}(s)\sigma ^{2}\left( 0,1,s\right) \\ 
+\pi _{D(0)}(s)(1-\pi _{A}(s))\sigma ^{2}\left( 1,0,s\right)\\
+(\pi _{A}(s))^2 (1-\pi _{D(0)}(s))\sigma ^{2}\left( 0,0,s\right) /(1-\pi _{A}(s)) \\ 
+\pi _{A}(s) \left[ \mu \left( 1,1,s\right) -\mu \left(
0,1,s\right) -(N_{L}/G)\right] ^{2}(1-\pi _{D(1)}(s))\pi _{D(1)}(s) \\ 
+\left[ 
\begin{array}{c}
\left( 1-\pi _{A}(s) \right) \mu \left( 1,0,s\right) +\pi
_{A}(s) \mu \left( 0,0,s\right)  \\ 
-Y_{H}-(N_{L}/G)\left( 1-\pi _{A}(s) \right) 
\end{array}
\right] ^{2}\frac{(1-\pi _{D(0)}(s))\pi _{D(0)}(s)}{(1-\pi _{A}(s))} \\ 
+\left[ 
\begin{array}{c}
\left[ 
\begin{array}{c}
\mu \left( 1,1,s\right) \pi _{D(1)}(s) +\mu \left( 0,1,s\right)
\left( 1-\pi _{D(1)}(s) \right)  \\ 
-\mu \left( 1,0,s\right) \pi _{D(0) }(s) -\mu \left(
0,0,s\right) \left( 1-\pi _{D(0) }(s) \right) 
\end{array}
\right] \pi _{A}(s)  \\ 
+(\mu \left( 1,0,s\right) -Y_{H})\pi _{D(0) }(s)  \\ 
-\left[ (N_{L}/G)\left[ \pi _{D(1)}(s) \pi _{A}(s)
+\pi _{D(0) }(s) \left( 1-\pi _{A}(s)
\right) \right] \right] 
\end{array}
\right] ^{2} 
\end{array}}
\right]  
\end{align*}
and 
\begin{align*}
\varpi _{HL}(\mathbf{P}) &=\sum_{s\in \mathcal{S}}p(s)\left[ 
{\scriptsize
\begin{array}{c}
\pi _{D(1)}(s)\pi _{A}(s)\sigma ^{2}(1,1,s)+\left( 1-\pi _{D(1)}(s)\right)
\pi _{A}(s)\sigma ^{2}\left( 0,1,s\right) \\ 
+\pi _{D(0)}(s)(1-\pi _{A}(s))\sigma ^{2}\left( 1,0,s\right)\\
+(\pi _{A}(s))^2 (1-\pi _{D(0)}(s))\sigma ^{2}\left( 0,0,s\right) /(1-\pi _{A}(s)) \\
+\pi _{A}(s) \left[ 
\begin{array}{c}
\left[ \mu \left( 1,1,s\right) -\mu \left( 0,1,s\right) -(N_{L}/G)\right]
\times  \\ 
\left[ \mu \left( 1,1,s\right) -\mu \left( 0,1,s\right) -(N_{H}/G)\right] 
\end{array}
\right] (1-\pi _{D(1)}(s))\pi _{D(1)}(s) \\ 
+\left[ 
\begin{array}{c}
\left( 1-\pi _{A}(s) \right) \mu \left( 1,0,s\right) +\pi
_{A}(s) \mu \left( 0,0,s\right)  \\ 
-Y_{L}-(N_{H}/G)\left( 1-\pi _{A}(s) \right) 
\end{array}
\right] \times  \\ 
\left[ 
\begin{array}{c}
\left( 1-\pi _{A}(s) \right) \mu \left( 1,0,s\right) +\pi
_{A}(s) \mu \left( 0,0,s\right)  \\ 
-Y_{H}-(N_{L}/G)\left( 1-\pi _{A}(s) \right) 
\end{array}
\right] \frac{(1-\pi _{D(0)}(s))\pi _{D(0)}(s)}{(1-\pi _{A}(s))} \\ 
+\left[ 
\begin{array}{c}
\left[ 
\begin{array}{c}
\mu \left( 1,1,s\right) \pi _{D(1)}(s) +\mu \left( 0,1,s\right)
\left( 1-\pi _{D(1)}(s) \right)  \\ 
-\mu \left( 1,0,s\right) \pi _{D(0) }(s) -\mu \left(
0,0,s\right) \left( 1-\pi _{D(0) }(s) \right) 
\end{array}
\right] \pi _{A}(s)  \\ 
+(\mu \left( 1,0,s\right) -Y_{H})\pi _{D(0) }(s)  \\ 
-\left[ (N_{L}/G)\left[ \pi _{D(1)}(s) \pi _{A}(s)
+\pi _{D(0) }(s) \left( 1-\pi _{A}(s)
\right) \right] \right] 
\end{array}
\right] \times  \\ 
\left[ 
\begin{array}{c}
\left[ 
\begin{array}{c}
\mu \left( 1,1,s\right) \pi _{D(1)}(s) +\mu \left( 0,1,s\right)
\left( 1-\pi _{D(1)}(s) \right)  \\ 
-\mu \left( 1,0,s\right) \pi _{D(0) }(s) -\mu \left(
0,0,s\right) \left( 1-\pi _{D(0) }(s) \right) 
\end{array}
\right] \pi _{A}(s)  \\ 
+(\mu \left( 1,0,s\right) -Y_{L})\pi _{D(0) }(s)  \\ 
-\left[ (N_{H}/G)\left[ \pi _{D(1)}(s) \pi _{A}(s)
+\pi _{D(0) }(s) \left( 1-\pi _{A}(s)
\right) \right] \right] 
\end{array}
\right]
\end{array}
}
\right] ,
\end{align*}
where $( N_{L},N_{H},G) $ is as in \eqref{eq:expressionBoundsATT4}, and $\varpi _{H}^{2}(\mathbf{P}) $ is as $\varpi _{L}^{2}(\mathbf{P}) $ but with $(N_{L},Y_{H})$ replaced by $(N_{H},Y_{L})$. 
Moreover, $\varpi _{L}^{2}(\mathbf{P}) $ and $\varpi _{H}^{2}(\mathbf{P}) $ are positive and finite, uniformly in $\mathcal{P}_2$.
\end{theorem}
\begin{proof}
This proof is analogous to that of Theorem \ref{thm:AsyDist_ATE}, and therefore omitted. 
\end{proof}

\begin{proof}[Proof of Theorem \ref{thm:consist_ATT}]
    This proof is analogous to that of Theorem \ref{thm:consist}, and therefore omitted.
\end{proof}

\begin{theorem}\label{thm:AsyVarEstimation_ATT} 
Let $(\hat{p}(s),\hat{\pi}_{D(a)}(s),\hat{\mu}(d,a,s),\hat{\sigma} ^{2}(d,a,s)) : (d,a,s)\in \{0,1\}\times \{0,1\}\times \mathcal{S}) $ be as in \eqref{eq:AsyDistATE_var8}, and let
\begin{align}
\hat{\varpi}_{L}^{2} &=\sum_{s\in \mathcal{S}}\hat{p}(s)\left[ 
{\scriptsize
\begin{array}{c}
\hat{\pi}_{D(1)}(s){\pi}_{A}(s)\hat{\sigma}^{2}(1,1,s)+\left( 1-\hat{\pi}_{D(1)}(s)\right) {\pi}_{A}(s)\hat{\sigma}^{2}\left( 0,1,s\right) + \\ 
\hat{\pi}_{D(0)}(s)(1-{\pi}_{A}(s))\hat{\sigma}^{2}\left( 1,0,s\right) +({\pi}_{A}(s))^{2}(1-\hat{\pi}_{D(0)}(s))\hat{\sigma}^{2}\left(0,0,s\right) /(1-{\pi}_{A}(s) ) \\
+{\pi}_{A}(s)\left[ \hat{\mu}\left( 1,1,s\right) -\hat{\mu}\left(0,1,s\right) -\hat{\upsilon}_{L}\right] ^{2}(1-\hat{\pi}_{D(1)}(s))\hat{\pi}_{D(1)}(s) \\ 
+\left[ 
\begin{array}{c}
\left( 1-{\pi}_{A}(s) \right) \hat{\mu}\left( 1,0,s\right) +{\pi}_{A}(s) \hat{\mu}\left( 0,0,s\right)  \\ 
-Y_{H}-\hat{\upsilon}_{L}\left( 1-{\pi}_{A}(s) \right) 
\end{array}
\right] ^{2}\frac{(1-\hat{\pi}_{D(0)}(s))\hat{\pi}_{D(0)}(s)}{(1-{\pi}_{A}(s))} \\ 
+\left[ 
\begin{array}{c}
\left[ 
\begin{array}{c}
\hat{\mu}\left( 1,1,s\right) \hat{\pi}_{D(1)}(s) +\hat{\mu}\left( 0,1,s\right) \left( 1-\hat{\pi}_{D(1)}(s) \right)  \\ 
-\hat{\mu}\left( 1,0,s\right) \hat{\pi}_{D(0) }(s) -\hat{\mu}\left( 0,0,s\right) \left( 1-\hat{\pi}_{D(0) }\left(s\right) \right) 
\end{array}
\right] {\pi}_{A}(s)  \\ 
+(\hat{\mu}\left( 1,0,s\right) -Y_{H})\hat{\pi}_{D(0) }\left(s\right)  \\ 
-\left[ \hat{\upsilon}_{L}\left[ \hat{\pi}_{D(1)}(s) {\pi}_{A}(s) +\hat{\pi}_{D(0) }(s) \left( 1-{\pi}_{A}(s) \right) \right] \right] 
\end{array}
\right] ^{2} 
\end{array}}
\right] , \notag\\
\hat{\varpi}_{HL}&=\sum_{s\in \mathcal{S}}\hat{p}(s)\left[ 
{\scriptsize\begin{array}{c}
\hat{\pi}_{D(1)}(s){\pi}_{A}(s)\hat{\sigma}^{2}(1,1,s)+\left( 1-\hat{\pi}%
_{D(1)}(s)\right) {\pi}_{A}(s)\hat{\sigma}^{2}\left( 0,1,s\right) + \\ 
\hat{\pi}_{D(0)}(s)(1-{\pi}_{A}(s))\hat{\sigma}^{2}\left( 1,0,s\right) 
+({\pi}_{A}(s))^{2}(1-\hat{\pi}_{D(0)}(s))\hat{\sigma}^{2}\left(
0,0,s\right) /(1-{\pi}_{A}(s) ) \\ 
+{\pi}_{A}(s)\left[ 
\begin{array}{c}
\left[ \hat{\mu}\left( 1,1,s\right) -\hat{\mu}\left( 0,1,s\right) -\hat{%
\upsilon}_{L}\right]  \\ 
\times \left[ \hat{\mu}\left( 1,1,s\right) -\hat{\mu}\left( 0,1,s\right) -%
\hat{\upsilon}_{H}\right] 
\end{array}%
\right] (1-\hat{\pi}_{D(1)}(s))\hat{\pi}_{D(1)}(s) \\ 
+\left[ 
\begin{array}{c}
\left( 1-{\pi}_{A}(s) \right) \hat{\mu}\left( 1,0,s\right) +%
{\pi}_{A}(s) \hat{\mu}\left( 0,0,s\right)  \\ 
-Y_{H}-\hat{\upsilon}_{L}\left( 1-{\pi}_{A}(s) \right) 
\end{array}%
\right] \times  \\ 
\left[ 
\begin{array}{c}
\left( 1-{\pi}_{A}(s) \right) \hat{\mu}\left( 1,0,s\right) +%
{\pi}_{A}(s) \hat{\mu}\left( 0,0,s\right)  \\ 
-Y_{L}-\hat{\upsilon}_{H}\left( 1-{\pi}_{A}(s) \right) 
\end{array}%
\right] \frac{(1-\hat{\pi}_{D(0)}(s))\hat{\pi}_{D(0)}(s)}{(1-\hat{\pi}%
_{A}(s))} \\ 
+\left[ 
\begin{array}{c}
\left[ 
\begin{array}{c}
\hat{\mu}\left( 1,1,s\right) \hat{\pi}_{D(1)}(s) +\hat{\mu}%
\left( 0,1,s\right) \left( 1-\hat{\pi}_{D(1)}(s) \right)  \\ 
-\hat{\mu}\left( 1,0,s\right) \hat{\pi}_{D(0) }(s) -%
\hat{\mu}\left( 0,0,s\right) \left( 1-\hat{\pi}_{D(0) }\left(
s\right) \right) 
\end{array}%
\right] {\pi}_{A}(s)  \\ 
+(\hat{\mu}\left( 1,0,s\right) -Y_{H})\hat{\pi}_{D(0) }\left(
s\right)  \\ 
-\left[ \hat{\upsilon}_{L}\left[ \hat{\pi}_{D(1)}(s) {\pi}%
_{A}(s) +\hat{\pi}_{D(0) }(s) \left( 1-%
{\pi}_{A}(s) \right) \right] \right] 
\end{array}%
\right]  \\ 
\times \left[ 
\begin{array}{c}
\left[ 
\begin{array}{c}
\hat{\mu}\left( 1,1,s\right) \hat{\pi}_{D(1)}(s) +\hat{\mu}%
\left( 0,1,s\right) \left( 1-\hat{\pi}_{D(1)}(s) \right)  \\ 
-\hat{\mu}\left( 1,0,s\right) \hat{\pi}_{D(0) }(s) -%
\hat{\mu}\left( 0,0,s\right) \left( 1-\hat{\pi}_{D(0) }\left(
s\right) \right) 
\end{array}%
\right] {\pi}_{A}(s)  \\ 
+(\hat{\mu}\left( 1,0,s\right) -Y_{L})\hat{\pi}_{D(0) }\left(
s\right)  \\ 
-\left[ \hat{\upsilon}_{H}\left[ \hat{\pi}_{D(1)}(s) {\pi}_{A}(s) +\hat{\pi}_{D(0) }(s) \left( 1-%
{\pi}_{A}(s) \right) \right] \right] 
\end{array}%
\right] 
\end{array}}
\right] , \label{eq:AsyVar_ATT}
\end{align}
and $\hat{\varpi}_{H}^{2}$ is defined as in $\hat{\varpi}_{L}^{2}$ but with $(\hat{\upsilon}_{L},Y_{L})$ replaced by $(\hat{\upsilon}_{H},Y_{H})$. Then, 
\begin{equation*}
(\hat{\varpi}_{L}^{2},\hat{\varpi}_{H}^{2},\hat{\varpi}_{HL})/\hat{G}^2~\overset{p}{\to  }~(\varpi _{L}^{2}(\mathbf{P)},\varpi _{H}^{2}(\mathbf{P)},\varpi _{HL}(\mathbf{P)})/G^2
\end{equation*}%
uniformly in $\mathcal{P}_2$.
\end{theorem}
\begin{proof}
This proof is analogous to that of Theorem \ref{thm:AsyVarEstimation}, and therefore omitted. 
\end{proof}

\begin{proof}[Proof of Theorem \ref{thm:StoyeATT}.]
The proof is analogous to that of Theorem \ref{thm:Stoye}, and therefore omitted.
\end{proof}

\section{Appendix on alternative bounds estimators}\label{sec:appendixD}

\begin{lemma}\label{lem:AsyDist_ext}
Consider any sequence of $\{ \mathbf{P}_{n}\in \mathcal{P}_3\} _{n\geq 1}$ that satisfies \eqref{eq:keyconvergence} for all $(d,a,s)\in \{0,1\}^{2}\times \mathcal{S}$. Then, along $\{ \mathbf{P}_{n}\} _{n\geq 1}$, 
\begin{equation*}
\left( 
\begin{array}{c}
T_{n,1} \\ 
T_{n,2} \\ 
T_{n,3} \\
T_{n,A}
\end{array}
\right) ~\overset{d}{\to }~N\left( \left( 
\begin{array}{c}
\mathbf{0} \\ 
\mathbf{0} \\ 
\mathbf{0} \\
\mathbf{0} 
\end{array}%
\right) ,\left( 
\begin{array}{cccc}
\Sigma _{1} & \mathbf{0} & \mathbf{0}& \mathbf{0} \\ 
\mathbf{0} & \Sigma _{2} & \mathbf{0}& \mathbf{0}\\ 
\mathbf{0} & \mathbf{0} & \Sigma _{3}& \mathbf{0}\\
\mathbf{0} & \mathbf{0} &  \mathbf{0}&\Sigma _{A}
\end{array}%
\right) \right) ,
\end{equation*}%
where $(T_{n,1}',T_{n,2}',T_{n,3}',T_{n,A}')$ is as in \eqref{eq:Tn_defn} but with distribution $\mathbf{P}_{n}$, $\Sigma _{1}$, $\Sigma _{2}$, and $\Sigma _{3}$ are as in Lemma \ref{lem:AsyDist}, and $\Sigma _{A}\equiv diag( \tau ( s) ( 1-\pi _{A}( s) ) \pi _{A}( s)/ p( s) :s\in \mathcal{S}).$
\end{lemma}
\begin{proof}
    This proof follows from extending the arguments in Lemma \ref{lem:AsyDist}. The asymptotic analysis of $T_{n,A}$ requires Assumption \ref{ass:2}(f), which explains why we use $\mathcal{P}_3$. 
\end{proof}

\begin{theorem}\label{thm:AsyDist_ATEother}
Let $( \tilde{\theta}_{L},\tilde{\theta}_{H}) $ denote the alternative bounds estimator defined as in \eqref{eq:BoundsHat} but with $\left\{ n_{A}(S_{i})/n(S_i) :i=1,\ldots ,n\right\} $ replaced by $\left\{ \pi _{A}(S_{i}) :i=1,\ldots ,n\right\} $. Then, 
\begin{equation*}
\sqrt{n}\left(
\begin{array}{c}
\tilde{\theta}_{L}-\theta _{L}(\mathbf{P}) \\ 
\tilde{\theta}_{H}-\theta _{H}(\mathbf{P})
\end{array}%
\right) ~\overset{d}{\to}~N\left( {\bf 0}_2 
 ,\left( 
{ \begin{array}{cc}
\sigma _{L}^{2}(\mathbf{P}) & \sigma _{HL}\left( \mathbf{P}\right)  \\ 
\sigma _{HL}(\mathbf{P}) & \sigma _{H}^{2}\left( \mathbf{P}\right) 
\end{array}}
\right) +\left( 
{ \begin{array}{cc}
\Delta _{L}^{2}(\mathbf{P}) & \Delta _{HL}\left( \mathbf{P}\right)  \\ 
\Delta _{HL}(\mathbf{P}) & \Delta _{H}^{2}\left( \mathbf{P}\right) 
\end{array}}
\right) \right) , 
\end{equation*}
uniformly in $\mathcal{P}_3$, where $(\sigma _{H}^{2}(\mathbf{P}),\sigma _{L}^{2}(\mathbf{P}),\sigma _{HL}(\mathbf{P}) )$ is as in Theorem \ref{thm:AsyDist_ATE},
\begin{align*}
\Delta _{L}^{2}(\mathbf{P}) &~=~\sum_{s\in \mathcal{S}}p(s)\tau (s) ( 1-\pi _{A}(s) ) \pi _{A}(s) 
\left[ 
{\scriptsize\begin{array}{c}
\pi _{D(1)}(s)\mu (1,1,s)/\pi _{A}(s)+ \\ 
(1-\pi _{D(0) }(s))\mu (0,0,s)/(1-\pi _{A}(s) \\ 
+(1-\pi _{D(1)}(s))Y_{L}/\pi _{A}(s) \\ 
+\pi _{D(0) }(s)Y_{H}/(1-\pi _{A}(s))
\end{array}}
\right] ^{2}  \notag \\
\Delta _{HL}(\mathbf{P})& ~=~\sum_{s\in \mathcal{S}}p(s)\tau (s)
( 1-\pi _{A}(s) ) \pi _{A}(s) \times \\
&\left[\left[ 
{\scriptsize\begin{array}{c}
\pi _{D(1)}(s)\mu (1,1,s)/\pi _{A}(s)+ \\ 
(1-\pi _{D(0) }(s))\mu (0,0,s)/(1-\pi _{A}(s) \\ 
+(1-\pi _{D(1)}(s))Y_{L}/\pi _{A}(s) \\ 
+\pi _{D(0) }(s)Y_{H}/(1-\pi _{A}(s))
\end{array}}
\right]\times  \left[ 
{\scriptsize\begin{array}{c}
\pi _{D(1)}(s)\mu (1,1,s)/\pi _{A}(s)+ \\ 
(1-\pi _{D(0) }(s))\mu (0,0,s)/(1-\pi _{A}(s) \\ 
+(1-\pi _{D(1)}(s))Y_{H}/\pi _{A}(s) \\ 
+\pi _{D(0) }(s)Y_{L}/(1-\pi _{A}(s))
\end{array}}
\right] \right],
\end{align*}
and $\Delta _{H}^{2}(\mathbf{P})$ is as $\Delta _{L}^{2}(\mathbf{P})$ but with $(Y_{L},Y_{H})$ replaced by $(Y_{H},Y_{L})$.
\end{theorem}
\begin{proof}
This proof is analogous to that of Theorem \ref{thm:AsyDist_ATE}, and therefore omitted. The main difference in the argument is that Lemma \ref{lem:AsyDist} is replaced by Lemma \ref{lem:AsyDist_ext} (which additionally requires Assumption \ref{ass:2}(f)), explaining why we use $\mathcal{P}_3$ instead of $\mathcal{P}_1$.
\end{proof}

\begin{theorem}\label{thm:AsyDist_ATT_other} 
Let $( \tilde{\upsilon}_{L},\tilde{\upsilon}_{H}) $ denote the alternative bounds estimator defined as in \eqref{eq:BoundsHat_ATT} but with $\{ \pi _{A}(S_{i}) :i=1,\ldots ,n\} $ replaced by $\{ n_{A}(S_{i})/n(S_i) :i=1,\ldots ,n\} $. Then, 
\begin{equation}
\sqrt{n}\left( 
{\small\begin{array}{c}
\tilde{\upsilon}_{L}-\upsilon _{L}(\mathbf{P}) \\ 
\tilde{\upsilon}_{H}-\upsilon _{H}(\mathbf{P})
\end{array}}
\right) ~\overset{d}{\to}~N\left( {\bf 0}_2 ,
{\scriptsize
\begin{array}{c}
\frac{1}{G^2}\left[\left( 
{\scriptsize \begin{array}{cc}
\varpi _{L}^{2}(\mathbf{P})  & \varpi _{HL}\left( \mathbf{P}\right)  \\ 
\varpi _{HL}(\mathbf{P})  & \varpi _{H}^{2}\left( \mathbf{P}\right) 
\end{array}}
\right) +\left( 
{\scriptsize \begin{array}{cc}
\Delta _{L}^{2}(\mathbf{P}) & \Delta _{HL}\left( \mathbf{P}\right)  \\ 
\Delta _{HL}(\mathbf{P}) & \Delta _{H}^{2}\left( \mathbf{P}\right) 
\end{array}}
\right) \right]
\end{array}}
\right) ,  \label{eq:AsyDist_ATT1}
\end{equation}
uniformly in $\mathcal{P}_4$, where $(\varpi _{L}^{2}(\mathbf{P}),\varpi _{H}^{2}(\mathbf{P}),\varpi _{HL}(\mathbf{P}),N_L,N_H,G)$ are as in Theorem \ref{thm:AsyDist_ATT},
\begin{align*}
\Delta _{L}^{2}(\mathbf{P}) &~=~\sum_{s\in \mathcal{S}}p(s)\tau (s) ( 1-\pi _{A}(s) ) \pi _{A}(s) 
\left[ 
{\scriptsize\begin{array}{c}
\mu \left( 1,1,s\right) \pi _{D(1)}(s) +\mu \left( 0,1,s\right) \left( 1-\pi _{D(1)}(s) \right)  \\ 
-\mu \left( 1,0,s\right) \pi _{D(0) }(s) -\mu \left( 0,0,s\right) \left( 1-\pi _{D(0) }(s) \right)  \\ 
-(N_{L}/G)\left( \pi _{D(1)}(s) -\pi _{D(0) }(s) \right) 
\end{array}}
\right] ^{2}  \notag \\
\Delta _{HL}(\mathbf{P})& ~=~\sum_{s\in \mathcal{S}}p(s)\tau (s) ( 1-\pi _{A}(s) ) \pi _{A}(s) \times \\
&\left[\left[ 
{\scriptsize\begin{array}{c}
\mu \left( 1,1,s\right) \pi _{D(1)}(s) \\
+\mu \left( 0,1,s\right)
\left( 1-\pi _{D(1)}(s) \right)  \\ 
-\mu \left( 1,0,s\right) \pi _{D(0) }(s) \\
-\mu \left(
0,0,s\right) \left( 1-\pi _{D(0) }(s) \right)  \\ 
-(N_{L}/G)\left( \pi _{D(1)}(s) -\pi _{D(0) }\left(s\right) \right) 
\end{array}}
\right] \left[ 
{\scriptsize\begin{array}{c}
\mu \left( 1,1,s\right) \pi _{D(1)}(s)\\ +\mu \left( 0,1,s\right)\left( 1-\pi _{D(1)}(s) \right)  \\ 
-\mu \left( 1,0,s\right) \pi _{D(0) }(s)\\ 
-\mu \left(0,0,s\right) \left( 1-\pi _{D(0) }(s) \right)  \\ 
-(N_{H}/G)\left( \pi _{D(1)}(s) -\pi _{D(0) }\left(s\right) \right) 
\end{array}}
\right] \right],
\end{align*}
and $\Delta _{H}^{2}(\mathbf{P}) $ is as  $\Delta _{L}^{2}(\mathbf{P})$ but with $N_{L}$ replaced by $N_{H}$.
\end{theorem}
\begin{proof}
This proof is analogous to that of Theorem \ref{thm:AsyDist_ATT}, and therefore omitted. The main difference in the argument is that Lemma \ref{lem:AsyDist} is replaced by Lemma \ref{lem:AsyDist_ext} (which additionally requires Assumption \ref{ass:2}(f)), explaining why we replace $\mathcal{P}_2$ with $\mathcal{P}_4$.
\end{proof}


\section{Appendix on the ATU}\label{sec:appendixE}

\subsection{Identification}

We now provide the sharp identified set for the ATU $\phi( {\bf Q},{\bf G})\equiv E[Y(1) - Y(0)|D=0]$.

\begin{theorem}\label{thm:mainATU} 
Under Assumptions \ref{ass:1} and \ref{ass:2}(a)-(e), the identified set for the ATU is $\Phi _{I}({\bf P}) = [\phi_{L}({\bf P}),\phi_{H}({\bf P})]$, where 
{\small
\begin{align} 
\phi _{L}( {\bf P}) 
& ~\equiv~ \frac{1}{G_U} E\Bigg[  \left( \frac{Y_i A_i}{P(A_i =1 | S_i)} - \frac{Y_i(1-A_i)}{P(A_i = 0 | S_i)} \right) P(A_i =0 | S_i)+  \frac{(Y_L-Y_i)(1-D_i)A_i}{P(A_i = 0 | S_i)}  \Bigg],  \notag \\ 
\phi _{H}( {\bf P})
& ~\equiv ~\frac{1}{G_U} E\Bigg[  \left( \frac{Y_i A_i}{P(A_i =1 | S_i)} - \frac{Y_i(1-A_i)}{P(A_i = 0| S_i)} \right) P(A_i =0 | S_i) +  \frac{(Y_H-Y_i)(1-D_i)A_i}{P(A_i = 0 | S_i)}  \Bigg],  \label{eq:bounds_ATU}
\end{align}
}
and
{\small
\begin{align} 
G_U = E \left[ \left( \frac{D_i A_i}{ P(A_i =1 | S_i) } - \frac{D_i (1-A_i)}{P(A_i = 0 | S_i) } \right)  P(A_i =0 | S_i)  + \frac{(1-D_i)A_i}{P(A_i = 0 | S_i)}  \right]. 
\label{eq:bounds_ATU_G}
\end{align}
}
\end{theorem}
\begin{proof}
This proof is analogous to that of Theorem \ref{thm:mainATT}, and it is therefore omitted. 
\end{proof}

\subsection{Inference}

By analogy with the results for the ATT, we propose to estimate the identified of the ATU with $\hat\Phi _{I} ~\equiv~[\hat{\phi}_{L},\hat{\phi}_{U}]$,  where
\begin{align}
\hat{\phi}_{L} ~&\equiv~ \frac{1}{n \hat{G}_U} \sum_{i=1}^n \left[ \left( \tfrac{Y_{i}A_{i}}{n_{A}\left( S_{i}\right) /n\left( S_{i}\right)} - \tfrac{Y_{i} \left(1-A_{i} \right)}{1- n_{A}\left( S_{i}\right) /n\left( S_{i}\right)}  \right) (1-{\pi}_{A}\left( S_{i} \right))  + \tfrac{\left(Y_{L}-Y_{i}\right)(1-D_{i})A_{i}}{n_{A}\left( S_{i}\right) /n\left( S_{i}\right)} \right]     \notag\\
\hat{\phi}_{H} ~&\equiv~ \frac{1}{n \hat{G}_U} \sum_{i=1}^n \left[ \left( \tfrac{Y_{i}A_{i}}{n_{A}\left( S_{i}\right) /n\left( S_{i}\right)} - \tfrac{Y_{i} \left(1-A_{i} \right)}{1- n_{A}\left( S_{i}\right) /n\left( S_{i}\right)}  \right) (1-{\pi}_{A}\left( S_{i} \right))  + \tfrac{\left(Y_{H}-Y_{i}\right)(1-D_{i})A_{i}}{n_{A}\left( S_{i}\right) /n\left( S_{i}\right)} \right]    
    \label{eq:BoundsHat_ATU}  
\end{align}
and
\begin{equation}
    \hat{G}_U~\equiv~ \frac{1}{n} \sum_{i=1}^n \left[ \left(\tfrac{D_{i}A_{i}}{n_{A}\left( S_{i}\right) /n\left( S_{i}\right)} - \tfrac{D_{i}\left(1-A_{i}\right)}{1-n_{A}\left( S_{i}\right) /n\left( S_{i}\right)}\right)(1- {\pi}_{A}\left( S_{i} \right)) + \tfrac{(1-D_{i})A_{i}}{n_{A}\left( S_{i}\right) /n\left( S_{i}\right)} \right].
   \label{eq:G_hatATU}
\end{equation}
The asymptotic properties of these bounds are analogous to those of the ATT bounds. Recall from that  $\mathcal{P}_2$ is the set of probabilities $\mathbf{P}$ generated by $(\mathbf{Q},\mathbf{G})$ that satisfy Assumptions \ref{ass:1} and \ref{ass:2}(a)-(e). The next result establishes that the estimator of the identified set in \eqref{eq:BoundsHat_ATU} is uniformly consistent.
\begin{theorem}\label{thm:consistATU}
For any $\varepsilon >0$, 
\begin{equation*}
\underset{n\to \infty }{\lim \inf }~\inf_{{\bf P} \in \mathcal{P}_2 }~{\bf P}(~d_{H}(\hat\Phi _{I},\Phi _{I}({\bf P}))\leq \varepsilon ~)~=~1,  \label{eq:consistencyATU}
\end{equation*}
where $d_H(A,B)$ denotes the Hausdorff distance between sets $A, B\subset \mathbb{R}$.
\end{theorem}
\begin{proof}
This proof is analogous to that of Theorem \ref{thm:consist}, and it is therefore omitted. 
\end{proof}

We propose the following confidence interval (CI) for the ATU:
\begin{equation}
\hat{C}_{\Phi}\left( 1-\alpha \right) 
~\equiv~[~ \hat{\phi}_{L}-{\hat{\vartheta}_{L} \hat{c}_{L}}/{\sqrt{n}}~,~\hat{\phi}_{H}+{\hat{\vartheta}_{H} \hat{c}_{H}}/{\sqrt{n}}~] ,
\label{eq:CSstoyeATU}
\end{equation}
where $( \hat{c}_{L},\hat{c}_{H}) $ are the minimizers of $\hat{\vartheta}_{L}{c}_{L}+\hat{\vartheta}_{H}{c}_{H}$ subject to the following constraints
\begin{align*}
P\bigg( -{c}_{L}\leq Z_{1}~\cap~\tfrac{\hat{\vartheta}_{HL}}{\hat{\vartheta}_{H}\hat{\vartheta}_{L}}Z_{1}\leq c_{H}+\tfrac{\sqrt{n}( \hat{\phi}_{H}-\hat{\phi}_{L}) }{\hat{\vartheta}_{H}}+Z_{2}\sqrt{1-\tfrac{\hat{\vartheta}_{HL}^2}{\hat{\vartheta}_{H}^2\hat{\vartheta}_{L}^2}} ~\bigg|~ W^{(n)} \bigg)  
&\geq 1-\alpha  \notag\\
P\bigg( -c_{L}-\tfrac{\sqrt{n}( \hat{\phi}_{H}-\hat{\phi}_{L}) }{\hat{\vartheta}_{L}}-Z_{2}\sqrt{1-\tfrac{\hat{\vartheta}_{HL}^2}{\hat{\vartheta}_{H}^2\hat{\vartheta}_{L}^2}}
\leq \tfrac{\hat{\vartheta}_{HL}}{\hat{\vartheta}_{H}\hat{\vartheta}_{L}}Z_{1}~\cap~Z_{1}\leq c_{H} ~\bigg|~ W^{(n)} \bigg)  
&\geq 1-\alpha ,
\end{align*}%
and $Z_{1},Z_{2}$ are i.i.d.\ $N( 0,1) $. 
We now establish its asymptotic uniform validity.

\begin{theorem}\label{thm:StoyeATU}
The CI in \eqref{eq:CSstoyeATU} satisfies
\begin{equation}
    \underset{n\to \infty }{\lim}~\inf_{{\bf P}\in \mathcal{P}_2 }~\inf_{\phi \in \Theta_{I}({\bf P})}~{\bf P}( \phi \in \hat{C}_{\Phi}( 1-\alpha )) ~=~ 1-\alpha .\label{eq:coverageATU}
\end{equation}
\end{theorem}
\begin{proof}
This proof is analogous to that of Theorem \ref{thm:Stoye}, and it is therefore omitted. 
\end{proof}

\subsection{Auxiliary results}

\begin{lemma}\label{lem:expressionBoundsATU} 
Let
\begin{align}
    \check{\phi}_{L} ~&\equiv~ \frac{1}{n \check{G}_U} \sum_{i=1}^n \left[ \left( \tfrac{Y_{i}A_{i}}{n_{A}\left( S_{i}\right) /n\left( S_{i}\right)} - \tfrac{Y_{i} \left(1-A_{i} \right)}{1- n_{A}\left( S_{i}\right) /n\left( S_{i}\right)}  \right)(1- \hat{\pi}_{A}\left( S_{i} \right)) 
    + \tfrac{\left(Y_{L}-Y_{i}\right)(1-D_{i})A_{i} }{1-n_{A}\left( S_{i}\right) /n\left( S_{i}\right)} \right]     \notag\\
    \check{\phi}_{H} ~&\equiv~ \frac{1}{n \check{G}_U} \sum_{i=1}^n \left[ \left( \tfrac{Y_{i}A_{i}}{n_{A}\left( S_{i}\right) /n\left( S_{i}\right)} - \tfrac{Y_{i} \left(1-A_{i} \right)}{1- n_{A}\left( S_{i}\right) /n\left( S_{i}\right)}  \right)(1- \hat{\pi}_{A}\left( S_{i} \right))
    + \tfrac{\left(Y_{H}-Y_{i}\right)(1-D_{i})A_{i}}{1-n_{A}\left( S_{i}\right) /n\left( S_{i}\right)} \right],
    \label{eq:expressionBoundsATU1}
\end{align}
where 
\begin{equation}
    \check{G}_U~\equiv~ \frac{1}{n} \sum_{i=1}^n \left[ \left(\tfrac{D_{i}A_{i}}{n_{A}\left( S_{i}\right) /n\left( S_{i}\right)} - \tfrac{D_{i}\left(1-A_{i}\right)}{1-n_{A}\left( S_{i}\right) /n\left( S_{i}\right)}\right) (1-\hat{\pi}_{A}\left( S_{i} \right)) + \tfrac{(1-D_{i})A_{i}}{1-n_{A}\left( S_{i}\right) /n\left( S_{i}\right)} \right].
   \label{eq:expressionBoundsATU2}
\end{equation}
and $\{\hat{\pi}_A(s):s \in \mathcal{S}\} $ is an arbitrary estimator of $\{{\pi}_A(s):s \in \mathcal{S}\} $. Then,
\begin{equation}
\sqrt{n}\left( 
\begin{array}{c}
\check{\phi}_{L}-\phi _{L}(\mathbf{P}) \\ 
\check{\phi}_{H}-\phi _{H}(\mathbf{P})
\end{array}
\right) ~=~\frac{1}{G_U}\left( 
\begin{array}{c}
U_{n}+V_{L,n}+W_{L,n} \\ 
U_{n}+V_{H,n}+W_{H,n}
\end{array}
\right) +\delta _{n},  \label{eq:expressionBoundsATU3}
\end{equation}
where $\phi_{L}(\mathbf{P})$ and $\phi _{H}(\mathbf{P})$ are as in \eqref{eq:bounds_ATU}, $\delta _{n}=o_{p}(1)$ uniformly in $\mathbf{P}\in \mathcal{P}_2$, and
\begin{align}
U_{n}& ~=~\sum_{s\in \mathcal{S}}\left[
{\scriptsize
\begin{array}{c}
 (1-\pi _{A}( s))T_{n,1}( 1,1,s)/\pi _{A}( s)  - T_{n,1}( 0,1,s) - T_{n,1}( 1,0,s) - T_{n,1}( 0,0,s) \\
+ p( s)  [ (1-\pi _{A}( s)) \mu ( 1,1,s) + \pi _{A}( s) \mu ( 0,1,s) ] T_{n,2}( 1,s) \\
+ p( s) ( 1-\pi _{A}( s) ) [  - \mu ( 1,0,s) + \mu ( 0,0,s) ] T_{n,2}(2,s) \\
\left[ 
\begin{array}{c}
\left[ 
\begin{array}{c}
\mu ( 1,1,s) \pi _{D(1)}( s) 
+ \mu ( 0,1,s) ( 1-\pi _{D(1)}(s) )   \\ 
- \mu (1,0,s) \pi _{D( 0) }( s)
-\mu ( 0,0,s) ( 1-\pi _{D( 0) }( s)) 
\end{array}
\right] (1-\pi _{A}( s))  \\ 
- \mu ( 0,1,s) (1-\pi _{D( 1) }( s)) 
\end{array}
\right] T_{n,3}( s) 
\end{array}}
\right]  \notag \\
V_{L,n}& ~=~ \sum_{s\in \mathcal{S}}\left[
{\scriptsize
\begin{array}{c}
p( s) [ \frac{N_L}{G_U} \pi _{A}( s) - Y_L ] T_{n,2}(1,s)+p( s) \frac{N_L}{G_U}( 1-\pi _{A}( s) ) T_{n,2}(2,s) + \\
\big[  Y_{L}(1-\pi _{D( 1) }( s)) - \frac{N_L}{G_U}(
(1-\pi _{D(1)}( s)) \pi _{A}( s) +
(1-\pi _{D( 0) }( s)) ( 1-\pi _{A}( s) ) 
)\big] T_{n,3}( s)
\end{array}}
\right]   \notag \\
W_{L,n}&  ~=~\sum_{s\in \mathcal{S}}p( s) \sqrt{n}( \hat{\pi}_{A}( s) -\pi _{A}( s) ) \left[ 
{\scriptsize\begin{array}{c}
- \mu ( 1,1,s) \pi _{D(1)}( s) - \mu (0,1,s) ( 1-\pi _{D(1)}( s) )  \\ 
+ \mu ( 1,0,s) \pi _{D( 0) }( s) + \mu (0,0,s) ( 1-\pi _{D( 0) }( s) )  \\ 
+ \frac{N_L}{G_U}( \pi _{D(1)}( s) - \pi _{D( 0) }( s) ) 
\end{array}}
\right]   \notag \\
N_{L}&  ~=~\sum_{s\in \mathcal{S}}p( s)\left[ 
{\scriptsize\begin{array}{c}
\left[ 
\begin{array}{c}
\mu (1,1,s)\pi _{D(1)}(s)
+\mu (0,1,s)( 1-\pi _{D(1)}(s))  \\ 
-\mu (1,0,s)\pi _{D(0)}(s)
-\mu (0,0,s)( 1-\pi _{D(0)}(s)) 
\end{array}
\right] (1-\pi _{A}( s)) \\
+(Y_{L}-\mu (0,1,s))
(1- \pi_{D(1) }( s) )
\end{array}} \right] \notag \\
N_{H}&   ~=~\sum_{s\in \mathcal{S}}p( s)\left[ 
{\scriptsize\begin{array}{c}
\left[ 
\begin{array}{c}
\mu (1,1,s)\pi _{D(1)}(s)+\mu (0,1,s)( 1-\pi _{D(1)}(s))  \\ 
-\mu (1,0,s)\pi _{D(0)}(s)-\mu (0,0,s)( 1-\pi _{D(0)}(s)) 
\end{array}
\right] (1-\pi _{A}( s)) \\
+(Y_{H} - \mu ( 0,1,s))
(1- \pi_{D(1) }( s) )
\end{array}}
\right] \notag \\
G_U&  ~=~\sum_{s\in \mathcal{S}}p( s) [ (1-\pi _{D( 1)}( s)) \pi _{A}( s) + (1-\pi _{D( 0) }(s)) (1-\pi _{A}( s))] ,\label{eq:expressionBoundsATU4}
\end{align}
with $(T_{n,1},T_{n,2},T_{n,3})$ as in \eqref{eq:Tn_defn}, and $V_{H,n}$ and $W_{H,n}$ are defined as in $V_{L,n}$ and $W_{L,n}$ but with $(N_L,Y_{L}) $ replaced by $( N_L,Y_{L}) $.
\end{lemma}
\begin{proof}
This proof is analogous to that of Lemma \ref{lem:expressionBoundsATT}, and it is therefore omitted. 
\end{proof}

\begin{theorem}\label{thm:AsyDist_ATU} 
We have 
\begin{equation}
\sqrt{n}\left( 
\begin{array}{c}
\hat{\phi}_{L}-\phi _{L}(\mathbf{P}) \\ 
\hat{\phi}_{H}-\phi _{H}(\mathbf{P})
\end{array}
\right) ~\overset{d}{\to  }~
 N\left( {\bf 0 }_2,
\frac{1}{G_U^{2}} \left( 
\begin{array}{cc}
\vartheta _{L}^{2}(\mathbf{P})  & \vartheta _{HL}\left( \mathbf{P}%
\right)  \\ 
\vartheta _{HL}(\mathbf{P})  & \vartheta _{H}^{2}\left( \mathbf{P}%
\right) 
\end{array}%
\right) \right) ,  \label{eq:AsyDist_ATU}
\end{equation}%
uniformly in $\mathcal{P}_2$, where 
\begin{align*}
\vartheta _{L}^{2}(\mathbf{P})  
&~=~\sum_{s\in \mathcal{S}}p(s)\left[
{\scriptsize
\begin{array}{c}
(1-\pi _{A}(s))^2 \pi _{D(1)}(s) \sigma ^{2}(1,1,s) / \pi _{A}(s)
+\left( 1-\pi _{D(1)}(s)\right) \pi _{A}(s)\sigma ^{2}\left( 0,1,s\right) \\ 
+ \pi _{D(0)}(s)(1-\pi _{A}(s))\sigma ^{2}\left( 1,0,s\right)
+ (1-\pi _{D(0)}(s)) (1-\pi _{A}(s)) \sigma ^{2}\left( 0,0,s\right) \\ 
+\left[ 
( 1-\pi _{A}(s) ) 
\mu \left( 1,1,s\right) 
+ \pi_{A}(s) \mu \left( 0,1,s\right)  + \frac{N_L}{G_U} \pi _{A}(s) - Y_{L}
\right] ^{2}
\frac{(1-\pi _{D(1)}(s))\pi _{D(1)}(s)}{\pi _{A}(s)} \\ 
+ ( 1 - \pi _{A}(s)) 
\left[ \mu \left(
0,0,s\right) + \frac{N_L}{G_U} - \mu \left( 1,0,s\right)\right] ^{2}
(1-\pi _{D(0)}(s))\pi _{D(0)}(s) \\ 
+\left[ 
\begin{array}{c}
\left[ 
\begin{array}{c}
\mu \left( 1,1,s\right) \pi _{D(1)}(s) +\mu \left( 0,1,s\right)
\left( 1-\pi _{D(1)}(s) \right)  \\ 
-\mu \left( 1,0,s\right) \pi _{D(0) }(s) -\mu \left(
0,0,s\right) \left( 1-\pi _{D(0) }(s) \right) 
\end{array}%
\right] (1-\pi _{A}(s))  \\ 
+ (Y_L - \mu \left( 0,1,s\right) ) (1 - \pi _{D\left( 1 \right) }(s))  \\ 
- \left[ \frac{N_L}{G_U}\left[ (1-\pi _{D(1)}(s)) \pi _{A}(s)
+ (1-\pi _{D(0) }(s) )\left( 1-\pi _{A}(s)
\right) \right] \right] 
\end{array}%
\right] ^{2} 
\end{array}}
\right] 
\end{align*}
and 
\begin{align*}
\vartheta _{HL}(\mathbf{P}) &~=~\sum_{s\in \mathcal{S}}p(s)\left[ 
{\scriptsize
\begin{array}{c}
( 1- \pi_{A}(s))^2 \pi _{D(1)}(s) 
\sigma ^{2}\left( 1,1,s\right) / \pi _{A}(s) \\
+\left( 1-\pi _{D(1)}(s)\right)
\pi _{A}(s)\sigma ^{2}\left( 0,1,s\right) 
+\pi _{D(0)}(s)(1-\pi _{A}(s))\sigma ^{2}\left( 1,0,s\right)\\
+ (1 - \pi _{D(0)}(s))( 1- \pi _{A}(s))\sigma ^{2}(0,0,s) \\ 
+ \left[ 
( 1-\pi _{A}(s) ) \mu \left( 1,1,s\right) +
\pi_{A}(s) \mu \left( 0,1,s\right)  
+ \frac{N_L}{G_U} \pi _{A}(s)- Y_L 
\right] 
\times  \\ 
\left[ 
\begin{array}{c}
( 1-\pi _{A}(s) ) 
\mu \left( 1,1,s\right) 
+\pi_{A}(s) \mu \left( 0,1,s\right)  \\ 
+ \frac{N_H}{G_U}\pi _{A}(s) - Y_H
\end{array}%
\right] \frac{(1-\pi _{D(1)}(s))\pi _{D(1)}(s)}{\pi _{A}(s)} \\ 
+ (1-\pi _{A}(s)) 
\left[ 
\begin{array}{c}
\left[ - \mu \left( 1,0,s\right) + \mu \left( 0,0,s\right) +\frac{N_L}{G_U}\right]
\times  \\ 
\left[ - \mu \left( 1,0,s\right) + \mu \left( 0,0,s\right) +\frac{N_H}{G_U}\right] 
\end{array}%
\right] (1-\pi _{D(0)}(s))\pi _{D(0)}(s) \\ 
+\left[ 
\begin{array}{c}
\left[ 
\begin{array}{c}
\mu \left( 1,1,s\right) \pi _{D(1)}(s) +\mu \left( 0,1,s\right)
\left( 1-\pi _{D(1)}(s) \right)  \\ 
-\mu \left( 1,0,s\right) \pi _{D(0) }(s) 
-\mu \left(0,0,s\right) \left( 1-\pi _{D(0) }(s) \right) 
\end{array}%
\right] (1-\pi _{A}(s))  \\ 
+(Y_L - \mu \left( 0,1,s\right) ) ( 1 - \pi _{D(1) }(s))  \\ 
-\left[ \frac{N_L}{G_U}\left[ (1-\pi _{D(1)}(s)) \pi _{A}(s)
+(1-\pi _{D(0) }(s)) \left( 1-\pi _{A}(s)
\right) \right] \right] 
\end{array}%
\right] \times  \\ 
\left[ 
\begin{array}{c}
\left[ 
\begin{array}{c}
\mu \left( 1,1,s\right) \pi _{D(1)}(s) 
+\mu \left( 0,1,s\right)
\left( 1-\pi _{D(1)}(s) \right)  \\ 
-\mu \left( 1,0,s\right) \pi _{D(0) }(s) 
-\mu \left(0,0,s\right) \left( 1-\pi _{D(0) }(s) \right) 
\end{array}%
\right] (1 - \pi _{A}(s))  \\ 
+ (Y_{H} - \mu \left( 0, 1,s\right))( 1 - \pi _{D(1) }(s))  \\ 
- \left[ \frac{N_H}{G_U}\left[ (1 - \pi _{D(1)}(s)) \pi _{A}(s)
+ (1 - \pi _{D(0) }(s) )\left( 1-\pi _{A}(s)
\right) \right] \right] 
\end{array}%
\right]
\end{array}
}
\right] ,
\end{align*}
where $( N_{L},N_{H},G_U) $ is as in \eqref{eq:expressionBoundsATU4}, and $\vartheta _{H}^{2}(\mathbf{P}) $ is as $\vartheta _{L}^{2}(\mathbf{P}) $ but with $(N_{L},Y_{L})$ replaced by $(N_{H},Y_{H})$. 
Moreover, $\vartheta _{L}^{2}(\mathbf{P}) $ and $\vartheta _{H}^{2}(\mathbf{P}) $ are positive and finite, uniformly in $\mathcal{P}_2$.
\end{theorem}
\begin{proof}
This proof is analogous to that of Theorem \ref{thm:AsyDist_ATT}, and it is therefore omitted.
\end{proof}

\begin{theorem}\label{thm:AsyVarEstimation_ATU} 
Let $(\hat{p}(s),\hat{\pi}_{D(a)}(s),\hat{\mu}(d,a,s),\hat{\sigma} ^{2}(d,a,s)) : (d,a,s)\in \{0,1\}\times \{0,1\}\times \mathcal{S}) $ be as in \eqref{eq:AsyDistATE_var8}, and let
\begin{align}
\hat{\vartheta}_{L}^{2} &~=~\sum_{s\in \mathcal{S}}\hat{p}(s)\left[ 
{\scriptsize
\begin{array}{c}
(1 - {\pi}_{A}(s))^{2} \hat{\pi}_{D(1)}(s) \hat{\sigma}^{2}\left(
1,1,s\right) / {\pi}_{A}(s) \\ 
+\left( 1-\hat{\pi}_{D(1)}(s)\right) {\pi}_{A}(s)\hat{\sigma}^{2}\left( 0,1,s\right) + 
+\hat{\pi}_{D(0)}(s)(1-{\pi}_{A}(s))\hat{\sigma}^{2}\left( 1,0,s\right) 
\\ 
+ (1-\hat{\pi}_{D(0)}(s))
(1 - {\pi}_{A}(s))\hat{\sigma}^{2}(0,0,s) \\
+\left[ 
\left( 1-{\pi}_{A}(s) \right) \hat{\mu}\left( 1,1,s\right) +%
{\pi}_{A}(s) \hat{\mu}\left( 0,1,s\right)  
+ \hat{\phi}_{L} {\pi}_{A}(s) - Y_L
\right] ^{2}
\frac{(1-\hat{\pi}_{D(1)}(s))\hat{\pi}_{D(1)}(s)}{\hat{\pi}%
_{A}(s)} \\ 
+(1-{\pi}_{A}(s))\left[ - \hat{\mu}\left( 1,0,s\right) + \hat{\mu}\left(0,0,s\right) + \hat{\phi}_{L}\right] ^{2}
(1-\hat{\pi}_{D(0)}(s))\hat{\pi}_{D(0)}(s) \\ 
+\left[ 
\begin{array}{c}
\left[ 
\begin{array}{c}
\hat{\mu}\left( 1,1,s\right) \hat{\pi}_{D(1)}(s) +\hat{\mu}%
\left( 0,1,s\right) \left( 1-\hat{\pi}_{D(1)}(s) \right)  \\ 
-\hat{\mu}\left( 1,0,s\right) \hat{\pi}_{D(0) }(s) -%
\hat{\mu}\left( 0,0,s\right) \left( 1-\hat{\pi}_{D(0) }\left(
s\right) \right) 
\end{array}%
\right] (1-{\pi}_{A}(s))  \\ 
+ ( Y_L - \hat{\mu}\left( 0,1,s\right) )
(1 - \hat{\pi}_{D(1) }(s))  \\ 
-\left[ \hat{\phi}_{L}\left[ (1-\hat{\pi}_{D(1)}(s)) \hat{\pi}%
_{A}(s) + ( 1 - \hat{\pi}_{D(0) }(s)) 
\left( 1- {\pi}_{A}(s) \right) \right] \right] 
\end{array}%
\right] ^{2} 
\end{array}}
\right] , \notag\\
\hat{\vartheta}_{HL}&~=~\sum_{s\in \mathcal{S}}\hat{p}(s)\left[ 
{\scriptsize
\begin{array}{c}
(1 - {\pi}_{A}(s))^{2}\hat{\pi}_{D(1)}(s)
\hat{\sigma}^{2}\left(1,1,s\right) / {\pi}_{A}(s) \\ 
+ \left( 1-\hat{\pi}_{D(1)}(s)\right) {\pi}_{A}(s)\hat{\sigma}^{2}\left( 0,1,s\right) 
+\hat{\pi}_{D(0)}(s)(1-{\pi}_{A}(s))\hat{\sigma}^{2}\left( 1,0,s\right) 
\\ 
+ (1 - \hat{\pi}_{D(0)}(s))(1 - {\pi}_{A}(s))
\hat{\sigma}^{2}(0,0,s) \\
+\left[ 
\left( 1-{\pi}_{A}(s) \right) \hat{\mu}\left( 1,1,s\right) 
+ {\pi}_{A}(s) \hat{\mu}\left( 0,1,s\right)  \hat{\phi}_{L} {\pi}_{A}(s) - Y_L
\right] \times  \\ 
\left[ 
\begin{array}{c}
\left( 1-{\pi}_{A}(s) \right) \hat{\mu}\left( 1,1,s\right) 
+ {\pi}_{A}(s) \hat{\mu}\left( 0,1,s\right)  \\ 
+ \hat{\phi}_{H} {\pi}_{A}(s) - Y_H
\end{array}%
\right] 
\frac{(1-\hat{\pi}_{D(1)}(s))\hat{\pi}_{D(1)}(s)}{{\pi}_{A}(s)} \\ 
+ (1 - {\pi}_{A}(s))
\left[ 
\begin{array}{c}
\left[ - \hat{\mu}\left( 1,0,s\right) + \hat{\mu}\left( 0,0,s\right) 
+ \hat{\phi}_{L}\right]  \\ 
\times \left[ - \hat{\mu}\left( 1,0,s\right) + \hat{\mu}\left( 0,0,s\right) 
+ \hat{\phi}_{H}\right] 
\end{array}%
\right] (1-\hat{\pi}_{D(0)}(s))\hat{\pi}_{D(0)}(s) \\ 
+\left[ 
\begin{array}{c}
\left[ 
\begin{array}{c}
\hat{\mu}\left( 1,1,s\right) \hat{\pi}_{D(1)}(s)
+\hat{\mu}\left( 0,1,s\right) \left( 1-\hat{\pi}_{D(1)}(s) \right)  \\ 
-\hat{\mu}\left( 1,0,s\right) \hat{\pi}_{D(0) }(s) 
- \hat{\mu}\left( 0,0,s\right) \left( 1-\hat{\pi}_{D(0) }\left(
s\right) \right) 
\end{array}%
\right] (1 - {\pi}_{A}(s))  \\ 
+ (Y_L - \hat{\mu}\left( 0, 1,s\right))(1 - \hat{\pi}_{D(1) })\left(s\right)  \\ 
- \left[ \hat{\phi}_{L}\left[ (1 - \hat{\pi}_{D(1)}(s)) {\pi}_{A}(s) + (1 - \hat{\pi}_{D(0) }(s)) \left( 1-{\pi}_{A}(s) \right) \right] \right] 
\end{array}%
\right]  \\ 
\times 
\left[ 
\begin{array}{c}
\left[ 
\begin{array}{c}
\hat{\mu}\left( 1,1,s\right) \hat{\pi}_{D(1)}(s) +\hat{\mu}\left( 0,1,s\right) \left( 1-\hat{\pi}_{D(1)}(s) \right)  \\ 
-\hat{\mu}\left( 1,0,s\right) \hat{\pi}_{D(0) }(s) 
- \hat{\mu}\left( 0,0,s\right) \left( 1-\hat{\pi}_{D(0) }(s) \right) 
\end{array}%
\right] (1 - {\pi}_{A}(s))  \\ 
+ ( Y_H - \hat{\mu}\left( 0, 1,s\right) ) ( 1 - \hat{\pi}_{D\left( 1 \right)) }(s)  \\ 
- \left[ \hat{\phi}_{H}\left[ (1 - \hat{\pi}_{D(1)}(s)) {\pi}_{A}(s) + (1 - \hat{\pi}_{D(0) }(s)) \left( 1-{\pi}_{A}(s) \right) \right] \right] 
\end{array}%
\right] \\ 
+{\pi}_{A}(s) \left( 1-{\pi}_{A}(s) \right)
\tau (s) \times  
\end{array}}
\right] , \label{eq:AsyVar_ATU}
\end{align}
and $\hat{\vartheta}_{H}^{2}$ is defined as in $\hat{\vartheta}_{L}^{2}$ but with $(\hat{\phi}_{L},Y_{L})$ replaced by $(\hat{\phi}_{H},Y_{H})$. Then, 
\begin{equation*}
(\hat{\vartheta}_{L}^{2},\hat{\vartheta}_{H}^{2},\hat{\vartheta}_{HL})/\hat{G}_U^2~\overset{p}{\to  }~(\vartheta _{L}^{2}(\mathbf{P)},\vartheta _{H}^{2}(\mathbf{P)},\vartheta _{HL}(\mathbf{P)})/G_U^2
\end{equation*}%
uniformly in $\mathcal{P}_2$.
\end{theorem}
\begin{proof}
This proof is analogous to that of Theorem \ref{thm:AsyVarEstimation}, and therefore omitted. 
\end{proof}